\documentclass[11pt,a4paper]{article}
\usepackage{latexsym,amssymb,amsmath,amsthm,amsfonts,enumerate,verbatim,xspace,
exscale}

\input xy
\xyoption{all} \CompileMatrices \UseComputerModernTips

\usepackage{graphicx,tabularx}
\usepackage{amsmath,amsbsy,amsfonts,amssymb}
\usepackage{color}
\usepackage{amsthm}

\title{The Jacobiator of nonholonomic systems and the geometry of reduced nonholonomic brackets}

\author{P. Balseiro  \footnotemark}
\author{{\sc{Paula Balseiro}\thanks{
         Universidade Federal Fluminense, Instituto de Matem\'atica, Rua Mario Santos Braga S/N, 24020-140,
        Niteroi, Rio de Janeiro, Brazil. \newline{\texttt{E-mail: pbalseiro@vm.uff.br}}}
        }  \ \
}

\theoremstyle{plain}
\newtheorem{theorem}{Theorem}[section]
\newtheorem{lemma}[theorem]{Lemma}
\newtheorem{proposition}[theorem]{Proposition}
\newtheorem{corollary}[theorem]{Corollary}
\newtheorem*{theorem*}{Theorem}
\newtheorem{remarkth}[theorem]{Remark}

\theoremstyle{definition}
\newtheorem{definition}[theorem]{Definition}
\newtheorem{example}[theorem]{Example}

\newenvironment{remark}{\begin{remarkth}\upshape}{\hfill$\diamond$\end{remarkth}}

\def\W{\mathcal{W}}
\def\M{\mathcal{M}}

\def\V{\mathcal{V}}
\def\S{\mathcal{S}}
\def\C{\mathcal{C}}
\def\Ham{\mathcal{H}}
\def\Lag{\mathcal{L}}
\def\R{\mathbb{R}}

\def\L{\mbox{Leg}}
\def\red{{\mbox{\tiny{red}}}}
\def\nh{{\mbox{\tiny{nh}}}}
\def\kin{{\mbox{\tiny{kin}}}}
\def\B{{\mbox{\tiny{$B$}}}}
\def\subW{{\mbox{\tiny{$\W$}}}}
\def\subC{{\mbox{\tiny{$\C$}}}}
\def\subS{{\mbox{\tiny{$\S$}}}}
\def\subM{{\mbox{\tiny{$\M$}}}}

\def\vecOm{\boldsymbol{\Omega}}

\def\a{\alpha}

\def\vecom{\boldsymbol{\omega}}
\newcommand{\SO}{\mbox{$\textup{SO}$}}

\def\vecep{\boldsymbol{\epsilon}}
\def\vecL{\boldsymbol{\lambda}}
\def\vecR{\boldsymbol{\rho}}
\def\vecgamma{\boldsymbol{\gamma}}
\def\vecalpha{\boldsymbol{\alpha}}
\def\vecbeta{\boldsymbol{\beta}}

\begin{document}
\maketitle

\begin{abstract}

In this paper, we consider the hamiltonian formulation of
nonholonomic systems with symmetries and study several aspects of
the geometry of their reduced almost Poisson brackets, including the
integrability of their characteristic distributions. Our starting
point is establishing global formulas for the nonholonomic
Jacobiators, before and after reduction, which are used to clarify
the relationship between reduced nonholonomic brackets and twisted
Poisson structures. For certain types of symmetries (generalizing
the  Chaplygin case), we obtain genuine Poisson structures on the
reduced spaces and analyze situations in which the reduced
nonholonomic brackets arise by applying a gauge transformation to these
Poisson structures. We illustrate our results with mechanical
examples, and in particular show how to recover several well-known facts in
the special case of Chaplygin symmetries.
\end{abstract}

\tableofcontents

\section{Introduction} \label{S:Intro}

Nonholonomic systems are mechanical systems with constraints in their
velocities that are not derived from constraints in the positions.
Examples include rolling contact constraints and special types of
sliding constraints, see e.g. \cite{BlochBook,CushmannBook,BKMM}. In
this paper we use the hamiltonian viewpoint to nonholonomic systems
to study several aspects of their geometry in the presence of
symmetries.

Nonholonomic systems are described by a configuration manifold $Q$,
a lagrangian $\Lag: TQ \to \R$, assumed to be of mechanical type,
and a non-integrable distribution $D \subset TQ$ (i.e., a non-involutive subbundle of $TQ$), defining the
permitted velocities of the system. These systems admit a
hamiltonian formulation (developed e.g. in \cite{BS93,IbLeMaMa1999,
MarsdenKoon}) in terms of triples $(\M,\pi_\nh, \Ham_\subM)$, where
$\M$ is the submanifold of $T^*Q$ given by the image of the
constraint distribution $D$ under the Legendre transform $\L: TQ\to
T^*Q$, $\Ham_\subM$ is the restriction to $\M$ of the hamiltonian
function on $T^*Q$ corresponding to $\Lag$, and $\pi_\nh$ is a
bivector field on $\M$ naturally defined from $D$ and the canonical
symplectic form $\Omega_Q$ on $T^*Q$. The associated {\it
nonholonomic bracket} (see
\cite{IbLeMaMa1999,Marle1998,SchaftMaschke1994}) $\{ \cdot , \cdot
\}_\nh$ is given by
$$
\{ f , g \}_\nh = \pi_\nh(df,dg),\;\;\; f,g \in C^\infty(\M),
$$
and the evolution of the nonholonomic system is given by the flow of
the {\it nonholonomic vector field}
$$
X_{\nh} =  \{\cdot , \Ham_\subM \}_\nh.
$$
An important feature that distinguishes nonholonomic systems from
usual hamiltonian systems is that the bracket $ \{\cdot , \cdot
\}_\nh$ does not satisfy the Jacobi identity, so it is just an {\it
almost Poisson bracket}. In fact, it can be shown that the
characteristic distribution\footnote{The distribution on $\M$
generated by vector fields of the form $X_f:=\{\cdot , f \}_\nh$, for
$f\in C^\infty(\M)$.} of $ \{\cdot , \cdot \}_\nh$ is not
integrable, in contrast with Poisson structures, whose
characteristic distributions are always tangent to symplectic
leaves. Since almost Poisson brackets and bivector fields are
equivalent, we will treat them indistinguishably.

The present paper is concerned with the lack of integrability of
nonholonomic brackets, especially in the case of nonholonomic
systems with symmetries, i.e., when $Q$ is equipped with a free and
proper action of a Lie group $G$ preserving $D$ and $\Lag$. In this
case, $\M$ inherits a $G$-action such that $\pi_\nh$, $\Ham_\subM$
and $X_\nh$ are $G$-invariant. As a consequence the orbit space
$\M/G$ is endowed with a bivector field $\pi_\red^\nh$ (defining the
{\it reduced nonholonomic bracket} $ \{\cdot , \cdot \}_\red^\nh$),
a reduced hamiltonian $\Ham_\red$, and a reduced nonholonomic vector
field $X_\red^\nh$ defining the reduced dynamics and given by
$$
 X_\red^\nh =  \{\cdot , \Ham_\red \}_\red^\nh.
$$
In many concrete situations, the reduced bivector $\pi_\red^\nh$ has
better integrability properties than $\pi_{\nh}$. Indeed, in some
cases the reduced bracket is a Poisson bracket, possibly after a
time reparametrization, see e.g.
\cite{Chapligyn_reducing_multiplier, EhlersKoiller, FedorovJovan,
Hoch, Fernandez, Naranjo2008, JovaChap, Oscar}. This is a useful
scenario for the analysis of the integrability of the system
\cite{JovaChap,FedorovJovan} or to develop a Hamilton-Jacobi theory
\cite{Oscar,JC}.

Even if $\pi^\nh_\red$ is not a Poisson structure, it may still have
an integrable characteristic distribution, which means that it is
described by {\it almost symplectic leaves}. Note that the existence
of this foliation may give information about the nonholonomic
dynamics, for example concerning the presence of conserved
quantities, since each leaf is invariant by the flow of
$X_\red^\nh$. One of our motivations in this paper is understanding
geometric features of reduced nonholonomic brackets, particularly
the integrability of their characteristic distributions.

It is worth noticing that one may have many brackets generating the
nonholonomic vector field. More precisely, we will say that an
almost Poisson bracket $\{\cdot, \cdot \}$  on $\M$ (respectively on
$\M/G$) {\it describes the dynamics} if $\{\cdot ,  \Ham_\subM\} =
X_\nh$ (respectively $\{\cdot , \Ham_\red\}= X_\red^\nh$). So, even
if $\{\cdot , \cdot \}_\red^\nh$ does not satisfy properties of
interest, one can still hope to find another bracket describing
the dynamics with the desired properties. A general mechanism to
produce brackets describing the dynamics was introduced in
\cite{PL2011}. It is based on the fact that one can use 2-forms to
modify almost Poisson brackets via the so-called {\it gauge
transformations} \cite{SeveraWeinstein}. In this procedure, 2-forms
$B$ on $\M$ are used to ``deform" the bivector field $\pi_\nh$
keeping its characteristic distribution unchanged. The resulting
modified bracket will be denoted by $\pi_\B$, and two brackets
related by a gauge transformation are called {\it gauge related}. If
$B$ is annihilated by $X_{\nh}$, then $\pi_\B$ still describes the
dynamics. An interesting point is that the reduction of gauge
related bivector fields need not be gauge related on $\M/G$. In
fact, one may verify (see \cite{PL2011}) that $\pi_\nh$ and $\pi_\B$
may lead to reduced brackets $\pi_\red^\nh$, $\pi_\red^\B$ with
fundamentally different geometric properties. For example,
$\pi_\red^\B$ can have integrable characteristic distribution (or
even be a Poisson bracket) while $\pi_\red^\nh$ does not. So we will
be concerned in this paper not only with the geometry of
$\pi_\red^\nh$ but also with other brackets of the form
$\pi_\red^\B$ describing the reduced dynamics.

In order to describe our results, it is convenient to recall some
fundamental facts regarding {\it Chaplygin symmetries} (see e.g.
\cite{Koiller1992,BlochBook}). A nonholonomic system with symmetries
is called {\it Chaplygin} if the constraint distribution $D$
complements the vertical space $V \subset TQ$ with respect to the
group action:
\begin{equation} \label{Intro:Chaplygin}
TQ=D \oplus V.
\end{equation}
In this case the reduced nonholonomic bivector $\pi_\red^\nh$ is
nondegenerate \cite{Koiller1992}, so it is equivalent to a
nondegenerate 2-form $\Omega_\red^\nh$ on $\M/G$.  The nonholonomic
character of the reduced system appears in the fact that $\Omega_\red^\nh$ is not closed.
 The following results are well known.
\begin{enumerate}
\item[$(a)$] The differential of $\Omega_\red^\nh$ is given by  \cite{BS93}
\begin{equation} \label{Intro:BS}
d\Omega_\red^\nh = - d \langle {\mathcal J}, {\mathcal K} \rangle_\red,
\end{equation}
where $\langle{\mathcal J}, {\mathcal K} \rangle_\red$ is a 2-form
on $\M/G$ explicitly obtained from the canonical momentum map $\mathcal{J}$ for
the lifted action on $T^*Q$ and the curvature $\mathcal{K}$ of the constraint
distribution (viewed as a principal connection).
\item[$(b)$] There is an identification  \cite{Koiller1992}
\begin{equation} \label{Intro:Identif} \M/G \simeq T^*(Q/G).
\end{equation}
Hence $\M/G$ has a natural symplectic structure $\Omega_{\mbox{\tiny{
$Q/G$}}}$, usually not describing the dynamics.
\item[$(c)$] The relation between $\Omega_\red^\nh$ and $\Omega_{\mbox{\tiny{$Q/G$}}}$ is given by \cite{MovingFrames}
\begin{equation} \label{Intro:Jair}
\Omega_\red^\nh = \Omega_{\mbox{\tiny{$Q/G$}}} - \langle {\mathcal J}, {\mathcal K} \rangle_\red.
\end{equation}
Note that, in particular, this explains \eqref{Intro:BS} in $(a)$.
\end{enumerate}

In this paper, we provide extensions of the results in $(a)$, $(b)$,
$(c)$ to more general situations. Instead of assuming condition
\eqref{Intro:Chaplygin}, we will replace it by the weaker {\it
dimension assumption} \cite{BKMM}
\begin{equation} \label{Intro:DimenAssumption}
TQ=D+V.
\end{equation}
Note that many classical nonholonomic systems with symmetries
satisfy \eqref{Intro:DimenAssumption} but not
\eqref{Intro:Chaplygin}, see e.g. \cite{BlochBook,BGMConservation}.
Some explicit examples will be discussed at the end of the paper.

Our starting point is  establishing a global, coordinate-free
formula for the Jacobiator of the nonholonomic bracket
$\{\cdot,\cdot\}_\nh$ on $\M$ (extending the local formula given in
\cite[Sec.~2.5]{MarsdenKoon}), not making use of symmetries. Once
symmetries are present (satisfying \eqref{Intro:DimenAssumption}),
the geometric meaning of this Jacobiator formula becomes more clear,
and leads to an expression for the Jacobiator of the reduced
nonholonomic bracket $\{\cdot,\cdot\}_\red^\nh$ on $\M /G$. In fact,
all our formulas can be easily extended to more general brackets
obtained from $\{\cdot,\cdot\}_\nh$ and $\{\cdot,\cdot\}_\red^\nh$
via gauge transformations by 2-forms.

One application of our nonholonomic Jacobiator formulas is
clarifying the link between reduced nonholonomic bivectors
$\pi_\red^\nh$ (or, more generally, $\pi_\red^\B$) and twisted
Poisson brackets. This issue was already raised in e.g. \cite[Sec.~1.8]{YKS}. Recall that a {\it twisted Poisson structure}
\cite{SeveraWeinstein} is a special type of almost Poisson bracket
whose Jacobiator is controlled by a closed 3-form; an important
feature of twisted Poisson structures is that they always have
integrable characteristic distributions. From the formula describing
the reduced Jacobiators, we find explicit conditions for the reduced
brackets to be twisted Poisson as well as a mechanical
interpretation of the 3-form with respect to which the reduced bracket is twisted. Moreover, in the special case
of Chaplygin symmetries, the Jacobiator of
$\{\cdot,\cdot\}_\red^\nh$ can be expressed in terms of
$d\Omega_\red^\nh$ (see \eqref{E:Jacobi} and \eqref{Eq:Jacdomega}
below), and our formula recovers \eqref{Intro:BS} in $(a)$.

In order to generalize $(b)$ and $(c)$, we consider symmetries
satisfying \eqref{Intro:DimenAssumption} and fix an invariant
vertical complement $W \subseteq V$ of $D$, so that $TQ=D\oplus W$,
satisfying an additional condition saying that $W$ can be realized
as the vertical space for a subgroup of the group $G$ of symmetries.
Once such $W$ is chosen, we prove that one naturally obtains a {\it
Poisson structure} $\Lambda$ on the reduced space $\M/G$. We also
show that the submanifold $W^\circ/G$ of $T^*Q/G$, where
$W^\circ\subset T^*Q$ is the annihilator of $W$, has a natural
Poisson structure $\Lambda_0$, and that there is an identification
$$
\M/G \simeq W^\circ/G
$$
which is a Poisson diffeomorphism. In the Chaplygin case, the
vertical complement $W$ is necessarily equal to $V$, the quotient
$V^\circ/G$ is canonically identified with $T^*(Q/G)$, and
$\Lambda_0$ coincides with the Poisson structure defined by
canonical symplectic 2-form $\Omega_{\mbox{\tiny{$Q/G$}}}$, thus
recovering the result in $(b)$.

In general, the Poisson structure $\Lambda$ on $\M/G$ does not
describe the nonholonomic dynamics, so one is led to compare the
Poisson structure $\Lambda$ with the reduced nonholonomic bivector
$\pi_\red^\nh$, as done in $(c)$. Note that, in terms of bivector
fields, the relation in \eqref{Intro:Jair} is equivalent to saying
that the bivectors corresponding to $\Omega_\red^\nh$ and
$\Omega_{\mbox{\tiny{$Q/G$}}}$ are gauge related by the 2-form $\langle
{\mathcal J}, {\mathcal K} \rangle_\red$. Unlike the Chaplygin case,
we observe that, in general, $\pi_\red^\nh$ is not necessarily gauge
related to $\Lambda$. However we obtain sufficient conditions,
fulfilled in many examples, guaranteeing that $\pi_\red^\nh$ (or,
more generally, $\pi_\red^\B$) is obtained from the Poisson bivector
$\Lambda$ by a gauge transformation. We show that this happens e.g.
for the vertical rolling disk, the snakeboard and the Chaplygin
ball. Notice that proving that the reduced dynamics is described by
a bivector gauge related to the Poisson structure $\Lambda$ has
several implications; for example, the nonholonomic flow is
restricted to the symplectic leaves of $\Lambda$, so Casimirs of
$\Lambda$ are conserved quantities of the system.

\noindent {\it Outline of the paper}. In Section
\ref{S:Poisson-NHsystem} we review the basics of almost Poisson
structures, gauge transformations  and  nonholonomic systems that
are useful to our paper. In Section \ref{S:JacobiatorGeneral} we
prove the nonholonomic Jacobiator formula (see Theorem
\ref{T:GeneralJacobiatorNH} and Corollary~\ref{C:MarsdenKoon}).
Nonholonomic systems with symmetries satisfying the dimension
assumption \eqref{Intro:DimenAssumption} are considered in Section
\ref{S:Jacobiator-Symmetries}. We prove a formula for the Jacobiator
of nonholonomic brackets in this case (Theorem \ref{T:JacobiatorNH})
as well as a formula for the Jacobiator of the reduced brackets
(Corollary \ref{C:RedJacobiator}); here we also establish the
connection with twisted Poisson structures (Corollary
\ref{C:Twisted}). We show how the Jacobiator formulas simplify once
additional conditions are imposed on symmetries (see
Theorem~\ref{T:JacobVerticalSymm} and
Corollary~\ref{C:TwistedVertSymm}), and illustrate our results in
the example of a nonholonomic particle and Chaplygin systems. In
Section \ref{S:Poisson} we define the Poisson structure $\Lambda$ on
$\M/G$, based on a suitable choice of $G$-invariant vertical
complement $W$ of $D$ (see Prop.~\ref{P:poisson1}). We also define,
in two equivalent ways, a Poisson structure $\Lambda_0$ on
 $W^\circ\!/G$ (see Propositions~\ref{P:W0isPoisson} and \ref{P:W0-Poisson2})
which is independent of the kinetic metric, and prove that it is
Poisson diffeomorphic to $(\M/G, \Lambda)$
(Prop.~\ref{P:k-k0isomorphism}). In Section~\ref{S:KerVertical} we
give a description of the symplectic leaves of the Poisson structure
$\Lambda$ on $\M / G$ (in terms of the leaves of the canonical
Poisson structure on $T^*Q/G$, see Theorem~\ref{T:PoissonPw}), which
turns out to be related to the nonholonomic momentum map (see
Theorem~\ref{T:JnhIsMomenMapforJK} and
Corollary~\ref{C:LeavesOfPiP}). We also analyze, in Theorem
\ref{T:Reduced-Dyn}, when the reduced nonholonomic brackets can be
written as gauge transformations of the Poisson structure $\Lambda$.
In Section~\ref{S:Examples}, we work out several mechanical examples
illustrating our results, including the process of hamiltonization of the homogeneous sphere rolling without sliding on a cylinder using a suitable gauge transformation. We have also included an Appendix
collecting some facts about the reduction of presymplectic forms
used in the paper.



\medskip

\noindent {\it Acknowledgments}: I thank FAPERJ, CNPq (Brazil) and the GMC
Network (projects MTM2010-12116-E and MTM2012-34478, Spain), especially J.C. Marrero,
E. Padron and D. Martin de Diego, for supporting this project. I
have also benefited from attending the conferences {\it Applied
Dynamics and Geometric Mechanics}, held in Oberwolfach, and the {\it
Focus Program on Geometry, Mechanics and Dynamics, the Legacy of
Jerry Marsden}, held at the Fields Institute, so I thank the various
organizers, particularly A. Bloch and T. Ratiu, for their
invitations. I am grateful to Alejandro Cabrera, Oscar Fernandez,
Luis Garcia-Naranjo, David Iglesias-Ponte, Jair Koiller and Tom
Mestdag for several interesting discussions and encouragement, and
Henrique Bursztyn
 for his many comments on the manuscript. I particularly thank Oscar Fernandez for pointing out the
snakeboard example and Luis Garcia-Naranjo for drawing my attention to the homogeneous ball rolling on a cylinder.
Finally I thank the referees for their valuable comments that led to several improvements in the paper. 

\section{Almost Poisson structures and nonholonomic systems} \label{S:Poisson-NHsystem}

\subsection{Almost Poisson brackets and gauge transformations.}  \label{Sss:GaugeTransf}

In this section we recall basic definitions related to almost Poisson manifolds.

An {\it almost Poisson bracket} $\{ \cdot , \cdot \}$ on a manifold $P$ is an $\R$-bilinear bracket $\{ \cdot , \cdot \}: C^\infty (P) \times C^\infty(P) \to C^\infty(P)$ that is skew-symmetric and satisfies the Leibniz condition:
$$
\{fg,h\} = f \{g,h\} + \{f,h\}g, \qquad \mbox{for} \ f,g,h \in C^\infty(P).$$
If the Jacobi identity is also satisfied then the bracket $\{ \cdot , \cdot \}$ is called a {\it Poisson} structure.

Let $P$ be a  manifold equipped with an almost Poisson bracket $\{\cdot , \cdot \}$.   The {\it hamiltonian vector field} $X_f$ of a function $f\in C^\infty(P)$ is the vector field on $P$
defined by
$$
X_f(g)=\{g,f\},$$ for all $g\in C^\infty(P)$.
The \emph{characteristic distribution} of $\{\cdot , \cdot \}$ is the distribution on the manifold $P$ whose fibers are spanned by the hamiltonian vector fields. In general the characteristic distribution is not integrable. If it is integrable, so that it is tangent to leaves (of possibly varying dimensions), then each leaf inherits a nondegenerate 2-form; i.e., an almost Poisson structure with integrable characteristic distribution gives rise to a (singular) foliation with almost symplectic leaves. When an almost Poisson bracket is Poisson then its characteristic distribution is integrable and each leaf is symplectic.

Due to the Leibniz condition there is a correspondence between almost Poisson brackets and bivector fields $\pi \in \Gamma (\bigwedge ^2 TP)$ given by $$\pi(df, dg) = \{f , g \},$$ for $f, g \in C^\infty(P)$. We will work indistinguishably with almost Poisson brackets and bivector fields. We  denote by $\pi^\sharp : T^*P \rightarrow TP$ the map such that   $\beta(\pi^\sharp(\alpha)) = \pi(\alpha, \beta)$. More generally, if $\phi \in \Omega^k(P)$, then $\pi^\sharp(\phi)$ is a $k$-vector field defined on 1-forms $\a_1,...,\a_k$ by $\pi^\sharp(\phi)(\a_1,...,\a_k) = (-1)^k\phi(\pi^\sharp(\a_1),...,\pi^\sharp(\a_k))$.

Note that the characteristic distribution of $\pi$ is the image of
$\pi^\sharp$ and the hamiltonian vector field of $f \in C^\infty(P)$ is  $X_f = - \pi^\sharp(df)$.
The 3-vector field $[\pi, \pi]$, where $[\cdot,\cdot]$ is the
Schouten bracket (see e.g. \cite{GelfandDorfman79, GelfandDorfman80} or \cite{DufourZung}),  measures the
failure of the Jacobi identity of $\{\cdot,\cdot\}$ through the relation
\begin{equation}
\frac{1}{2}[\pi,\pi](df, dg, dh) = \{f, \{g,h\}\}+ \{g,\{h,f\}\} + \{h,\{f,g\}\}, \label{E:Jacobi}
\end{equation}
for $f, g, h \in C^\infty(P)$.

A bivector field $\pi$ on $P$ induces a bracket $[\cdot, \cdot ]_\pi$ on sections of $T^*P$ defined by \begin{equation}
[\alpha, \beta]_\pi = \pounds_{\pi^\sharp(\alpha)} \beta-
\pounds_{\pi^\sharp(\beta)}\alpha - d(\pi(\alpha, \beta))  =
\pounds_{\pi^\sharp(\alpha)} \beta -  {\bf i}_{\pi^\sharp(\beta)}
d\alpha. \label{eq:LAPoisson}
\end{equation} The bracket $[\cdot, \cdot ]_\pi$  is $\R$-bilinear, skew-symmetric and satisfies the Leibniz identity. In general, $\pi^\sharp$ is not necessarily bracket preserving; instead (see e.g. \cite{GelfandDorfman79, GelfandDorfman80} and \cite[Sec.~2.2]{BCrainic} for a proof),
\begin{equation}
\pi^\sharp ( [\alpha,\beta]_\pi) = [\pi^\sharp(\alpha),\pi^\sharp(\beta)] - \frac{1}{2} {\bf i}_{\alpha\wedge \beta} [\pi, \pi], \qquad \mbox{for } \alpha, \beta
\in \Omega^1(P). \label{E:not_Morph}
\end{equation}
If the bivector field $\pi$ is Poisson, i.e., $[\pi, \pi]=0$,  then there is Lie algebroid structure induced on $T^*P$ with bracket $[ \cdot, \cdot]_\pi$ and anchor map $\pi^\sharp: T^*P \to TP$, see e.g. \cite{DufourZung}.

The proof of Theorem \ref{T:GeneralJacobiatorNH} below will need the following formula. 
\begin{lemma}\label{L:CyclicFormula} If $(P, \pi)$ is an almost Poisson manifold, then for $f,g,h \in C^\infty(P)$,
\begin{equation}
 \frac{1}{2} [\pi, \pi] (df,dg,dh) = cycl. \left[ dh(\pi^\sharp ( [df,dg]_\pi) ) + dh( [\pi^\sharp(df),\pi^\sharp(dg)]) \right]. \label{E:Cyclic_not_Morph}
\end{equation}
\end{lemma}

\begin{proof} First, observe that the bracket defined in \eqref{eq:LAPoisson} satisfies $[df, dg]_{\pi} = d\{f, g\}$ and thus we have that $cycl.[dh(\pi^\sharp ( [df,dg]_\pi) ) ] = - cycl.[\{f,\{g,h\}\}]$. On the other hand, using \eqref{E:Jacobi} and \eqref{E:not_Morph} we get that 
\begin{equation*} \begin{split} 
cycl.[ dh( [\pi^\sharp(df),\pi^\sharp(dg)]) ] & =  cycl.[ dh(\pi^\sharp([df,dg]_\pi)) + \frac{1}{2} [\pi, \pi](df,dg,dh) \, ] \\
& = - cycl.[\{f,\{g,h\}\}] +3 \,cycl.[\{f,\{g,h\}\}].
\end{split} \end{equation*}
Summing up, and again by \eqref{E:Jacobi}, we obtain the desired result. 
\end{proof}

A regular distribution $F \subseteq TP$ and a  nondegenerate 2-section $\Omega_F \in \Gamma(\bigwedge^2F^*)$
define uniquely a bivector field $\pi$ on $P$ by the relation that, for $\alpha \in T^*P$ and $X \in F$,
\begin{equation} \label{Eq:bivector-section}
{\bf i}_{X} \Omega_F = \alpha |_F \qquad  \Longleftrightarrow \qquad \pi^\sharp(\alpha) = - X,
\end{equation}
where $\alpha|_F$ denotes the point-wise restriction of $\alpha$ to $F$.
Observe that, in this case, $F$ is the characteristic distribution of the bivector $\pi$. Moreover, any bivector field $\pi$ with regular characteristic distribution arises in this way. In particular, if $F$ is a regular distribution on $P$ and $\Omega \in \Omega^2(P)$ is a 2-form such that $\Omega |_F$ is nondegenerate then the pair $(F, \Omega)$ defines a bivector field on $P$.
 Note that in terms of the bracket $\{ \cdot , \cdot \}$ associated to $\pi$, \eqref{Eq:bivector-section} is written, for $f,g \in C^\infty(P)$, as
$$
\{f,g\} = \Omega_F(\pi^\sharp(df), \pi^\sharp(dg) ).$$

Given a nondegenerate 2-form $\Omega$ on $P$, \eqref{Eq:bivector-section} defines an associated bivector field $\pi$ for $F=TP$. In fact, there is a one-to-one correspondence between nondegenerate 2-forms and nondegenerate bivectors (i.e., those for which the map $\pi^\sharp: T^*P\to TP$ is an isomorphism). The Jacobiator in this case satisfies
$$
-d\Omega(\pi^\sharp(df), \pi^\sharp(dg), \pi^\sharp(dh)) =  \{f, \{g,h\}\}+ \{g,\{h,f\}\} + \{h,\{f,g\}\},
$$
for $f,g,h \in C^\infty(P)$, which is equivalent to
\begin{equation}\label{Eq:Jacdomega}
\frac{1}{2}[\pi,\pi]=\pi^\sharp(d\Omega).
\end{equation}

\subsubsection*{Gauge transformations of bivector fields}

It will be convenient to consider a special equivalence relation between bivector fields defined by the so-called {\it gauge transformations} \cite{SeveraWeinstein}.

Let $\pi$ be a bivector field on $P$ and $B$ be a 2-form on $P$ such that the endomorphism $(\textup{Id} + B^\flat\circ \pi^\sharp)$ of $T^*P$ is invertible\footnote{A gauge transformation of a bivector field can be defined for any 2-form but, if this invertibility condition is not satisfied, the result is an almost Dirac structure \cite{SeveraWeinstein}.}, where $B^\flat:TP \to T^*P$ is given by $B^\flat (X) = {\bf i}_XB$. The {\it gauge transformation of $\pi$ by the 2-form $B$} is the bivector field $\pi_\B$ on $P$ such that
\begin{equation} \label{Eq:IntroToGauge} \pi_\B^\sharp = \pi^\sharp \circ (\textup{Id} + B^\flat\circ \pi^\sharp)^{-1}. \end{equation}

Two bivectors $\pi$ and $\tilde \pi$ are called {\it gauge related} if there exists a 2-form $B$ defining a gauge transformation from $\pi$ to $\tilde \pi$.

\begin{remark}\label{R:GaugeT}
\
\begin{enumerate}
\item[$(i)$] Gauge related bivectors have the same characteristic distribution. In particular, if $\pi$ is a bivector field with an integrable characteristic distribution, then any gauge related bivector has the same associated foliation, even though the leafwise almost symplectic structures may change. Also note that gauge related brackets have the same Casimirs.

\item[$(ii)$] If $\pi$ is a regular bivector field defined by a distribution $F$ and a nondegenerate 2-section $\Omega_F$ then a {\it gauge transformation} of $\pi$ by the 2-form $B$ is the bivector field $\pi_\B$ defined by the distribution $F$ and the 2-section $\Omega_F - B|_F$.
Note that the invertibility of the endomorphism $(\textup{Id} + B^\flat\circ \pi^\sharp)$ is equivalent to  $\Omega_F - B|_F $ being nondegenerate.
\item[$(iii)$] If $\pi$ and $\tilde \pi$ are bivector fields defined by nondegenerate 2-forms $\Omega$ and $\tilde \Omega$ (via \eqref{Eq:bivector-section} with $F=TP$)  then they are automatically gauge related by the 2-form $B=\Omega-\tilde \Omega$.

\item[$(iv)$] Let $\pi$ and $\pi_\B$ be gauge related bivector fields by a 2-form $B$. If $\pi$ is Poisson and $B$ is  closed then $\pi_\B$ is Poisson, \cite{SeveraWeinstein}.
\end{enumerate}
\end{remark}

\begin{example}
On $\R^5$ with coordinates $(x,y,z,p_1,p_2)$, consider the following bivector field:
\begin{equation} \label{Ex:gauge}
\pi=a \frac{\partial}{\partial x} \wedge \frac{\partial}{\partial p_1} + \frac{\partial}{\partial y} \wedge \frac{\partial}{\partial p_2} -ab \frac{\partial}{\partial p_1} \wedge \frac{\partial}{\partial p_2},
\end{equation}
 where $a,b$ are non-zero functions, and the 2-form $B=b\, dx \wedge dy$. The endomorphism $\textup{Id} + B^\flat\circ \pi^\sharp$ is invertible since, written in the basis $\{dx,dy,dz,dp_1, dp_2\}$, the matrix $B^\flat\circ \pi^\sharp$ is upper triangular.  In order to compute the gauge transformation of $\pi$ by $B$, denoted by $\pi_\B$, we use equation \eqref{Eq:IntroToGauge} to see that
$$\pi^\sharp(dx)= \pi_\B^\sharp  \circ (dx + B^\flat\circ \pi^\sharp(dx)) = \pi_\B^\sharp(dx)$$ Doing the same for the rest of the variables we obtain that
$\pi_\B^\sharp (dp_1)= \pi^\sharp(ab\, dy +dp_1)$ and $\pi_\B^\sharp
(dp_2)= \pi^\sharp(-bdx+dp_2)$, while on the other elements of the
basis $\pi$ and $\pi_\B$ coincide. Therefore we obtain that
$\displaystyle{\pi_\B=a \frac{\partial}{\partial x} \wedge
\frac{\partial}{\partial p_1} + \frac{\partial}{\partial y} \wedge
\frac{\partial}{\partial p_2}.}$ Similarly, the gauge transformation
of $\pi$ by $-B$ is
\begin{equation} \label{Ex:PiB} \pi_{\mbox{\tiny $-B$}} = a \frac{\partial}{\partial x} \wedge \frac{\partial}{\partial p_1} + \frac{\partial}{\partial y} \wedge \frac{\partial}{\partial p_2} -2ab \frac{\partial}{\partial p_1} \wedge \frac{\partial}{\partial p_2}.
\end{equation}
 The three bivectors $\pi$, $\pi_\B$ and $\pi_{\mbox{\tiny $-B$}}$ are gauge related and $z$ is a Casimir for all of them. Finally, observe that $\tilde \pi = \frac{\partial}{\partial y} \wedge \frac{\partial}{\partial p_2} -ab \frac{\partial}{\partial p_1} \wedge \frac{\partial}{\partial p_2}
$ is not gauge related with $\pi$ since $x$ is a Casimir only for $\tilde \pi$.
\end{example}

\subsection{Nonholonomic systems and the nonholonomic bracket} \label{Ss:nh-systems}

Let us consider a nonholonomic system determined by a lagrangian
$\Lag:TQ \to \R$ of the form
$$
\Lag = \frac{1}{2}\kappa - U,
$$
 where $\kappa$ is a Riemannian metric (called the {\it kinetic energy metric}) on $Q$ and $U \in C^\infty(Q)$ is the potential energy, and a nonintegrable distribution $D$ on $Q$ (i.e., a non-involutive subbundle $D\subset TQ$), which describes the permitted velocities of the system. We denote the annihilator of $D$ by $D^\circ$.

Note that the Legendre transform  $\L : TQ \to T^*Q$ is
 a global diffeomorphism since  $\L= \kappa^\sharp$, where $\kappa^\sharp:TQ \to T^*Q$ is given by $\langle \kappa^\sharp(v), w \rangle = \kappa(v,w)$ for $v,w \in TQ$.

Let $\Ham \in C^\infty(T^*Q)$ be the hamiltonian function $\Ham:T^*Q \to \R$ associated with the lagrangian $\mathcal{L}$. In local coordinates $(q^i,p_i)$ of $T^*Q$ the hamiltonian is given by $\Ham(q^i,p_i) = (\frac{\partial \mathcal{L}}{\partial \dot q^i}. \dot q^i- \mathcal{L} )\circ \L ^{-1}$.   If locally, $D^\circ = \textup{span} \{ \epsilon^a \ : \ a= 1,...,k < n.\}$ where $\epsilon^a$ are 1-forms on $Q$, then the equations of motion can be written as a first order system on the cotangent bundle $T^*Q$ given by
\begin{equation}
\label{E:Eqns_of_motion_Ham_form} \dot q^i = \frac{\partial \Ham}{\partial p_i}, \qquad \dot p_i= - \frac{\partial
\Ham}{\partial q^i} + \lambda_a\epsilon^a_i, \qquad i=1,\dots,n,
\end{equation}
where $n = \textup{dim}\, Q$ and  $\lambda_a, a=1,\dots,k,$ are functions which are uniquely determined by the fact that the constraints are satisfied.
Then the constraint equations become
\begin{equation}
\label{E:Constraints_Ham_form} \epsilon_i^a(q) \frac{\partial \Ham}{\partial p_i}=0, \qquad a=1,\dots,k.
\end{equation}

The geometry underlying a nonholonomic system allows one to write the equations of motion in an intrinsic way.  Let $\M := \L(D) \subset T^*Q$ be the {\it constraint submanifold}. Since the Legendre transform is linear on the fibers, $\M$ is a vector subbundle of $T^*Q$. We denote by $\tau:\M \to Q$ the restriction to $\M$ of the canonical projection $\tau_Q : T^*Q \to Q$.

If $\Omega_Q$ is the canonical 2-form on $T^*Q$, we denote by $\Omega_\subM$ the 2-form on $\M$ defined by $\Omega_\subM := \iota^* \Omega_Q$ where $\iota : \M\to T^*Q$ is the natural inclusion.

The constraints are intrinsically written as a regular, non-integrable distribution $\C$ on $\M$ given, at each $m \in \M$, by
\begin{equation} \label{Eq:C}
\C_m= \{v \in T_m \M \ : \ T\tau (v) \in D_{\tau(m)}\}.
\end{equation}
A fundamental result given in \cite[Sec.~5]{BS93} is that the point-wise restriction of $\Omega_\subM$ to $\C$, denoted by $\Omega_\subM |_\C$, is nondegenerate.

The {\it nonholonomic vector field}  $X_\nh \in \mathfrak{X}(\M)$ is the vector field uniquely defined by
\begin{equation}
\label{Eq:NHDyncamics} {\bf i}_{X_\nh} \Omega_\subM |_\C = d\Ham_\subM |_\C,
\end{equation}
where $\Ham_\subM$ is the restriction to $\M$ of the hamiltonian function $\Ham$, i.e.,  $\Ham_\subM = \iota^*\Ham  : \M \to \R$ .
The integral curves of $X_\nh$ are solutions of the system \eqref{E:Eqns_of_motion_Ham_form}.

We define the {\it nonholonomic} bivector field  $\pi_{\mbox{\tiny nh}}$ to be the bivector associated to the distribution $\C$ and the 2-form $\Omega_\subM$, as in \eqref{Eq:bivector-section}
(see \cite{SchaftMaschke1994,Marle1998, IbLeMaMa1999}). Then $(\M, \pi_\nh)$ is an almost Poisson manifold.  The distribution $\C$ is the characteristic distribution of the bivector $\pi_\nh$ and since $\C$ is always not integrable,  $\pi_\nh$ is never a Poisson structure. The  nonholonomic bivector describes the dynamics in the sense that
$$
\pi_\nh^\sharp (d\Ham_\subM) = - X_\nh.
$$

So, in hamiltonian form, a nonholonomic system is described by the triple $(\M, \pi_\nh, \Ham_\subM)$, defined from the lagrangian $\Lag$ and the distribution $D$.

We call  $(\M, \pi_\nh)$ {\it the almost Poisson manifold associated to the nonholonomic system}.

In \cite{PL2011} we saw that it is worth to explore different bivectors $\tilde \pi$ on $\M$ such that $X_\nh   = - \tilde \pi^\sharp(d\Ham_\subM)$. In this case we say that $\tilde \pi$ also {\it describes the dynamics}.

In order to generate new bivector fields on $\M$ we will use {\it gauge transformations} of the classical nonholonomic bivector $\pi_\nh$ by 2-forms (as in Section \ref{Sss:GaugeTransf}). If we want to generate bivector fields describing the dynamics, we need the following definition given in \cite[Sec.4.2]{PL2011},
\begin{definition} \label{D:dynamicalGauge}
Let $(P, \pi)$ be an almost Poisson manifold with a distinguished Hamiltonian function $H \in C^\infty(P)$. Given a 2-form $B$ on $P$ such that  $(\textup{Id} + B^\flat\circ \pi^\sharp)$ is invertible,  the gauge transformation of $\pi$ associated to the 2-form $B$ is said to be a {\it dynamical gauge transformation} if
$${\bf i}_{X_H} B =0,
$$
where $X_H$ is the hamiltonian vector field associated to $H$.
\end{definition}
Thus, any bivector $\pi_\B$ that is dynamically gauge related to $\pi$ satisfies
$
\pi_\B^\sharp (dH) = \pi^\sharp(dH).
$
We are interested in using Definition \ref{D:dynamicalGauge} to perform a dynamical gauge transformation of the almost Poisson structure $\pi_\nh$, where the distinguished function is $\Ham_\subM$.

\begin{remark} \label{R:DynGaugeSemi-basic} In \cite[Proposition 11]{PL2011}, it was shown that if we consider a semi-basic 2-form $B$ with respect to the fiber bundle $\tau: \M \to Q$ ( i.e.,  ${\bf i}_X B =0$  if  $T\tau(X)=0$) then  $(\textup{Id} + B^\flat\circ \pi_\nh^\sharp)$ is invertible.

\end{remark}

\section{The Jacobiator of the nonholonomic bracket}
\label{S:JacobiatorGeneral}

Let $(\M, \pi_\nh)$ be the almost Poisson manifold associated to a nonholonomic system. In this section we present a coordinate-free formula for the Jacobiator of the classical nonholonomic bivector $\pi_\nh$, or any bivector gauge equivalent to it. This section will not use the hamiltonian function $\Ham_\subM$.

\subsubsection*{Splittings of $T\M$ adapted to the constraints}

Recall that the characteristic distribution of $\pi_\nh$ is the nonintegrable distribution $\C$ defined in \eqref{Eq:C}.
Since  $\C$ is a regular distribution on $\M$, it is possible to choose a complement $\W$ of $\C$  on $\M$ such that, for each $m \in \M$,
\begin{equation} \label{Eq:SplittingOfTM}
T_m\M = \C_m \oplus \mathcal{W}_m.
\end{equation}

Consider the projections $P_\subC :T\M \to \C$ and $P_\subW :T\M \to \W$ associated to the decomposition
\eqref{Eq:SplittingOfTM}.  Since $P_\subW :T\M \to \W$ can be seen as a $\W$-valued 1-form, following \cite{BKMM,MarsdenKoon}, we define the $\W$-valued 2-form ${\bf K}_\subW$ on $\M$ given by
\begin{equation}
\label{Def:K} {\bf K}_\subW(X,Y) = - P_\subW( [P_\subC (X), P_\subC(Y)] ) \qquad \mbox{for } X,Y \in \mathfrak{X}(\M).
\end{equation}

It is straightforward to check that ${\bf K}_\subW$ is $C^\infty (\M)$-bilinear. Note also that ${\bf K}_\subW \equiv 0$ if and only if $\C$ is involutive (see also \cite[Sec.~2.5]{MarsdenKoon}).

\begin{proposition} \label{Prop:Ksemi-basic}
The $\W$-valued 2-form ${\bf K}_\subW$  is semi-basic with respect to the bundle projection $\tau:\M \to Q$, i.e., \ ${\bf i}_X {\bf K}_\subW =0$ \ if \ $T\tau(X)=0$.
\end{proposition}

\begin{proof}
Let us consider $X, Y \in \mathfrak{X}(\M)$ such that $X$ is $\tau$-related with 0 and $Y\in \Gamma(\C)$ is $\tau$-projectable. Since $\textup{Ker}\, T\tau \subset \C$ then $[P_\subC (X), P_\subC(Y)] = [X,Y]$ is $\tau$-related with 0.
Therefore, $[P_\subC (X), P_\subC(Y)] \in \Gamma(\C)$ and thus, from the expression \eqref{Def:K} we obtain that  ${\bf K}_\subW(X,Y) =0$.  This implies that ${\bf i}_X {\bf K}_\subW \equiv 0$ since $\tau$-projectable vector fields generate $T\M$ at each point.
\end{proof}

\subsubsection*{Gauge transformations adapted to the splitting}

From \eqref{Eq:IntroToGauge} we observe that a gauge transformation of $\pi_\nh$ by a 2-form $B$ is completely determined by the point-wise restriction of $B$ to $\C$ (i.e., by the section $B|_\C$ of $\bigwedge^2\C^*$). That is, if two 2-forms $B$ and $\bar B$ are such that $(B-\bar B)|_\C \equiv 0$, then $\pi_\B$ and $\pi_{\mbox{\tiny{$\bar{B}$}}}$ agree.

Note that for any 2-form $\bar{B}$, one can find another 2-form $B$ which agrees with $\bar B$ when restricted to $\C$, so that $\pi_\B = \pi_{\mbox{\tiny{$\bar{B}$}}}$, and that satisfies  the additional condition
\begin{equation} \label{Eq:Bsection}
{\bf i}_Z B \equiv 0 \qquad \mbox{for all } Z \in \Gamma(\mathcal{W}).
\end{equation}
So from the point of view of gauge transformations, there is no loss in generality in assuming that \eqref{Eq:Bsection} holds. This condition will be assumed in the sequel to simplify some formulas.

\subsubsection*{The Jacobiator of $\pi_\nh$ and gauge related bivectors}

Now we state the theorem which characterizes the Jacobiator of any bivector $\pi_\B$ gauge related to $\pi_\nh$. Even though the most interesting statement for this Theorem occurs when a Lie group is acting, we will first state it in the more general setting without the consideration of a symmetry.

Let us consider the almost Poisson manifold $(\M, \pi_\nh)$ associated to a nonholonomic system, as in Section \ref{Ss:nh-systems}.

\begin{theorem} \label{T:GeneralJacobiatorNH}
Let $\pi_\B$ be a bivector on $\M$ gauge related to $\pi_{\emph \nh}$ by the 2-form $B$. Let $\W$ be a complement of $\C$ and assume that $B$ verifies \eqref{Eq:Bsection}. Then for $\a, \beta, \gamma \in T^*\M$ we have
\begin{equation*}
\begin{split}
\frac{1}{2}[\pi_\B, \pi_\B] (\alpha, \beta, \gamma) =  & \ cycl.  \left[ \Omega_{\mbox{\tiny{$\M$}}} ( {\bf K}_\subW(\pi_\B^\sharp(\alpha), \pi_\B^\sharp(\beta)), \pi_\B^\sharp (\gamma) ) - \gamma ( {\bf K}_\subW(\pi_\B^\sharp(\alpha), \pi_\B^\sharp(\beta))  \, )
\right] \\
& + \ d B (\pi_\B^\sharp(\alpha), \pi_\B^\sharp(\beta), \pi_\B^\sharp (\gamma) ).
\end{split}
\end{equation*}
\end{theorem}

\begin{proof} Since this formula is $C^\infty(\M)$-linear on $\a, \beta, \gamma$, it is sufficient
to check it on exact forms. Applying formula \eqref{E:Cyclic_not_Morph} we obtain, for $f,g,h \in C^\infty(\M)$,
\begin{eqnarray}
\frac{1}{2}[\pi_\B, \pi_\B](df,dg,dh) \hspace{-0.6cm}  && = \,
 cycl. \left[ \pi_\B^\sharp ( [df, dg]_{\pi_\B})(h) +
[\pi_\B^\sharp(df), \pi_\B^\sharp(dg) ](h) \, \right] \nonumber \\
&& \begin{split} =  \ cycl.  & \left[-\pi_\B^\sharp (dh)(\{ f,g\}_\B) + dh
(P_{\mbox{\tiny $\C$}}([\pi_\B^\sharp(df), \pi_\B^\sharp(dg) ]) ) \right. \\
& \left. \ + \ dh (
P_{\mbox{\tiny ${\mathcal W}$}}([\pi_\B^\sharp(df), \pi_\B^\sharp(dg) ]) )
\, \right], \end{split} \label{Proof:eq1}
\end{eqnarray}
where we are using the fact that $\textup{Id} = P_{\mbox{\tiny $\C$}} + P_{\mbox{\tiny $\mathcal{W}$}}$.
If we call $\tilde \Omega_\subM := \Omega_\subM - B$, then by Remark~\ref{R:GaugeT}$(ii)$ and \eqref{Eq:bivector-section} we have ${\bf i}_{\pi_\B^\sharp(dh)} \tilde \Omega_\subM |_\C = -dh |_\C$ and thus
\begin{eqnarray*}
dh (P_{\mbox{\tiny $\C$}}([\pi_\B^\sharp(df), \pi_\B^\sharp(dg) ]) ) & = & - \tilde \Omega_\subM(\pi_\B^\sharp(dh),P_{\mbox{\tiny $\C$}}([\pi_\B^\sharp(df), \pi_\B^\sharp(dg) ]) ) \\
& = & - \tilde \Omega_\subM (\pi_\B^\sharp(dh),[\pi_\B^\sharp(df), \pi_\B^\sharp(dg) ] ) + \tilde \Omega_\subM(\pi_\B^\sharp(dh),P_{\mbox{\tiny
$\mathcal{W}$}}([\pi_\B^\sharp(df), \pi_\B^\sharp(dg) ]) ).
\end{eqnarray*}
Thus, \eqref{Proof:eq1} is equivalent to
\begin{equation}  \label{Proof:eq3}
\begin{split} \frac{1}{2}[\pi_\B, \pi_\B](df,dg,dh) = cycl.\,  & \left[ - \pi_\B^\sharp(dh) \left( \tilde \Omega_\subM(\pi_\B^\sharp(df),
\pi_\B^\sharp(dg)) \right)  + \tilde \Omega_\subM ([\pi_\B^\sharp(df), \pi_\B^\sharp(dg) ] , \pi_\B^\sharp(dh)) \right.  \\
& \left. - \tilde \Omega_\subM (P_{\mbox{\tiny $\mathcal{W}$}}([\pi_\B^\sharp(df), \pi_\B^\sharp(dg) ]), \pi_\B^\sharp(dh)) - {\bf K}_\subW (\pi_\B^\sharp(df), \pi_\B^\sharp(dg)) (h) \right] ,\end{split}
\end{equation}
where the last term is the definition of ${\bf K}_\subW$ given in \eqref{Def:K}.

Note that the cyclic sum of the first two terms of \eqref{Proof:eq3} is just $-d \tilde \Omega_\subM (\pi_\B^\sharp(df), \pi_\B^\sharp(dg), \pi_\B^\sharp(dh))$ which is equal to \ $d B(\pi_\B^\sharp(df), \pi_\B^\sharp(dg), \pi_\B^\sharp(dh))$. \ By  \ \eqref{Eq:Bsection} we \ also \ have \ that \ $\displaystyle{ {\bf i}_{P_{\mbox{\tiny $\mathcal{W}$}}([\pi_\B^\sharp(df), \pi_\B^\sharp(dg) ])} \tilde \Omega_\subM = }$ $\displaystyle{ = {\bf i}_{P_{\mbox{\tiny $\mathcal{W}$}}([\pi_\B^\sharp(df), \pi_\B^\sharp(dg) ])}  \Omega_\subM}$. Finally, again using \eqref{Def:K}, we obtain from \eqref{Proof:eq3}
\begin{equation*}
\begin{split} \frac{1}{2}[\pi_\B, \pi_\B](df,dg,dh) =  & \ \ d B (\pi_\B^\sharp(df), \pi_\B^\sharp(dg), \pi_\B^\sharp(dh)) \\
& + cycl \, \left[ \Omega_\subM ({\bf K}_\subW (\pi_\B^\sharp(df), \pi_\B^\sharp(dg) ), \pi_\B^\sharp(dh)) -{\bf K}_\subW(\pi_\B^\sharp(df),
\pi_\B^\sharp(dg)) (h) \right].
\end{split}
\end{equation*}

\end{proof}

\bigskip

As a particular case of the previous Theorem, we obtain a formula for the Jacobiator of the nonholonomic bivector $\pi_\nh$.

\begin{corollary} \label{C:MarsdenKoon}
The Jacobiator of the classical nonholonomic bivector $\pi_{\emph \nh}$ is written, for a choice of $\W$, as
$$
\frac{1}{2}[\pi_{\emph \nh}, \pi_{\emph \nh}] (\alpha, \beta, \gamma)=   cycl \, \left[  \Omega_\subM ({\bf K}_\subW (\pi^\sharp_{\emph \nh}(\alpha), \pi^\sharp_{\emph \nh}(\beta) ), \pi^\sharp_{\emph \nh}(\gamma)) - \gamma \left( {\bf K}_\subW(\pi^\sharp_{\emph \nh}(\alpha), \pi^\sharp_{\emph \nh}(\beta)) \right) \right].
$$
\end{corollary}

In \cite[Sec.~2.5]{MarsdenKoon} , Marsden and Koon presented a local version of the above formula for an explicit choice of $\W$ that is induced by the local coordinates. Also, in agreement with  \cite{MarsdenKoon}, we see that  $\pi_\nh$ is Poisson if and only if  ${\bf K}_\subW \equiv 0$.


\section{The Jacobiator of reduced nonholonomic brackets} \label{S:Jacobiator-Symmetries}

In this section we use Theorem \ref{T:GeneralJacobiatorNH} to obtain a formula for the Jacobiator of any bivector field $\pi_\B$ gauge related to $\pi_\nh$ when the nonholonomic system admits symmetries. This formula will be in terms of a curvature-like object and the momentum associated to the cotangent lift of the group action. This will lead to a formula for the Jacobiator of reduced brackets, and it will provide information about the integrability of their characteristic distributions.

\subsection{Nonholonomic systems with symmetries} \label{Ss:DefwithSymmetries}

Let $G$ be a Lie group acting on $Q$ freely and properly. Suppose that $G$ is a symmetry of a given nonholonomic system, that is, the lifted action on $TQ$ leaves the lagrangian $\Lag$ and the constraints $D$ invariant. Then, $\M$ is an invariant submanifold of $T^*Q$ by the cotangent lift of the action.
Therefore we restrict the $G$-action on $T^*Q$ to $\M$:
$$
\varphi: G \times \M \to \M,
$$
so $\xi_{\M}(m) = \xi_{T^*Q}(m)$. As a consequence of the $G$-invariance of the constraints and lagrangian, the nonholonomic bivector $\pi_\nh$ and the hamiltonian function $\Ham_\subM$ are invariant by the $G$-action $\varphi$.

Throughout this section, we will not use explicitly the hamiltonian function $\Ham_\subM$; instead we will focus on the $G$-invariant distribution $D$ on $Q$ and the $G$-invariant almost Poisson manifold $(\M, \pi_\nh)$.

\subsubsection*{The dimension assumption}

At each $m \in \M$, let us denote by $\V_m := T_m (\textup{Orb}(m))$ the tangent to the orbit at $m$ of the $G$-action on $\M$.
 We will always assume that the following condition, known as the {\it dimension assumption}, holds:
\begin{equation} \label{Eq:DimensionAssumption}
T_m\M = \C_m + \V_m, \qquad \mbox{for all } m \in \M.
\end{equation}

Let us denote by $\S$ the distribution on $\M$ given, at each $m \in \M$, by
\begin{equation} \label{Eq:S}
\S_m=\mathcal{V}_m \cap \C_m.
\end{equation}
It is clear from \eqref{Eq:DimensionAssumption} that $\S$ has constant rank.
If, at each $q\in Q$, we denote by $V_q :=T_q(\textup{Orb}(q))$ the tangent to the orbit associated to the $G$-action on $Q$, \eqref{Eq:DimensionAssumption} is equivalent to the dimension assumption considered in \cite{BKMM}:
$$
T_qQ= D_q + V_q , \qquad \mbox{for all } q \in Q,
$$ by the definition of $\C$ and since, at each $m \in \M$, $\V_m$ and $V_{\tau(m)}$ are isomorphic. Analogously, we can define the distribution $S$ on $Q$ given by $S_q:=D_q \cap V_q$ for each $q \in Q$.

 If $\S = \{0\}$, the system is called {\it Chaplygin} and the dimension assumption is just
\begin{equation} \label{Eq:Chaplygin}
T\M = \C \oplus \V.
\end{equation}

From now on, we will always assume that the symmetry Lie group acts
freely and properly on $Q$  so that the dimension assumption is
satisfied.


\subsubsection*{The momentum map}

Recall that the submanifold $\iota: \M \hookrightarrow T^*Q$ is invariant under the cotangent lift of a $G$-action by symmetries.
Let us denote $\Omega_\subM = \iota^* \Omega_Q$ and $\Theta_\subM =
\iota^* \Theta_Q$, where $\Omega_Q$ and $\Theta_Q$ are the canonical 2-form and the Liouville 1-form on $T^*Q$,  respectively. We let ${\mathcal J}: \M \to \mathfrak{g}^*$ be the restriction to $\M$ of the canonical momentum map of the lifted action on $T^*Q$:
\begin{equation} \label{Eq:Def-MomentumMapOnM}
\langle {\mathcal J}(m) ,\eta \rangle = {\bf i}_{\eta_\subM(m)} \Theta_\subM (m), \qquad \mbox{for } \eta \in \mathfrak{g}.
\end{equation}
Note that ${\bf i}_{\eta_\subM} \Omega_\subM = d \langle {\mathcal J} , \eta \rangle$. More generally, we will need to consider vector fields $\eta_\M$ associated with elements $\eta \in \Gamma(\M\times \mathfrak{g})=C^\infty(\M,\mathfrak{g})$ which are not constant (i.e., $\eta_\M(m):= (\eta_m)_\M(m)$); in this case, we have
\begin{equation}\label{Eq:hamsection}
{\bf i}_{\eta_\subM} \Omega_\subM =  \langle d{\mathcal J} , \eta \rangle.
\end{equation}

\subsubsection*{Splittings of $T\M$ adapted to the symmetries and constraints} \label{SubSec:NhConnection}

Let $(\M, \pi_\nh)$ be the $G$-invariant almost Poisson manifold associated to a nonholonomic system with symmetries satisfying the dimension assumption \eqref{Eq:DimensionAssumption}.
Let $\W$ be a $G$-invariant distribution on $\M$ such that, at each $m \in \M$,
\begin{equation}
T_m\M = \C_m \oplus \mathcal{W}_m   \qquad \mbox{and} \qquad \W_m \subseteq \V_m.  \label{Eq:SplitbyConnection}
\end{equation}
Note that $\V_m= \S_m \oplus \W_m$.  We refer to such $\W$ as a {\it $G$-invariant vertical complement} of $\C$ in $T\M$.

Since the projection $P_\subW: T\M \to \W$ associated to the decomposition \eqref{Eq:SplitbyConnection}
is an equivariant map (i.e., $P_\subW (T\varphi_g(v_m))= T\varphi_g (P_\subW (v_m))$), it induces a $\mathfrak{g}$-valued 1-form $\mathcal{A}_\subW$ by
\begin{equation} \label{Eq:AkinPw}
 {\mathcal A}_\subW(v_m) = \xi \qquad \Longleftrightarrow \qquad P_\subW(v_m) = \xi_\subM(m), \quad \mbox{for} \ v_m \in T_m\M.
\end{equation}
Observe that $\mathcal{A}_\subW:T\M \to \mathfrak{g}$ is, in general, not surjective; note also that it satisfies $\mathcal{A}_\subW(T\varphi_g(v_m))=Ad_g\mathcal{A}_\subW(v_m)$ for $g \in G$.

\begin{remark} \label{R:Split-on-MandQ}
A choice of a $G$-invariant vertical complement $W$ of $D$ in $TQ$,
\begin{equation}
\label{Eq:DecompTQ}
TQ=D \oplus W \qquad \mbox{and} \qquad W \subseteq V,
\end{equation}
induces a $G$-invariant $\W$ satisfying \eqref{Eq:SplitbyConnection}. Indeed, since $T\tau|_\V : \V \to V$ is an isomorphism, we can take $\mathcal{W} := (T\tau|_\V)^{-1}(W)$.

Note that such $G$-invariant vertical complements $W$ always exist: for example, one can take $W$ to be $V \cap S^\perp$, where $S^\perp$ is the orthogonal complement of $S$ with respect to (the $G$-invariant) kinetic energy metric $\kappa$, as proposed in \cite{BKMM}.
In this case, $\mathcal{A}_\subW
= \tau^* \mathcal{A}^\kin$, where $\mathcal{A}^\kin$ is the $\mathfrak{g}$-valued 1-form on $Q$ used in \cite{BKMM} to define the nonholonomic connection.
However, in this paper, we will allow different choices of $W$; in fact, we will consider many examples in which $W$ is not defined in this way.
\end{remark}


\subsubsection*{The ${\W}$-curvature}
A $G$-invariant $\W$ as in \eqref{Eq:SplitbyConnection} and the resulting  map ${\mathcal A}_\subW: T\M \to \mathfrak{g}$ \eqref{Eq:AkinPw} induce a $\mathfrak{g}$-valued 2-form on $\M$, that might be
interpreted as the associated ``curvature''. More precisely:

\begin{definition} \label{Def:KinCurvature}
The $\W${\it-curvature} is the $\mathfrak{g}$-valued 2-form on $\M$ given by
\begin{equation} \label{Def:KinCurv-2form}
\mathcal{K}_\subW(X,Y):=  d {\mathcal A}_\subW(P_\subC (X),P_\subC (Y)) \qquad \mbox{ for } X,Y \in T\M,
\end{equation}
where $P_\subC: T\M \to \C$ is the projection associated to the decomposition \eqref{Eq:SplitbyConnection}.
\end{definition}

\begin{remark} \label{R:KinChaplygin} In the case of a Chaplygin system, where
$\C$ is a complement of the vertical space $\V$, the only choice of $\W$ is $\W = \V$, hence $\mathcal{A}_\V= {\mathcal A}$ is the principal connection on $\M$ with horizontal space given by $\C$. Therefore $\mathcal{K}_\V= \mathcal{K}$ is the classical curvature.
\end{remark}

\begin{proposition} \label{Prop:Prop1} For any $G$-invariant vertical complement $\W$ of $\C$ in $T\M$, the $\W$-curvature $\mathcal{K}_\subW$ verifies the following:
\begin{enumerate}
\item[$(i)$]  $\mathcal{K}_\subW(X,Y) = - \mathcal{A}_\subW ([X,Y])$
    for $X,Y \in \Gamma(\C)$.
\item[$(ii)$] $\mathcal{K}_\subW$ is ad-equivariant, i.e., $\pounds_{\xi_\subM}\mathcal{K}_\subW = ad_\xi \mathcal{K}_\subW$, for $\xi \in \mathfrak{g}$.

\end{enumerate}
\end{proposition}

\begin{proof}
$(i)$:  If $X, Y \in \Gamma(\C)$ then $$\mathcal{K}_\subW(X,Y) = d \mathcal{A}_\subW (X,Y) = X(
\mathcal{A}_\subW(Y)) - Y (\mathcal{A}_\subW(X)) - \mathcal{A}_\subW([X,Y])= -
\mathcal{A}_\subW([X,Y]).
$$
$(ii)$:  For $g \in G$, we have
\begin{eqnarray*}
\varphi_g ^* \mathcal{K}_\subW (X, Y) &=& d \mathcal{A}_\subW (P_\subC T \varphi_g (X), P_\subC T \varphi_g(Y)) = - \mathcal{A}_\subW ( [T \varphi_g P_\subC (X), T \varphi_g P_\subC (Y)])\\
&=&  - \mathcal{A}_\subW (T \varphi_g [ P_\subC (X),P_\subC (Y)]) = -Ad_g \mathcal{A}_\subW ( [ P_\subC (X),P_\subC (Y)]) =  Ad_g \mathcal{K}_\subW (X,Y),
\end{eqnarray*}
so the claim in $(ii)$ follows.
\end{proof}

As a consequence of item $(i)$ and \eqref{Eq:AkinPw}, we see that for $X,Y \in
\mathfrak{X}(\M)$ we have that
\begin{equation} \label{Eq:K-Kkin}
{\bf K}_\subW(X,Y)= \left(  \mathcal{K}_\subW(X, Y) \right)_\subM,
\end{equation}
for ${\bf K}_\subW$ in \eqref{Def:K}. Hence,  $\mathcal{K}_\subW \equiv 0$ if and only if $\C$ is involutive.

Finally, we consider a natural 3-form associated with the momentum map and the $\W$-curvature. Let $X,Y,Z$ be vector fields on $\M$. Then
$\mathcal{K}_\subW(X,Y) \in \Gamma(\M \times
\mathfrak{g})$ and $d{\mathcal J}(Z) \in \Gamma (\M \times \mathfrak{g}^*)$.
We define $ d {\mathcal J} \wedge \mathcal{K}_\subW $ to be the 3-form on $\M$ given by
\begin{equation} \label{Eq:dJ^K}
d {\mathcal J} \wedge \mathcal{K}_\subW (X,Y,Z) = cyclic \left[ \langle d{\mathcal J} (Z) ,{\mathcal K}_\subW (X,Y) \rangle \right]
\end{equation} where  $\langle \cdot , \cdot \rangle$ is the pairing between $\mathfrak{g}^*$ and $\mathfrak{g}$.


\subsection{The Jacobiator for the nonholonomic brackets in the presence of symmetries} \label{Ss:Jacob-sym}

As a consequence of Theorem \ref{T:GeneralJacobiatorNH}, we obtain the following Jacobiator formula:

\begin{theorem} \label{T:JacobiatorNH} Let $(\M,\pi_{\emph\nh})$ be the $G$-invariant almost Poisson manifold associated to a nonholonomic system with symmetries satisfying the dimension assumption \eqref{Eq:DimensionAssumption}. Suppose that $\pi_\B$ is a bivector field on $\M$ gauge related to $\pi_{\emph \nh}$ by a 2-form $B$
satisfying \eqref{Eq:Bsection}, and let $\W$ be a $G$-invariant vertical complement of $\C$. Then
\begin{equation}\label{Eq:JacpiB}
\frac{1}{2}[\pi_\B, \pi_\B] = -\pi_\B^\sharp (d B + d {\mathcal J} \wedge \mathcal{K}_\subW ) - \psi_{\pi_\B},
\end{equation}
where ${\mathcal J}: \M \to \mathfrak{g}^*$ is the canonical momentum map, $\mathcal{K}_\subW$ is the $\W$-curvature, and $\psi_{\pi_\B}$ is the trivector $\psi_{\pi_\B}(\alpha,\beta, \gamma) = cycl \left[ \gamma  \left( \mathcal{K}_\subW (\pi_\B^\sharp(\alpha),
\pi_\B^\sharp(\beta)) \right)_\subM \, \right]$, for $\a, \beta, \gamma \in T^*\M$.
\end{theorem}

\begin{proof} Let  $X,Y,Z \in \Gamma(\C)$. Note that $\mathcal{K}_\subW(X ,Y) \in \Gamma(\M \times
\mathfrak{g})$ is not necessarily a constant section, so using \eqref{Eq:hamsection} and \eqref{Eq:K-Kkin}, we get
$$ \Omega_\subM({\bf K}_\subW(X ,Y) , Z) =  \Omega_\subM( \left( \mathcal{K}_\subW(X ,Y) \right)_\subM , Z)= \langle d {\mathcal J}(Z) \, , \, \mathcal{K}_\subW(X ,Y) \rangle.$$
From Theorem \ref{T:GeneralJacobiatorNH} and using \eqref{Eq:dJ^K}, the Jacobiator of $\pi_\B$ can be written as
\begin{eqnarray*}
\frac{1}{2}[\pi_\B,\pi_\B](\a,\beta,\gamma) & = & \left(dB + d {\mathcal J} \wedge \mathcal{K}_\subW   \right) (\pi_\B^\sharp(\a), \pi_\B^\sharp(\beta), \pi_\B^\sharp(\gamma)) 
- cycl. \, \left[ \gamma \left(  \mathcal{K}_\subW(\pi_\B^\sharp(\a),
\pi_\B^\sharp(\beta))_{\!\subM} \right) \right],
\end{eqnarray*}
which completes the proof.
\end{proof}


Note that the lack of integrability of $\C$ can be seen directly
from Theorem~\ref{T:JacobiatorNH}. Since $\C$ is regular, it is
integrable if and only if it is involutive. By \eqref{E:not_Morph}
and \eqref{Eq:JacpiB}, one sees that $\C$ fails to be involutive due
to the presence of the trivector $\psi_{\pi_\B}$ (which vanishes if
and only if $\mathcal{K}_\subW \equiv 0$).

\begin{corollary}
For a nonholonomic system with a symmetry group satisfying the dimension assumption \eqref{Eq:DimensionAssumption}, and for a $G$-invariant vertical complement $\W$ of $\C$,
the Jacobiator of the nonholonomic bivector  $\pi_{\emph \nh}$ is given by
$$
\frac{1}{2}[\pi_{\emph \nh}, \pi_{\emph \nh}] =  -\pi^\sharp_{\emph \nh} (   d{\mathcal J} \wedge \mathcal{K}_\subW ) - \psi_{\pi_{\emph{\nh}}}.
$$
\end{corollary}


\subsection{The Jacobiator of reduced brackets} \label{Ss:Jacob-Reduced}

For a nonholonomic system with symmetries satisfying the dimension assumption, we will now study the Jacobiator of reduced brackets on $\M/G$.

\subsubsection*{Reduced brackets}

Consider the $G$-invariant almost Poisson manifold $(\M, \pi_\nh)$. Note that
if a bivector $\pi_\B$ is gauge related to $\pi_\nh$ by a $G$-invariant 2-form $B$, then it is also $G$-invariant (see  \cite[Proposition~13]{PL2011}).
In this case, the orbit projection $\rho : \M \to \M/G$ induces almost Poisson bivector fields $\pi_\red^\nh$ and  $\pi_\red^\B$ on $\M/G$, where
\begin{equation} \label{Eq:ReducedBivector}
T\rho \circ \pi^\sharp_\nh(\rho^*\alpha) = ({\pi}_\red^\nh)^\sharp(\alpha) \qquad \mbox{and} \qquad
T\rho \circ \pi_\B^\sharp(\rho^*\alpha) = ({\pi}_\red^\B)^\sharp(\alpha), \qquad \mbox{for } \alpha \in
\Omega^1(\M/G).
\end{equation}

Equivalently, let $\{ \cdot, \cdot \}_\B$ be the almost Poisson bracket on $C^\infty(\M)$ associated to the invariant bivector $\pi_\B$. The induced almost Poisson bracket  $\{\cdot, \cdot \}^\B_{\mbox{\tiny red}}$ on functions on $\M/G$ is given
by
\begin{equation} \label{Eq:RedBracketDef} \{f \circ \rho, g \circ \rho \}_\B (m)= \{f,g
\}^\B_{\mbox{\tiny{red}}} (\rho(m)), \qquad \mbox{for } f,g \in C^\infty(\M/G), m \in \M.
\end{equation}

Let $\W$ be a $G$-invariant vertical complement of $\C$. From Theorem \ref{T:JacobiatorNH},
we have

\begin{corollary} \label{C:RedJacobiator} Let $\pi_\B$ be a bivector field on $\M$ gauge related to $\pi_{\emph \nh}$ by a $G$-invariant 2-form $B$ satisfying condition \eqref{Eq:Bsection}.
Then
$$
\frac{1}{2}[\pi^\B_{\emph \red}, \pi^\B_{\emph \red}] = -\pi_\B^\sharp (d B + d {\mathcal J} \wedge \mathcal{K}_\subW  ) \circ \rho^*.
$$
\end{corollary}

\begin{proof}
If $\pi^\B_\red$ is $\rho$-related to $\pi_\B$ then, from \eqref{Eq:ReducedBivector}, their Schouten brackets are $\rho$-related as well. That is, $$[\pi^\B_\red, \pi^\B_\red] = T\!\rho [\pi_\B, \pi_\B] \circ \rho^*.$$
\end{proof}

As observed in \cite{PL2011}, $\pi_\red^\nh$ and $\pi^\B_\red$ can have fundamentally different properties even though $\pi_\nh$ and $\pi_\B$ are gauge related. For instance, in general $\pi_\red^\nh$ and $\pi^\B_\red$  have different characteristic distributions (and thus are not gauge related). In the special case that the 2-form $B$ is basic (with respect to the orbit projection $\rho:\M \to \M/G$), then the reduced brackets $\pi_\red^\nh$ and $\pi^\B_\red$ are gauge related:

\begin{proposition} \label{P:diagramBbasic}
Let $P$ be a manifold equipped with an action of a Lie group $G$ that is free and proper. Let $\pi$ and $\pi_\B$ be $G$-invariant bivector fields on $P$, and denote by $\pi_{\emph\red}$ and $\pi_{\emph\red}^\B$ the induced bivectors on $P/G$.  If $\pi$ and $\pi_\B$ are gauge related by a basic 2-form $B$ (with respect to $\rho: P \to P/G$),  then $\pi_{\emph\red}$ and $\pi_{\emph\red}^\B$ are gauge related by the 2-form $B_{\emph\red}$ where $B_{\emph\red}$ is the 2-form on $P/G$ such that $\rho^*B_{\emph\red} = B$.
 \end{proposition}

\begin{proof}
Since $(\textup{Id} + B^\flat \circ \pi^\sharp)$ is invertible on $T^*P$ then, using the relations \eqref{Eq:ReducedBivector} and the fact that $B$ is basic, we obtain that $(\textup{Id} + B_\red^\flat \circ \pi_\red^\sharp)$ is invertible on $T^*(P/G)$.  By \eqref{Eq:IntroToGauge} we have that $\pi^\sharp = \pi_\B^\sharp \circ (\textup{Id} + B^\flat \circ \pi^\sharp)$ which is equivalent to
$\pi^\sharp \circ \rho^* = \pi_\B^\sharp \circ \rho^*(\textup{Id} + B_\red^\flat \circ \pi_\red^\sharp)$. Therefore, $\pi_\red^\sharp = (\pi_\red^\B)^\sharp \circ (\textup{Id} + B_\red^\flat \circ \pi_\red^\sharp)$.
\end{proof}

The goal of the remainder of this section is to find conditions guaranteeing that reduced almost Poisson brackets on $\M/G$ have integrable characteristic distributions. A class of almost Poisson brackets that has this integrability property (not necessarily satisfying the Jacobi identity) are the so-called {\it twisted Poisson} brackets. The appearance of these brackets in nonholonomic systems was initially  observed in \cite{PL2011}.

\subsubsection*{Reduced brackets and twisted Poisson structures}

A bivector field $\pi$ on a manifold $P$ is called a $\phi$-{\it twisted Poisson} \cite{SeveraWeinstein} if $\phi$ is a closed 3-form on $P$ such that
\begin{equation}
\label{E:twisted_Schout}
\frac{1}{2} [\pi, \pi] = \pi^\sharp(\phi).
\end{equation}

If $\{ \cdot, \cdot \}$ is the bracket defined by the $\phi$-twisted Poisson structure $\pi$, then (\ref{E:twisted_Schout}) becomes
\begin{equation*}
\{f, \{g,h\}\}+ \{g,\{h,f\}\} + \{h,\{f,g\}\} - \phi(X_f, X_g, X_h)=0, \label{twistedJAC}
\end{equation*}
for $f,g, h \in C^\infty(P)$ and $X_f = \{ \cdot , f\}$.

Let $\pi$ be a $\phi$-twisted Poisson bivector on the manifold $P$ and $[\cdot, \cdot]_\pi$ be the bracket on $T^*P$ given by (\ref{eq:LAPoisson}). Note that $\pi^\sharp$ does not preserve this bracket. However, using (\ref{E:not_Morph}) and (\ref{E:twisted_Schout}) we obtain
$$
[\pi^\sharp(\alpha),\pi^\sharp(\beta)] = \pi^\sharp \left( [\alpha,\beta]_\pi + {\bf i}_{\pi^\sharp(\alpha)\wedge \pi^\sharp(\beta)} \phi \right),
$$
for 1-forms $\alpha, \beta$ on $P$.  This induces the following modification of the bracket (\ref{eq:LAPoisson}):
\begin{equation*}
[\alpha, \beta]_\phi = \pounds_{\pi^\sharp(\alpha)} \beta - {\bf i}_{\pi^\sharp(\beta)} d\alpha + {\bf i}_{\pi^\sharp(\alpha)\wedge \pi^\sharp(\beta)} \phi, \label{eq:LAtwisted}
\end{equation*}
which makes $(T^*P, [\cdot, \cdot]_\phi, \pi^\sharp)$ into a Lie algebroid
\cite{SeveraWeinstein} (see also \cite[Sec.~2.2]{BCrainic}).
As a consequence, if $\pi$ is twisted Poisson, its characteristic distribution $\pi^\sharp(T^*P)$ is an integrable distribution (generally singular) since it is the image of the anchor map of a Lie algebroid.
Moreover, each leaf $\iota_{\mathcal O}: \mathcal{O} \hookrightarrow P$ of the corresponding foliation of $P$
is endowed with a non-degenerate 2-form $\Omega_{\mathcal O}$ that is not necessarily closed: if $\pi$ is $\phi$-twisted,
then  $\iota^*_{\mathcal O} \phi = d\Omega_{\mathcal O}$. Note also that if $\pi$ is $\phi$-twisted and $\phi'$ is a closed 3-form whose pullback to each leaf coincides with that of $\phi$, then $\pi$ is also $\phi'$-twisted.

\begin{example}
\
\begin{enumerate}
\item[$(i)$] If $\pi$ is the bivector field on $P$ defined by a nondegenerate 2-form  $\Omega$, then $\pi$ is $\phi$-twisted and $\phi=d\Omega$, see \eqref{Eq:Jacdomega}.
\item[$(ii)$]  Any bivector field $\pi$ on a manifold $P$ with an integrable regular characteristic distribution is twisted Poisson, even though there is no canonical choice of closed 3-form $\phi$, \cite{PL2011}.
\item[$(iii)$]  Let $\pi$ be a $\phi$-twisted Poisson bivector field on $P$. If
$\pi_\B$ is a bivector field obtained from $\pi$ by a gauge transformation by
a 2-form $B$,  then $\pi_\B$ is $(\phi-dB)$-twisted.
In particular, if $\pi$ is Poisson, then $\pi_\B$ is $(-dB)$-twisted.
\end{enumerate}
\end{example}

Coming back to our setting, observe that the almost Poisson manifold
$(\M, \pi_\nh)$ associated with a nonholonomic system is never
twisted Poisson since its characteristic distribution $\C$ is not
integrable. Moreover, since all bivector fields gauge related to
$\pi_\nh$ have $\C$ as characteristic distribution, none of them is
twisted Poisson.

Suppose that $G$ is a symmetry group for the nonholonomic system satisfying the dimension assumption \eqref{Eq:DimensionAssumption}. Consider the corresponding $G$-invariant almost Poisson manifold $(\M, \pi_\nh)$.
The next result gives sufficient conditions for reduced brackets to be twisted Poisson based on  Corollary \ref{C:RedJacobiator}.

\begin{corollary} \label{C:Twisted} Let $\W$ be a $G$-invariant complement of $\C$ in $T\M$ such that the 3-form $ d {\mathcal J} \wedge \mathcal{K}_\subW$ is closed.
Let $B$ be a $G$-invariant 2-form on $\M$ satisfying condition \eqref{Eq:Bsection} and defining a bivector field $\pi_\B$ gauge related to $\pi_{\emph\nh}$.
If  $d B + d {\mathcal J} \wedge \mathcal{K}_\subW $ is basic with respect to $\rho: \M \to \M/G$
then the induced reduced bivector field $\pi_{\emph \red}^\B$ is $\phi$-twisted Poisson for the 3-form $\phi$ on $\M/G$ defined by the condition
$$
\rho^* \phi = -( d B + d {\mathcal J} \wedge \mathcal{K}_\subW).
$$
\end{corollary}

In particular, if we now consider the $G$-invariant hamiltonian $\Ham_\M: \M\to \mathbb{R}$ and if the 2-form $B$ in Corollary \ref{C:Twisted} defines a {\it dynamical} gauge transformation (as in Def.~\ref{D:dynamicalGauge}), then the reduced nonholonomic vector field $X^\nh_\red$ can be written as
$$
({\pi}_\red^\B)^\sharp(d \Ham_\red) = - X^\nh_\red,
$$
where $\Ham_\red:\M/G \to \R$ is the reduced hamiltonian function ( i.e, $\Ham_\subM = \rho^*\Ham_\red$). So, in this case, the reduced dynamics is described by a twisted Poisson bracket. Note that we are not assuming any regularity condition on $\pi^\B_\red$.

\subsection{A special class of symmetries}\label{sebsec:verticalcond}

We now show how many of the results discussed so far simplify once
we impose an additional condition on the $G$-invariant complement
$\W$ of $\C$.

We say that $\W$ satisfies the {\it vertical-symmetry condition} if
there is a Lie subalgebra $\mathfrak{g}_\subW \subseteq
\mathfrak{g}$ such that, for all $m\in \M$,
\begin{equation}\label{Eq:VerticalSymmetries}
\W_m =  \{\xi_\M(m)\,|\, \xi\in \mathfrak{g}_\subW \}.
\end{equation}

Observe that we only need to require $\mathfrak{g}_\W$ to be a
vector subspace since in this case it automatically satisfies that
$[\mathfrak{g},\mathfrak{g}_\W]\subset \mathfrak{g}_\W$ by the
$G$-invariance of $\W$.

The  vertical-symmetry condition is a restrictive condition on the
symmetries, as it implies the existence of a subgroup $G_\subW$ of
$G$ for which the system is Chaplygin ($G_\subW$ does not need to be embeded in $G$). But, as we will see, many
non-Chaplygin systems satisfy this condition, including  the
vertical rolling disk (Section \ref{Ex:VerticalDisk}), the
nonholonomic particle (Sections \ref{Ex:NHParticle},
\ref{Ex:NHParticle2}), the snakeboard (Section \ref{Ex:skateboard})
and the rigid bodies with generalized constraints (Section
\ref{Ex:RigidBody})).

\begin{remark} Note that systems satisfying the vertical-symmetry condition
can be seen as Chaplygin systems with additional symmetries, so one
can in principle treat them in two steps: first taking care of the
Chaplygin symmetries, then treating the extra ones (along the lines
of e.g. \cite{Hoch,Oscar}). Comparing this viewpoint with the
direct, one-step approach is the subject of a separate work in
preparation \cite{PO}.
\end{remark}

The next result shows some consequences of \eqref{Eq:VerticalSymmetries}.

\begin{proposition} \label{Prop:Prop1_VertSymm} For any $G$-invariant vertical  complement $\W$ of $\C$ in $T\M$ satisfing the vertical-symmetry condition \eqref{Eq:VerticalSymmetries}, the $\W$-curvature $\mathcal{K}_\subW$ verifies
\begin{enumerate}
\item[$(i)$] $\mathcal{K}_\subW (X,Y)= d \mathcal{A}_\subW(X,Y),  \;\; \mbox{ for all } X \in \Gamma(\C), Y \in \mathfrak{X}(\M).$
\item[$(ii)$] $d \mathcal{K}_\subW (X,Y,Z)=0$,\; for all $X,Y,Z \in \Gamma(\C).$
\item[$(iii)$] $d\mathcal{J} \wedge {\mathcal K}_\subW = d \langle {\mathcal J} , \mathcal{K}_\subW \rangle$.
\end{enumerate}
\end{proposition}

\begin{proof}

$(i)$: If $X,Y \in \Gamma(\C)$, then  $\mathcal{K}_\subW (X,Y)= d \mathcal{A}_\subW(X,Y)$ by the definition of $\mathcal{K}_\subW$.
If $X\in \Gamma(\C)$ and $Y \in \Gamma(\mathcal{W})$, then $\mathcal{K}_\subW(X,Y)=0$. In order to check that $d
\mathcal{A}_\subW(X,Y)=0$,  we write $Y = f_i \eta^i_\M$ where $f_i \in C^\infty(\M)$ and $\eta^i \in \mathfrak{g}_\subW$. Then,
using that $[X,\eta^i_\subM] \in \Gamma(\C)$ (because $\C$ is
$G$-invariant) and that $\mathcal{A}_\subW(\eta^i_\subM)=\eta^i$, we get
\begin{equation*}
d\mathcal{A}_\subW(X,Y)(m) = d\mathcal{A}_\subW(X, f_i \eta^i_\subM) =
f_i \left( X(\eta^i)- \mathcal{A}_\subW([X,\eta^i_\subM]) \right) =0.
\end{equation*}

$(ii)$: Let $X,Y,Z \in \Gamma(\C)$. Then using item $(i)$ we obtain
\begin{eqnarray*}
d \mathcal{K}_\subW (X,Y,Z)\! & = &\! cycl [ \, Z(\mathcal{K}_\subW (X,Y)) - \mathcal{K}_\subW([X,Y],Z) \, ] 
= cycl [\, - Z(\mathcal{A}_\subW ([X,Y])) - d \mathcal{A}_\subW ([X,Y],Z) \, ] \\
&=& \! cycl [\, - Z(\mathcal{A}_\subW ([X,Y])) + Z(\mathcal{A}_\subW ([X,Y])) + \mathcal{A}_\subW ([[X,Y],Z]) \, ]=0.
\end{eqnarray*}

$(iii)$: It follows from \eqref{Eq:dJ^K} and $d{\mathcal K}_\subW |_\C \equiv 0$.
\end{proof}

We can simplify
Theorem \ref{T:JacobiatorNH} and Corollary \ref{C:RedJacobiator} when $\W$ is chosen such that condition \eqref{Eq:VerticalSymmetries} holds.

\begin{theorem} \label{T:JacobVerticalSymm} Let $(\M,\pi_{\emph\nh})$ be the $G$-invariant almost Poisson manifold associated to a nonholonomic system with symmetries satisfying the dimension assumption \eqref{Eq:DimensionAssumption}.
Let $\W$ be a $G$-invariant vertical complement of $\C$ satisfying the vertical-symmetry condition \eqref{Eq:VerticalSymmetries}. Then:
\begin{enumerate}
\item[$(i)$]  the Jacobiator of a bivector field $\pi_\B$ on $\M$ gauge related to $\pi_{\emph \nh}$ by a 2-form $B$ satisfying condition \eqref{Eq:Bsection} is given by
$$
\frac{1}{2}[\pi_\B, \pi_\B] = -\pi_\B^\sharp (d B +  d\langle{\mathcal J} , \mathcal{K}_\subW \rangle ) - \psi_{\pi_\B}.
$$
\item[$(ii)$] If, in addition, $B$ is $G$-invariant, then the Jacobiator of the reduced bivector field $\pi_{\emph\red}^\B$ is
$$
\frac{1}{2}[\pi^\B_{\emph \red}, \pi^\B_{\emph \red}] = -\pi_\B^\sharp (d B + d \langle {\mathcal J} ,\mathcal{K}_\subW \rangle ) \circ \rho^*.
$$
\end{enumerate}
\end{theorem}

\begin{corollary} \label{C:TwistedVertSymm}
Under the same hypotheses of Theorem \ref{T:JacobVerticalSymm}~$(ii)$, if $B$ is such that $dB + d\langle {\mathcal J}, {\mathcal K}_\subW \rangle$ is basic, then $\pi_{\emph\red}^\B$ is a $\phi$-twisted Poisson bivector on $\M/G$ where $\phi$ is the 3-form on $\M/G$ such that $\rho^*\phi = -(dB + d\langle {\mathcal J}, {\mathcal K}_\subW \rangle)$.
\end{corollary}

\bigskip

\begin{example}[The nonholonomic particle]
\label{Ex:NHParticle}

Consider a particle in $Q=\R^3$ subjected only to the constraint
$\epsilon=0$, where
$$
\epsilon = dz-ydx,
$$
with lagrangian $\Lag=\frac{1}{2}(\dot x^2+\dot y^2+\dot z^2)$.
The system admits symmetries given by the action of the Lie
group $\R^2$ on $Q$ such that the vertical space is
$V= \textup{span} \left\{ \frac{\partial}{\partial x}, \frac{\partial}{\partial z}
\right\}$.
The non-integrable distribution $D$ on $Q$ is given by $D = \textup{span} \left\{ \frac{\partial}{\partial x}+ y
\frac{\partial}{\partial z} , \frac{\partial}{\partial y} \right\}$ and we choose the complement $W$ of $D$ to be $W=\textup{span} \left\{
\frac{\partial}{\partial z} \right\}$. Observe that $W$ is not equal to $S^\perp \cap V$.
In canonical coordinates, the constraint manifold is
$$
\M = \{(x,y,z;p_x,p_y,p_z) \ : \ p_z=yp_x\}.
$$
From here we see that
\begin{equation} \label{Ex:NHP_C+W}
\C = \textup{span}  \left\{ \frac{\partial}{\partial x}+ y
\frac{\partial}{\partial z} , \frac{\partial}{\partial y}, \frac{\partial}{\partial p_x}, \frac{\partial}{\partial p_y}  \right\} \qquad \mbox{and} \qquad \W=\textup{span} \left\{
\frac{\partial}{\partial z} \right\}.
\end{equation}

Note that the vertical-symmetry condition \eqref{Eq:VerticalSymmetries} is satisfied, for
\begin{equation} \label{Ex:NHPart-LieSubspaces}
\mathfrak{g}_\subW = \textup{span} \{ ( 0,1) \} \subset \mathfrak{g}=\mathbb{R}^2.
\end{equation}

\begin{remark}\label{Rem:horiz}
Note that the subbundle $\mathfrak{g}_\subS := \{(m,\xi)\,:\, \xi_\subM(m)\in \S_m\}\subset \M\times \mathfrak{g}$ in this case is given by  $\mathfrak{g}_\subS |_m = \textup{span} \{(1 , y) \}$, for $m \in \M$;
so the system does not admit horizontal symmetries in the sense of \cite{BKMM}. So, the vertical-symmetry condition may hold without horizontal symmetries.
\end{remark}

The projection to $\W$ adapted to the decomposition \eqref{Ex:NHP_C+W}  is given by $P_\subW (v_m) = \epsilon(v_m)\frac{\partial}{\partial z}$, for $v_m \in T_m\M$, where we consider $\{dx, dy, dp_x,dp_y, \epsilon\}$  a local basis of $T^*\M$.
Therefore $\mathcal{A}_\subW: T\M \to \mathfrak{g}$ is given by $\mathcal{A}_\subW = (0, \epsilon)$ (see \eqref{Eq:AkinPw}).
Since $\mathfrak{g}_\subW$ is a vector space, the $\W$-curvature is $\mathcal{K}_\subW = (0, d \epsilon) = (0, dx\wedge dy)$.
To compute the momentum map ${\mathcal J}:\M \to \mathfrak{g}^*$ we use \eqref{Eq:Def-MomentumMapOnM} for the canonical 1-form $\Theta_\subM =p_xdx+p_ydy + yp_xdz$ and $(1,y)_\subM= \frac{\partial}{\partial x}+ y
\frac{\partial}{\partial z}$  and $(0,1)_\subM= \frac{\partial}{\partial z}$. Therefore, in canonical coordinates,
${\mathcal J}= ((1+y^2) p_x, yp_x).$
The 2-form $\langle {\mathcal J}, \mathcal{K}_\subW \rangle$ is given by
\begin{equation} \label{Eq:JK-NHParticle}
\langle {\mathcal J},
\mathcal{K}_\subW\rangle = y p_x dx\wedge dy.
\end{equation}

In this case, we can explicitly check that $d{\mathcal J} \wedge {\mathcal K}_\subW = d\langle {\mathcal J},
\mathcal{K}_\subW\rangle$.  By Theorem \ref{T:JacobiatorNH} or \ref{T:JacobVerticalSymm} we have the ingredients to write the Jacobiator for $\pi_\nh$ (and for gauge related bivector fields $\pi_\B$).  More interesting is analyzing the Jacobiator of the reduced bivector field $\pi_\red^\nh$ on the reduced manifold $\M/G \simeq \R^3$ (with local coordinates $(y,p_x,p_y)$). Observe that $\langle {\mathcal J}, \mathcal{K}_\subW\rangle$ is
not basic (not even $d \langle {\mathcal J}, \mathcal{K}_\subW\rangle$ is basic)
but
\begin{equation} \label{Eq:dJK-NHParticle}
d \langle {\mathcal J}, \mathcal{K}_\subW\rangle (
\pi_\nh^\sharp(\rho^*dy), \pi_\nh^\sharp(\rho^*dp_x),
\pi_\nh^\sharp(\rho^*dp_y) ) = 0.
\end{equation}
By Theorem~\ref{T:JacobVerticalSymm}$(ii)$, we conclude that $[\pi^\nh_\red, \pi^\nh_\red]=0$, i.e., the reduced nonholonomic bracket is Poisson. Equation \eqref{Eq:dJK-NHParticle} can be easily checked from the fact that
$$
\pi^\sharp_\nh(dy)=\frac{\partial}{\partial p_y}, \quad \pi^\sharp_\nh(dp_x)= -\frac{1}{y^2+1} \left(\frac{\partial}{\partial x} + y \frac{\partial}{\partial z} \right) - \frac{yp_x}{y^2+1} \frac{\partial}{\partial p_y}, \quad  \pi^\sharp_\nh(dp_y)=  - \frac{\partial}{\partial y} + \frac{yp_x}{y^2+1} \frac{\partial}{\partial p_x}.
$$

\end{example}

More examples where $d\langle {\mathcal J}, \mathcal{K}_\subW \rangle$ is basic, or where we need a dynamical gauge in order to obtain $d(\langle {\mathcal J}, \mathcal{K}_\subW \rangle +B)$ basic, are worked out in Section \ref{S:Examples}.


\begin{example}[Chaplygin systems] \label{Ex:ChapSystems} 
 It is well known that a Chaplygin system \eqref{Eq:Chaplygin} is defined by a nondegenerate 2-form $\Omega_\red^\nh$ on $\M/G$. Using Remark \ref{R:KinChaplygin} and recalling that in this case the 2-form $\langle \mathcal{J},\mathcal{K} \rangle$ is basic,   we observe that Corollary \ref{C:TwistedVertSymm} (for $B=0$) recovers the result given in \cite{BS93} by just applying  \eqref{Eq:Jacdomega}:
 \begin{equation}\label{Eq:BS}
 d\Omega_{\red}^{\nh}  =  -d \langle {\mathcal J}, \mathcal{K} \rangle_{\red}, 
 \end{equation}
where $\langle {\mathcal J}, \mathcal{K} \rangle_{\red}$ is the 2-form on $\M/G$ such that $\rho^*(\langle {\mathcal J}, \mathcal{K} \rangle_{\red}) = \langle {\mathcal J}, \mathcal{K} \rangle$.

One also has a more general statement including gauge transformations.
Since $\W=\V$, condition \eqref{Eq:Bsection} reduces to requiring that $B$ is semi-basic with respect to $\M \to \M/G$.
So since $B$ is $G$-invariant, it must be basic, and thus there is a 2-form $B_\red$ on $\M/G$ such that $\rho^*B_\red = B$.
Now, if $\pi_\B$ is gauge related to $\pi_\nh$ by a basic 2-form $B$,
then (by Proposition~\ref{P:diagramBbasic}) their reductions
$\pi_\red^\nh$ and $\pi^\B_\red$ are gauge related by $B_\red$.
Therefore there is a nondegenerate 2-form $\Omega^\B_\red$ defining $\pi^\B_\red$, such that 
$\Omega^\B_\red = \Omega_\red^\nh - B_\red$ (see
Remark~\ref{R:GaugeT}$(iii)$).  Hence, $d\Omega^\B_\red  =  -d (\langle {\mathcal J}, \mathcal{K} \rangle_\red + B_\red)$ as Corollary~\ref{C:TwistedVertSymm} shows. 

\end{example}

\section{Poisson structures associated with some nonholonomic systems with symmetries} \label{S:Poisson}

For a given nonholonomic system with an action by symmetries
satisfying the dimension assumption, this section shows how the
choice of a $G$-invariant vertical complement $\W$ of $\C$
satisfying the vertical-symmetry condition endows the reduced
manifold $\M /G$ with a genuine Poisson structure $\Lambda$. We will also
discuss conditions under which the reduced nonholonomic dynamics is
described by a bracket obtained from $\Lambda$ by a gauge transformation.
We illustrate our results showing that one recovers known facts about Chaplygin systems; in this case $\Lambda$ is the canonical 2-form on $T^*Q$.

\subsection{Poisson structures on reduced spaces}
\label{Ss:ViaJK}

Let $(\M, \pi_\nh)$ be the almost Poisson manifold associated with a nonholonomic system equipped with a symmetry group $G$ such that the dimension assumption is satisfied. We start by observing

\begin{proposition} \label{L:JKbasic}
For any $G$-invariant vertical complement $\W$ of $\C$, the 2-form $\langle {\mathcal J}, {\mathcal K}_\subW \rangle$ on $\M$ satisfies the following properties:
\begin{enumerate}
\item[$(i)$] it is semi-basic with respect to the bundle projection $\tau: \M \to Q$,
\item[$(ii)$] it is $G$-invariant.
\end{enumerate}
\end{proposition}
\begin{proof}
Statement $(i)$ relies on \eqref{Eq:K-Kkin} and on the fact that ${\bf K}_\subW$ is semi-basic, which was proved in Proposition \ref{Prop:Ksemi-basic}.

To prove $(ii)$, recall that $\mathcal{K}_\subW$ is ad-equivariant (see Prop.~\ref{Prop:Prop1}).
By the equivariance property of the canonical momentum map ${\mathcal J}: \M \to \mathfrak{g}^*$, we obtain
$$
\pounds_{\xi_\subM} \langle {\mathcal J},  \mathcal{K}_\subW \rangle = \langle \pounds_{\xi_\subM} {\mathcal J}, \mathcal{K}_\subW \rangle + \langle {\mathcal J}, \pounds_{\xi_\subM} \mathcal{K}_\subW \rangle = \langle - ad_{\xi}^* {\mathcal J}, \mathcal{K}_\subW \rangle + \langle {\mathcal J}, ad_\xi \mathcal{K}_\subW \rangle = 0.
$$
\end{proof}

As a consequence of the previous Proposition and Remark \ref{R:DynGaugeSemi-basic},
if $B:=  -\langle {\mathcal J}, \mathcal{K}_\subW \rangle$, the
endomorphism $(\mathrm{Id} + B^\flat\circ \pi_\nh^\sharp)$ of $T^*\M$ is invertible. Therefore, the gauge transformation of $\pi_\nh$ by the 2-form $-\langle {\mathcal J}, \mathcal{K}_\subW \rangle$ is well defined and produces a new bivector field on $\M$ that we denote by $\pi_{\mbox{\tiny ${\mathcal J}\! \mathcal{K}$  }}$.

Equivalently, we can describe $\pi_{\mbox{\tiny ${\mathcal J}\! \mathcal{K}$  }}$ by noticing (see Remark~\ref{R:GaugeT}$(ii)$) that
the point-wise restriction of the 2-form
$$
\Omega_{\mbox{\tiny ${\mathcal J}\! \mathcal{K}$}} := \Omega_\subM + \langle {\mathcal J}, \mathcal{K}_\subW \rangle
$$
to $\C$ is nondegenerate. Then  $\pi_{\mbox{\tiny ${\mathcal J}\! \mathcal{K}$  }}$ is the bivector determined by  the distribution $\C$ and the section $\Omega_{\mbox{\tiny ${\mathcal J}\! \mathcal{K}$}}|_\C$ (as in \eqref{Eq:bivector-section}).  Note also that it may be the case that  $\pi_{\mbox{\tiny ${\mathcal J}\! \mathcal{K}$   }}$ does not describe the dynamics since ${\bf i}_{X_\nh} \langle {\mathcal J}, \mathcal{K}_\subW  \rangle$ usually does not vanish.

\begin{proposition}\label{P:poisson1}
Assume that the $G$-invariant vertical complement $\W$ of $\C$ satisfies the vertical-symmetry condition \eqref{Eq:VerticalSymmetries}. Then  $\pi_{\mbox{\tiny ${\mathcal J}\! \mathcal{K}$   }}$ is a $G$-invariant bivector field satisfying
$$
\frac{1}{2}[\pi_{\mbox{\tiny ${\mathcal J}\! \mathcal{K}$   }}, \pi_{\mbox{\tiny ${\mathcal J}\! \mathcal{K}$   }}]= - \Psi_{\pi_{\mbox{\tiny ${\mathcal J}\! \mathcal{K}$   }}},
$$
and its reduction $\pi^{\mbox{\tiny ${\mathcal J}\! \mathcal{K}$   }}_{\emph\red}$
defines a Poisson structure on $\M / G$.
\end{proposition}

\begin{proof}
The $G$-invariance of $\pi_{\mbox{\tiny ${\mathcal J}\! \mathcal{K}$   }}$ follows from the $G$-invariance of $\langle {\mathcal J}, \mathcal{K}_\subW  \rangle$, see Prop.~\ref{L:JKbasic}$(ii)$.
The formula for $[\pi_{\mbox{\tiny ${\mathcal J}\! \mathcal{K}$   }}, \pi_{\mbox{\tiny ${\mathcal J}\! \mathcal{K}$   }}]$ follows from  Theorem \ref{T:JacobVerticalSymm}$(i)$, once we set $B=- \langle {\mathcal J}, \mathcal{K}_\subW \rangle$.
The fact that $\pi^{\mbox{\tiny ${\mathcal J}\! \mathcal{K}$   }}_\red$
is a Poisson structure is an immediate consequence of
Corollary~\ref{C:TwistedVertSymm}.
\end{proof}

To simplify our notation, we will denote the Poisson structure on $\M / G$ defined in Prop.~\ref{P:poisson1} by $\Lambda$,
\begin{equation}\label{Eq:PoissonW}
\Lambda:=  \pi^{\mbox{\tiny ${\mathcal J}\! \mathcal{K}$   }}_\red,
\end{equation}
assuming that the choice of $\W$ is clear from the context.

\begin{example}[The nonholonomic particle]\label{Ex:NHParticle2}
Recall Example \ref{Ex:NHParticle}. The nonholonomic bivector $\pi_\nh$ is given
by $$\pi_\nh= \frac{1}{y^2+1} \left( \frac{\partial}{\partial x} +
y\frac{\partial}{\partial z} \right) \wedge \frac{\partial}{\partial
p_x} + \frac{\partial}{\partial y} \wedge \frac{\partial}{\partial
p_y} - \frac{yp_x}{1+y^2} \frac{\partial}{\partial p_x} \wedge
\frac{\partial}{\partial p_y}.$$ Observe the similarity with the
bivector in \eqref{Ex:gauge}. The gauge transformation of $\pi_\nh$
by the 2-form $-\langle {\mathcal J}, \mathcal{K}_\subW \rangle$
given in \eqref{Eq:JK-NHParticle},  is
$$
\pi_{\mbox{\tiny ${\mathcal J}\! \mathcal{K}$   }} = \frac{1}{y^2+1} \left( \frac{\partial}{\partial x} + y\frac{\partial}{\partial z} \right) \wedge \frac{\partial}{\partial p_x} + \frac{\partial}{\partial y} \wedge \frac{\partial}{\partial p_y}  - 2\frac{yp_x}{1+y^2} \frac{\partial}{\partial p_x} \wedge \frac{\partial}{\partial p_y}.$$ The orbit projection $\rho:\M \to \M/G$ is given by $\rho(x,y,z;p_x,p_y) = (y,p_x,p_y)$ and thus the Poisson bivector field on $\M/G$ is
\begin{equation} \label{Ex:NHPLambda}
\displaystyle{\Lambda =  \frac{\partial}{\partial y} \wedge \frac{\partial}{\partial p_y} - 2 \frac{yp_x}{1+y^2} \frac{\partial}{\partial p_x} \wedge \frac{\partial}{\partial p_y}}.
\end{equation}
\end{example}



\subsection{A metric-independent Poisson manifold}
\label{Ss:Poisson-Indep}


In this section we only use part of the data of a nonholonomic
system with symmetries, since the lagrangian will play no role.

Consider a (free and proper) $G$-action on the configuration manifold
$Q$ preserving the constraint distribution $D$, and let $W$ be a
$G$-invariant vertical complement of $D$ in $TQ$, that is, $TQ = D
\oplus W$ and $W \subset V$.

Analogously to the condition introduced in
Section~\ref{sebsec:verticalcond}, we will say that $W$ satisfies
the {\it vertical-symmetry condition} if there is a Lie subalgebra
$\mathfrak{g}_{\mbox{\tiny{$W$}}}$
of $\mathfrak{g}$ such that, for all $q\in
Q$,
\begin{equation} \label{Eq:VerticalSymm-Q}
W_q = \{\eta_Q(q)\,|\, \eta \in \mathfrak{g}_{\mbox{\tiny{$W$}}}\}.
\end{equation}

Since $W$ is $G$-invariant, the same holds for its annihilator $W^\circ \subset T^*Q$ with respect to the lifted action to $T^*Q$.
We will now show that $W^\circ /G$ inherits a Poisson structure $\Lambda_0$ when $W$ satisfies the vertical-symmetry condition. Later, in Section \ref{Ss:TheEquivalence}, we will
relate the Poisson manifold $(W^\circ /G,\Lambda_0)$ with the reduced Poisson structures of Prop.~\ref{P:poisson1}.

Let $\V_0 \subset T(W^\circ)$ be the vertical space with respect to the
$G$-action on $W^\circ$ (the restriction of the cotangent lift), and
denote by $\tau_0 : W^\circ \to Q$ the canonical projection (the
restriction of $\tau_Q: T^*Q\to Q$).  As in Section \ref{Ss:nh-systems}, let us consider the distribution $\C_0$  on $W^\circ$ defined by
\begin{equation} \label{Eq:C0}
\C_0 :=\{ v \in T(W^\circ) \ : \ T\tau_0 (v) \in D \}. 
\end{equation}
Observe that $\C_0$ is $G$-invariant, since $D$ is $G$-invariant.

Consider $\Omega_{\mbox{\tiny{$W^{\!\circ}$}}} := \iota_0^*
\Omega_Q$, where $\iota_0: W^\circ \hookrightarrow T^*Q$ is the
natural inclusion.

\begin{lemma}\label{L:nondeg}
The restriction $\Omega_{\mbox{\tiny{$W^{\!\circ}$}}} |_{\C_0}$ is
nondegenerate.
\end{lemma}

\begin{proof}
The lemma is a direct consequence of the result in \cite[Sec.~5]{BS93}
recalled in Section \ref{Ss:nh-systems}: if we consider any metric
$\kappa_0$ for which $D$ and $W$ are orthogonal, then we have $\M=
\kappa_0^\sharp(D) = W^\circ$, and the induced constraint
distribution on $\M$ is $\C_0$.
\end{proof}

Therefore, as seen in Section \ref{Ss:nh-systems},  the distribution $\C_0$, given in \eqref{Eq:C0}, and the 2-form $\Omega_{\mbox{\tiny{$W^{\!\circ}$}}}$ induce a $G$-invariant bivector field $\pi_0$ on $W^\circ$
defined by the relation
\begin{equation} \label{eq:defPiC0}
{\bf i}_X \Omega_{\mbox{\tiny{$W^{\!\circ}$}}} |_{\C_0} = \alpha
|_{\C_0} \qquad \Leftrightarrow \qquad \pi_{0}^\sharp(\alpha) = - X
\qquad \mbox{at each } \alpha \in T^*(W^\circ).
\end{equation}
We denote the reduction of $\pi_0$ to $W^\circ/G$ by $\Lambda_0$.

Then, from Remark
\ref{R:Split-on-MandQ}, we know that the decomposition $TQ=D \oplus
W$ induces a splitting $T(W^\circ) = \C_0 \oplus \W_0,$ where
$
\mathcal{W}_0 := (T\tau_0 |_{\mathcal{V}_0} )^{-1} (W).
$
In this case, since $W$
is $G$-invariant then $\W_0$ is $G$-invariant and also note that if $W$ verifies the vertical-symmetry condition, then  $\W_0$ verifies it as well.

\begin{proposition} \label{P:W0isPoisson}
If $W$ is a $G$-invariant vertical complement of $D$ satisfying the
vertical-symmetry condition, then  $\Lambda_0$, the reduction of
$\pi_0$ on $W^\circ$ to $W^\circ/G$, is a Poisson structure.
\end{proposition}

\begin{proof}
By Theorem \ref{T:JacobVerticalSymm}, we know that 
$$
\frac{1}{2}[\pi_{0} ,  \pi_{0}]  = - \pi_{0}^\sharp ( d \langle {\mathcal J}_0, \mathcal{K}_{\subW_0} \rangle ) - \psi_{ \pi_{0}},
$$
where ${\mathcal J}_0:W^\circ \to \mathfrak{g}^*$ is restriction of the canonical momentum map on $T^*Q$ to $W^\circ$ and $\mathcal{K}_{\subW_0}$ is the $\W_0$-curvature with respect to the decomposition $T(W^\circ)=\C_0 \oplus \W_0$. We are going to show that $\langle {\mathcal J}_0, \mathcal{K}_{\subW_0}  \rangle \equiv 0$, for which it suffices to check that $\langle {\mathcal J}_0, \mathcal{K}_{\subW_0}  \rangle|_{\C_0} \equiv 0$.
Consider (local) basis $\{X_i, Z_a\}$ of $TQ$ adapted to the
decomposition $D \oplus W$. If $(q)$ are local
coordinates on $Q$, then we denote by $(q; p_i, p_a)$ the  coordinates on
$T^*Q$ associated to the dual basis $\{X^i, Z^a\}$ of $T^*Q$. The submanifold $W^\circ$ is represented by $(q; p_i, 0)$. A local basis of $T^*W^\circ$ is given by $\{X^i, Z^a, dp_i\}$ (here we use the same notation for 1-forms on $Q$ and their
pullbacks to $T^*Q$ and $W^\circ$) and the dual (local) basis of $TW^\circ$ is given by $\{X_i, Z_a, \frac{\partial}{\partial p_i}\}$.   In order to compute the 2-form $\langle
{\mathcal J}_0, \mathcal{K}_{\subW_0} \rangle$ observe that $\W_0 = \textup{span}\{Z_a\}$ and thus for $X,
Y \in \Gamma(\C_0)$,  $P_{\subW_0} ([X, Y]) = {Z}^a([X,Y]). Z_a = -
d{Z}^a(X,Y). Z_a $. If $\xi^a$ is a local basis of  sections of
$\mathfrak{g}_{\mbox{\tiny{$W$}}} \times Q$ such that $(\xi^a)_Q =
Z_a $ then the coordinates of $\mathcal{K}_{\subW_0}(X,Y)$ are
$(\mathcal{K}_{\subW_0}  (X,Y) )_a = d {Z}^a (X,Y)$.  On the other
hand, note that the Liouville 1-form on $T^*Q$ can be written as
\begin{equation}\label{Eq:Liouv}
\Theta_Q|_{(q; p_i,p_a)}= p_iX^i+ p_aZ^a.
\end{equation}
 So $\langle \mathcal{J}_0, \xi^a\rangle =
\iota_0^*p_a$. Hence $ \langle {\mathcal J}_0, \mathcal{K}_{\subW_0}
\rangle |_{\C_0} = (\iota^*_0 p_a) d {Z}^a |_{\C_0} \equiv 0$, since
$\iota^*_0 p_a=0$. Therefore, applying Theorem
\ref{T:JacobVerticalSymm}$(ii)$, it is clear that the reduction of
$\pi_0$ is Poisson.
\end{proof}

We now use the results in the Appendix to give an alternative
characterization of the Poisson structure $\Lambda_0$ on
$W^\circ/G$. Consider the 2-form $\Omega_{\mbox{\tiny{$W^{\!\circ}$}}}$ on $W^\circ$ and
the distribution $\W_0$ as in \eqref{Eq:C0}.

\begin{lemma}\label{L:kernel}
If $W$ satisfies the vertical-symmetry condition, then $
\textup{Ker}\, \Omega_{\mbox{\tiny{$W^{\!\circ}$}}} = \W_0$.
\end{lemma}

\begin{proof}
For $v \in {\W_0}|_{m}$, $m \in W^\circ$, we have that
$\Theta_{\mbox{\tiny{$W^{\!\circ}$}}} (v)(m) = m(T\tau_0(v)) = 0$,
since $T\tau_0(v)\in W$. Hence $\W_0 \subset \textup{Ker}\,
\Theta_{\mbox{\tiny{$W^{\!\circ}$}}}$. To see that $\W_0 \subset
\textup{Ker}\, \Omega_{\mbox{\tiny{$W^{\!\circ}$}}}$, it is
sufficient to check that ${\bf i}_{\eta_{\scriptscriptstyle{\W_0}}}
\Omega_{\mbox{\tiny{$W^{\!\circ}$}}} = 0$ for all $\eta \in
\mathfrak{g}_{\mbox{\tiny{$W$}}}$, where
$\eta_{\scriptscriptstyle{\W_0}}$ is the infinitesimal generator
corresponding to $\eta$. Then
$$
{\bf i}_{\eta_{\mbox{\tiny{$\W_0$}}}}
\Omega_{\mbox{\tiny{$W^{\!\circ}$}}} = - {\bf
i}_{\eta_{\scriptscriptstyle{\W_0}}} d
\Theta_{\mbox{\tiny{$W^{\!\circ}$}}} = -
\pounds_{\eta_{\scriptscriptstyle{\W_0}}}
\Theta_{\mbox{\tiny{$W^{\!\circ}$}}} + d {\bf
i}_{\eta_{\scriptscriptstyle{\W_0}}}
\Theta_{\mbox{\tiny{$W^{\!\circ}$}}} = 0
$$
using the $G$-invariance of $\Theta_{\mbox{\tiny{$W^\circ$}}}$. As
remarked in Lemma~\ref{L:nondeg}, $\Omega_{\mbox{\tiny{$W^\circ$}}}
|_{\C_0}$ is nondegenerate, hence $\W_0 = \textup{Ker}\,
\Omega_{\mbox{\tiny{$W^{\!\circ}$}}}$.
\end{proof}

It follows from Lemma~\ref{L:A-PresymplecticRed}$(i)$ in the
Appendix that the presymplectic submanifold $(W^\circ,
\Omega_{\mbox{\tiny{$W^{\!\circ}$}}})$ of $(T^*Q,\Omega_Q)$ reduces
to a Poisson structure on $W^\circ / G$. The fact that this Poisson
structure coincides with $\Lambda_0$ of Prop.~\ref{P:W0isPoisson} is
a direct consequence of Lemma~\ref{L:B-PresymAlmPoisson}:

\begin{proposition} \label{P:W0-Poisson2}
Let $W$ be a $G$-invariant vertical complement of $D$ in $TQ$
satisfing the vertical-symmetry condition. Then the Poisson
structure on $W^\circ / G$ obtained by reduction of
$\Omega_{\mbox{\tiny{$W^{\!\circ}$}}}$ (as in
Lemma~\ref{L:A-PresymplecticRed}$(i)$) coincides with the Poisson
structure $\Lambda_0$ of Prop.~\ref{P:W0isPoisson}.
\end{proposition}

As a motivation to Propositions \ref{P:W0isPoisson} and \ref{P:W0-Poisson2}, we will illustrate these results in the case of Chaplygin symmetries in Example \ref{Ex:Ch} below.


\subsection{The equivalence} \label{Ss:TheEquivalence}

In this section, we compare the Poisson structures $\Lambda$ and
$\Lambda_0$, obtained in Propositions~\ref{P:poisson1} and
\ref{P:W0isPoisson}. We recall the set-up defined by a nonholonomic
system with symmetries.

Let $D$ be the $G$-invariant constraint distribution on the
configuration manifold $Q$, and let $W$ be a $G$-invariant vertical
complement of $D$. This defines an invariant bivector field $\pi_0$
on the submanifold $W^\circ \subset T^*Q$ as in \eqref{eq:defPiC0}.
On the other hand, by using the ($G$-invariant) kinetic energy
metric $\kappa$, we consider the submanifold $\M \subset T^*Q$
equipped with the $G$-invariant vertical complement
\begin{equation}\label{Eq:W}
\W := (T\tau|_\V)^{-1}(W)
\end{equation}
of $\C$ (induced by $W$ as in Remark \ref{R:Split-on-MandQ}). On
$\M$, we have the invariant bivector field $\pi_{\mbox{\tiny
${\mathcal J}\! \mathcal{K}$   }}$, introduced in
Section~\ref{Ss:ViaJK}. Note that if $W$ satisfies the
vertical-symmetry condition, then $\W$ also does (and
$\mathfrak{g}_\subW = \mathfrak{g}_{\mbox{\tiny{$W$}}}$).

The next result relates the almost Poisson manifolds
$(\M,\pi_{\mbox{\tiny ${\mathcal J}\! \mathcal{K}$   }})$ (which
depends on the metric $\kappa$) and $(W^\circ,\pi_0)$ (which is
metric-independent).

\begin{proposition}
\label{P:k-k0isomorphism} For a nonholonomic system with symmetries
satisfying the dimension assumption, let $W$ be a $G$-invariant
vertical complement of $D$ in $TQ$, and consider the almost Poisson
manifolds $(\M,\pi_{\mbox{\tiny ${\mathcal J}\! \mathcal{K}$   }})$
and $(W^\circ,\pi_0)$.
Then there exists a $G$-equivariant diffeomorphism $\Psi : \M \to
W^\circ$ preserving the almost Poisson structures.

As a consequence, if $W$ satisfies the vertical-symmetry condition,
so that $(\M /G, \Lambda)$  and $(W^\circ\!/G, \Lambda_0)$ are
Poisson manifolds, $\Psi$ induces a Poisson diffeomorphism between
them:
\begin{equation}  \label{D:MetricsDiag}
\xymatrix{ (\M, \pi_{\mbox{\tiny ${\mathcal J}\! \mathcal{K}$}})
\ar[d] \ar[rr]^{\Psi}_{\simeq} & &
(W^\circ, \pi_{0} ) \ar[d]_{} \\
(\M/G , \Lambda) \ar[rr]^{\Psi_{\emph\red}}_{\simeq} & &
(W^\circ\!/G, \Lambda_0 ). }
\end{equation}
\end{proposition}

\begin{proof}

Let us consider a local basis $\{X_i, Z_a\}$ of $TQ$ adapted to the
decomposition $TQ=D \oplus W$, such that $\{X_i\}$ is orthonormal
with respect to the kinetic energy metric $\kappa$.   Let $\{X^i,
Z^a\}$ be the dual (local) basis of $T^*Q$. This basis induce local
coordinates $(q; p_i, p_a)$ on $T^*Q$, where $(q)$ are local
coordinates on $Q$.  A local basis of $T^*(T^*Q)$ is given by $\{X^i, Z^a, dp_i,dp_a\}$ and we denote its dual basis by $\{X_i, Z_a, \frac{\partial}{\partial p_i}, \frac{\partial}{\partial p_a}\}$. The Liouville 1-form $\Theta_Q$ on $T^*Q$ is
written as $\Theta_Q= p_iX^i + p_a Z^a$ (see \eqref{Eq:Liouv}), so
\begin{equation}
\Omega_Q= X^i\wedge dp_i - p_i dX^i + Z^a \wedge dp_a - p_a dZ^a. \label{eq:OmegaCoord}
\end{equation}

The local coordinate description of $\M= \kappa^\sharp(D)$ is $\{(q;
p_i, \kappa_{ia} p_i)\}$, where $ \kappa_{ia} =  \kappa(X_i, Z_a)$,
since $\kappa^\sharp(X_i) = {X}^i + \kappa(X_i, Z_a) {Z}^a.$ So the
pull back of the canonical 2-form $\Omega_Q$ to $\M$ is
$$
\Omega_\subM = X^i \wedge d p_i - p_i \cdot dX^i  + Z^a \wedge d
(\kappa_{ia} p_i) - (\kappa_{ia} p_i) \cdot d Z^a .
$$
Observe that, locally, $\C = \textup{span} \{  X_i
,\frac{\partial}{\partial p_i} + \kappa_{ia}
\frac{\partial}{\partial p_a} \}$ and $\W = \textup{span} \{ Z_a
\}$.

Let us compute the 2-form $\langle {\mathcal J}, \mathcal{K}_\subW
\rangle$ (as done in the proof of Prop.~\ref{P:W0isPoisson}) in
these coordinates. We take $\{ \xi^a \}$ a local basis of sections
of $\mathfrak{g}_{\mbox{\tiny{$W$}}} \times Q$ such that $(\xi^a)_Q
= Z_a$. For $X, Y \in \Gamma(\C)$, the coordinates of
$\mathcal{K}_\subW(X,Y)$ are $(\mathcal{K}_\subW  (X,Y) )_a = d
{Z}^a (X,Y)$.  Therefore we obtain
\begin{equation} \label{Proof:JK}
\langle {\mathcal J}, \mathcal{K}_{\subW} \rangle|_\C =
(\kappa_{ia} p_i ) \cdot d Z^a |_\C,
\end{equation}
and, since $Z^a |_\C =0$,
\begin{equation} \label{eq:OmegaJK-Coord}
\Omega_{\mbox{\tiny ${\mathcal J}\! \mathcal{K}$}} |_\C =
(\Omega_\subM + \langle {\mathcal J}, \mathcal{K}_\subW \rangle)
|_\C = (X^i \wedge d p_i -  p_i \cdot d X^i) |_{\C}.
\end{equation}

Consider the metric $\kappa_0$ on $Q$ given at each $ X,Y \in TQ$ by
\begin{equation} \label{Eq:k0}
\kappa_0(X,Y) = \kappa (P_D(X), P_D(Y)) + \kappa(P_W(X), P_W(Y)),
\end{equation}
where $P_D: TQ \to D$ and $P_W : TQ \to W$ are the projections
associated to the decomposition $TQ=D \oplus W$.  The spaces $D$ and
$W$ are orthogonal with respect to this metric. Moreover, the basis
$\{ X_i, Z_a\}$ is $\kappa_0$-orthonormal. Then the coordinate
description of $W^\circ=\kappa_0^\sharp(D)$ is $\{(q; p_i,0)\}$, and
the decomposition $T(W^\circ) = \C_0 \oplus \mathcal{W}_0 $ is
given, locally, by
 $ \C_0 = \textup{span} \{  X_i ,\frac{\partial}{\partial p_i} \}$ and $\W_0 = \textup{span} \{ Z_a \} $.
 On the other hand,
\begin{equation} \label{eq:OmegaW0-Coord}
\Omega_{\mbox{\tiny{$W^{\!\circ}$}}} = \iota^*_0\Omega_Q = X^i \wedge d p_i -  p_i \cdot d X^i.
\end{equation}

We have to see that there exists a diffeomorphism $\Psi:\M \to
W^\circ$ preserving the almost Poisson structures $ \pi_{\mbox{\tiny
${\mathcal J}\! \mathcal{K}$}}$ and $\pi_{0}$. For that, it suffices
to check that
\begin{equation}\label{Eq:apmorphism}
T\Psi (\C) = \C_0 \qquad \mbox{and} \qquad \Omega_{\mbox{\tiny
${\mathcal J}\! \mathcal{K}$}} |_{\C}= \Psi^* (\Omega_{W^\circ})
|_{\C}.
\end{equation}

The map
\begin{equation} \label{Eq:MapPsi}
\Psi := (\kappa_0 |_D)^\sharp \circ (\kappa |_D^{-1})^\sharp : \M
\to W^\circ
\end{equation}
is a diffeomorphism that, by the definition of
$\kappa$, satisfies $\Psi (X^i + \kappa_{ia} Z^a) = (\kappa_0 |_D)
(X_i ) = X^i.$ That is, the global diffeomorphism $\Psi : \M \to
W^\circ$, written in local coordinates, maps $(q; p_i, \kappa_{ia}
p_i)$ on $\M$ into $(q; p_i,0)$ on $W^\circ$. A direct computation
shows that the conditions in \eqref{Eq:apmorphism} hold, so $\Psi$
preserves almost Poisson structures. It  follows from \eqref{Eq:MapPsi} that $\Psi$ is $G$-equivariant so that the induced map
$\M / G\to W^\circ /G$ also preserves the reduced brackets.

\end{proof}

\begin{example}[The nonholonomic particle] \label{Ex:NHParticle3}
Let us keep the notation of Example~\ref{Ex:NHParticle} and \ref{Ex:NHParticle2}. Consider the local coordinates $(\tilde p_x, \tilde p_y,
\tilde p_z)$ associated to the local basis $\{dx,dy,\epsilon\}$ of
$T^*Q$. Then, $\{dx,dy\}$ generates $W^\circ$ so the associated
local coordinates in $W^\circ$ are $(\tilde p_x, \tilde p_y,0)$.
Therefore,
$$
\Omega_{\mbox{\tiny{$W^{\!\circ}$}}} = dx\wedge d\tilde p_x+dy
\wedge d\tilde p_y= dx\wedge dp_x + dy\wedge dp_y$$and
$\displaystyle{\C_0=span \left\{\frac{\partial}{\partial x} + y
\frac{\partial}{\partial z}, \frac{\partial}{\partial y},
\frac{\partial}{\partial p_x}, \frac{\partial}{\partial p_y}
\right\} }$ where $( p_x,p_y,  p_z)$ are the canonical coordinates
on $T^*Q$. The \ isomorphism \ $\Psi:\M \to W^\circ$  \ of \  Proposition \ref{P:k-k0isomorphism}, \ written \ in \ 
canonical \  coordinates, \ is \ $\displaystyle{\Psi(p_x,p_y,yp_z) = \left(
\frac{p_x}{1+y^2},p_y,0 \right).}$ We can check that
$\Psi^*\Omega_{\mbox{\tiny{$W^{\!\circ}$}}} |_{\C_0} =
(1+y^2)dx\wedge dp_x + dy\wedge dp_y +2yp_xdx\wedge dy =
\Omega_{\mbox{\tiny $J\! \mathcal{K}$   }} |_\C.$

On the other hand, the Poisson bivector field $\Lambda_0$ is computed using the presymplectic 2-form $\Omega_{\mbox{\tiny{$W^{\!\circ}$}}}$ and \eqref{Eq:PresymBivector} and thus we obtain that $\Lambda_0=  \frac{\partial}{\partial y}\wedge \frac{\partial}{\partial p_y}$. The Poisson bivector $\Lambda$ in \eqref{Ex:NHPLambda} is isomorphic to $\Lambda_0$ via the isomorphism $\Psi_\red:\M/G \to W^\circ/G$.
\end{example}

\medskip

\begin{example}[Chaplygin systems]\label{Ex:Ch}
In the case of Chaplygin symmetries, as recalled in Example~\ref{Ex:ChapSystems}, the reduced
bivector $\pi^\nh_\red$ on $\M / G$ is defined by a 2-form $\Omega_\red^\nh$.  By Prop.~\ref{P:diagramBbasic}, the reduced bivector $\Lambda$ on $\M / G$  is obtained from $\pi^\nh_\red$
by a gauge transformation by $-\langle \mathcal{J},\mathcal{K} \rangle_\red$, which is equivalent to
saying that $\Lambda$ is defined by the nondegenerate 2-form
\begin{equation}\label{Eq:lambda2form}
\Omega_\red^\nh + \langle \mathcal{J},\mathcal{K} \rangle_\red.
\end{equation}
The assertion that $\Lambda$ is a Poisson structure is the same as saying that the 2-form \eqref{Eq:lambda2form} is closed (i.e., symplectic), which is in agreement with \eqref{Eq:BS}.

The Poisson structure $\Lambda_0$
is defined on $V^\circ / G$. By Prop.~\ref{P:W0-Poisson2}, Lemma~\ref{L:kernel} and Remark~\ref{rm:ch}
in the Appendix,
$\Lambda_0$ corresponds to a symplectic form $\Omega_0$ on $V^\circ / G$, uniquely determined
by the property that
\begin{equation}\label{Eq:pb0}
\rho^* \Omega_0 = \Omega_{\mbox{\tiny{$V^{\!\circ}$}}}=\iota_0^* \Omega_Q.
\end{equation}
If $J: T^*Q\to \mathfrak{g}^*$ is the canonical momentum map for the lifted action on $T^*Q$, then
$V^\circ = J^{-1}(0)$ and it is well-known that the corresponding Marsden-Weinstein reduction is
the cotangent bundle $(T^*(Q/G),\Omega_{\mbox{\tiny{$Q/G$}}})$, see e.g. \cite{MMORR}.
With the identification $V^\circ /G = J^{-1}(0)/G = T^*(Q/G)$, it follows from
\eqref{Eq:pb0} that $\Omega_0$ coincides with the canonical symplectic form  $\Omega_{\mbox{\tiny{$Q/G$}}}$, so the Poisson structure $\Lambda_0$ is just the canonical one on
$T^*(Q/G)$.
Now, by Prop.~\ref{P:k-k0isomorphism} we identify $\M /G$ with $T^*(Q/G)$, $\Lambda$ coincides with $\Lambda_0$ and we recover the following result from \cite{MovingFrames}:  
%
%
%
\begin{corollary}
On a Chaplygin system, we have an identification $\M / G \cong T^*(Q/G)$ in such a way that
$$
\Omega_{\emph\red}^{\emph\nh} = \Omega_{\mbox{\tiny{$Q/G$}}} - \langle \mathcal{J}, \mathcal{K}\rangle_{\emph\red},
$$
where $\Omega_{\mbox{\tiny{$Q/G$}}}$ is the canonical symplectic 2-form on $T^*(Q/G)$.
\end{corollary}
%
\end{example}

\section{Symplectic leaves and the reduced dynamics} \label{S:KerVertical}

Throughout this section, we consider nonholonomic systems with
symmetries satisfying the dimension assumption. We let $W\subseteq
TQ$ be a $G$-invariant vertical complement of the constraint
distribution $D$ satisfying the vertical-symmetry condition
\eqref{Eq:VerticalSymmetries}, and let $\W \subseteq T\M$ be as in
Remark~\ref{R:Split-on-MandQ}. Our aim is to relate the Poisson
structure \eqref{Eq:PoissonW} with the reduced nonholonomic
dynamics.

\subsection{Symplectic leaves}

Let us consider the Poisson structures $\Lambda$ on $\M /G$ (as in
\eqref{Eq:PoissonW}) and $\Lambda_0$ on $W^\circ / G$ (as in
Propositions~\ref{P:W0isPoisson} and \ref{P:W0-Poisson2}). We start
by describing their symplectic leaves.

Using the description of $\Lambda_0$ in Prop.~\ref{P:W0-Poisson2} as
a reduction of $\Omega_{\mbox{\tiny{$W^\circ$}}}$, we have the
following diagram:
\begin{equation} \label{d:Diagram-BWmap}
\xymatrix{(W^\circ, \Omega_{W^\circ}) \ar@{^{(}->}[r]^{\iota_{0}}
\ar[d]_{\rho_0} &  (T^*Q, \Omega_Q) \ar[d]_{\rho} \\  (W^\circ\! /
G, \Lambda_0) \ar@{^{(}->}[r]_{\bar{\iota}_0} & (T^*Q/G,
\pi_{\mbox{\tiny{can} }}). }
\end{equation}
As remarked in Lemma~\ref{L:A-PresymplecticRed}, the symplectic
leaves of $\Lambda_0$ are the (connected components of) the
intersections of the leaves of $\pi_{\mbox{\tiny{can} }}$
with the submanifold $W^\circ /G \hookrightarrow T^*Q / G$.
Moreover, the leaves of $\Lambda_0$ sit inside the leaves of
$\pi_{\mbox{\tiny{can} }}$ as symplectic submanifolds. It is
worth noticing that $W^\circ /G$ is not generally a Poisson
submanifold of $T^*Q/G$ (it is a {\it Poisson-Dirac submanifold} in
the sense of \cite[Sec.~8]{CrFe}).

More explicitly, denoting by $J: T^*Q \to \mathfrak{g}^*$ the
canonical momentum map  for the lifted action on $T^*Q$ and letting
$\mathcal{J}_0 = J\circ \iota_0: W^\circ \to \mathfrak{g}^*$ be its
restriction to $W^\circ$, it directly follows from Cor.~\ref{C:Asl}
that the symplectic leaves of $\Lambda_0$ are given by (connected
components of) the quotients $\mathcal{J}_0^{-1}(\mu)/G_\mu$, for
$\mu\in \mathfrak{g}^*$. The symplectic structures on the leaves are
described as follows: recalling that the symplectic leaves of
$\pi_{\mbox{\tiny{can} }}$ are given by (connected components
of) the Marsden-Weinstein quotients $(J^{-1}(\mu)/G_\mu,
\omega_\mu)$, \cite[Sec.~2]{MMORR} the symplectic structure on
$\mathcal{J}_0^{-1}(\mu)/G_\mu$ is given by
$\bar{\iota}_0^*\omega_\mu$, where we keep the same notation
$$
\bar{\iota}_0: \mathcal{J}_0^{-1}(\mu)/G_\mu \hookrightarrow
J^{-1}(\mu)/G_\mu
$$ 
for the restriction of the the inclusion.

Using the isomorphism $\Psi: \M\to W^\circ$ of
Prop.~\ref{P:k-k0isomorphism} and the induced Poisson diffeomorphism
$\Psi_{\red}: (\M /G, \Lambda) \to (W^\circ /G,\Lambda_0)$, we
conclude the following.

\begin{theorem} \label{T:PoissonPw}
Consider the natural embedding $\bar{\iota}_0 \circ \Psi_{\emph\red}: \M /G
\hookrightarrow T^*Q/G$. Then:
\begin{enumerate}
\item [$(i)$] The symplectic leaves of the Poisson manifold
$(\M / G, \Lambda)$ are connected components of the intersections of
leaves of $\pi_{\mbox{\tiny{\emph{can}}}}$ with the image of $\M /G$
in $T^*Q/G$, and they sit inside the leaves of
$\pi_{\mbox{\tiny{\emph{can}}}}$ as symplectic submanifolds.

\item [$(ii)$] The symplectic leaves of $\Lambda$ are connected
components of quotients $({\mathcal J}_0 \circ \Psi)^{-1} (\mu) /
G_\mu$, for $\mu \in \mathfrak{g}^*$, equipped with symplectic forms
$(\bar{\iota}_0 \circ \Psi_{\emph\red})^*\omega_\mu$, where we keep the
notation $\bar{\iota}_0 \circ \Psi_{\emph\red}: ({\mathcal J}_0 \circ
\Psi)^{-1} (\mu) / G_\mu \hookrightarrow J^{-1}(\mu)/G_\mu$ for the
restricted embedding.
\end{enumerate}
\end{theorem}

We illustrate this result with the nonholonomic particle,
following Examples \ref{Ex:NHParticle}, \ref{Ex:NHParticle2} and \ref{Ex:NHParticle3}.

\begin{example}\label{Ex:NHParticle4}
Recall that, in canonical coordinates, $T^*Q/G$ is given by
$(y,p_x,p_y,p_z)$ and the leaves of the Poisson bracket
$\pi_{\mbox{\tiny can}}$ are the level sets of the functions $p_x$
and $p_z$. Since $W^\circ /G$ is locally described by the
coordinates $(y,p_x,p_y)$, then, as we just saw, the leaves of
$\Lambda_0$ are the level sets of $p_x$. Now, using the
diffeomorphism  $\Psi:\M \to W^\circ$, given in canonical
coordinates by $\displaystyle{\Psi(p_x,p_y,yp_z) = \left(
\frac{p_x}{1+y^2},p_y,0 \right),}$  we obtain the symplectic leaves
of $\Lambda$: by Theorem \ref{T:PoissonPw} they are given by the
inverse image by $\Psi\red$ of the leaves of
$\Lambda_0$, that is, the level sets of $(1+y^2)p_x$ determine the
leaves of $\Lambda$.
\end{example}


\subsection{Relation with the nonholonomic momentum map} \label{Ss:NHMomentumMap}

In what follows, we analyze the role of the nonholonomic momentum
map \cite{BKMM} in the description of the symplectic leaves of the
Poisson manifold $(\M/G,\Lambda)$.

The definition of the nonholonomic momentum map, that we now recall, does not use any choice of complement $W$ of $D$.
Following \cite{BKMM}, let us consider $\mathcal{S} =
\mathcal{C}\cap \mathcal{V} \subset T\M$ (see \eqref{Eq:S}) and
define the subbundle $\mathfrak{g}_\subS \to \M$ of the trivial
bundle $\mathfrak{g}_\subM = \M\times \mathfrak{g} \to \M$ by
\begin{equation} \label{Def:gS}
\xi \in \mathfrak{g}_\subS|_m \ \Longleftrightarrow \ \xi_\subM (m)
\in \S_m, \mbox{ for } m \in \M.
\end{equation}
We can similarly consider $S=V\cap D\subset TQ$ and
$\mathfrak{g}_S\subset \mathfrak{g}\times Q$,
$\mathfrak{g}_{\mbox{\tiny{$S$}}}|_q = \{\xi \ : \ \xi_Q (q) \in
S_q\}$. Since the action on $\M$ is the lifted action on $Q$,
$\mathfrak{g}_\subS |_m$ coincides with
$\mathfrak{g}_{\mbox{\tiny{$S$}}} |_q$ for $\tau(m)=q$.

For an arbitrary section $\xi \in \Gamma(\mathfrak{g}_\subS) \subset
\Gamma(\mathfrak{g}_\subM) = C^\infty(\M,\mathfrak{g})$, we have
\begin{equation} \label{Eq:NoConsNHMomentMap}
 {\bf i}_{\xi_\subM} \Omega_\subM = - {\bf i}_{\xi_\subM} d \Theta_\subM =
  d ({\bf i}_{\xi_\subM} \Theta_\subM)
- \pounds_{\xi_\subM} \Theta_\subM.
\end{equation}

The {\it nonholonomic momentum map} \cite{BKMM} is the map
${\mathcal J}^{\mbox{\tiny nh}} : \M \to \mathfrak{g}_\subS^*$ given
by
\begin{equation} \label{Eq:Definition-NHmomentmap}
{\mathcal J}^{\mbox{\tiny nh}}(m) (\xi) = {\bf i}_{\xi_\subM(m)}
\Theta_\subM (m),
\end{equation}
for $\xi\in \mathfrak{g}_\subS|_m$. Note that ${\mathcal
J}^{\mbox{\tiny nh}} \in \Gamma(\mathfrak{g}_\subS^*)$ so that, if
$\xi \in \Gamma(\mathfrak{g}_\subS)$, then $\langle {\mathcal
J}^{\mbox{\tiny nh}}, \xi \rangle = {\bf i}_{\xi_\subM} \Theta_\subM
\in C^\infty(\M)$.

From equations \eqref{Eq:NoConsNHMomentMap} and
\eqref{Eq:Definition-NHmomentmap}, we see that, for any $\xi \in
\Gamma(\mathfrak{g}_\S)$, we have
\begin{equation} \label{Eq:NH-momentum}
{\bf i}_{\xi_\subM} \Omega_\subM  = d \langle {\mathcal
J}^{\mbox{\tiny nh}} , \xi \rangle  - \pounds_{\xi_\subM}
\Theta_\subM .
\end{equation}
If we now consider the $G$-invariant hamiltonian function
$\Ham_\subM \in C^\infty(\M)$, then from \eqref{Eq:NH-momentum} we
observe that $\langle {\mathcal J}^\nh,\xi \rangle \in C^\infty(\M)$
is conserved by the flow of $X_\nh$ if and only if ${\bf i}_{X_\nh}
\pounds_{\xi_\subM} \Theta_\subM = 0$.

\begin{remark}
If $\xi$ is a constant section, then $\pounds_{\xi_\subM}
\Theta_\subM = 0$ because of the $G$-invariance of the Liouville
1-form. This implies the conservation of the nonholonomic momentum
map in the case of horizontal symmetries, as it was proven in
\cite{ArnoldKozlovNei}.
\end{remark}

The next result relates ${\mathcal J}^\nh : \M \to
\mathfrak{g}_\subS^*$ and ${\mathcal J}_0 \circ \Psi: \M \to
\mathfrak{g}^*$ (considered in Theorem \ref{T:PoissonPw}$(ii)$)
once a choice of $W$ is made. Observe that, for each $m \in\M$, the Lie algebra $\mathfrak{g}$ can
be split into
\begin{equation} \label{eq:g=gs+gw}
\mathfrak{g} = \mathfrak{g}_\subS |_m \oplus \mathfrak{g}_\subW,
\end{equation}
observing that the splitting is simplified by the vertical-symmetry condition,
in that $\mathfrak{g}_\subW$ is independent of $m$. Let
\begin{equation}\label{eq:proj}
P_{\mathfrak{g}_\subS} : \mathfrak{g}_\subM \to \mathfrak{g}_\subS
\end{equation}
be the projection associated to the splitting \eqref{eq:g=gs+gw}. 

\begin{remark}\label{rem:novc}
One can consider a similar splitting to \eqref{eq:g=gs+gw}
if the vertical-symmetry condition does not hold for $W$, except that the complement of $\mathfrak{g}_\subS$ will also vary with $m$; the projection
\eqref{eq:proj} is also defined similarly. 
\end{remark}

\begin{theorem} \label{T:JnhIsMomenMapforJK}
Consider the almost Poisson manifold $(\M, \pi_{\mbox{\tiny
${\mathcal J}\! \mathcal{K}$}})$, defined in Section \ref{Ss:ViaJK}.
The nonholonomic momentum map ${\mathcal J}^{\emph\nh} : \M \to
\mathfrak{g}_\subS^*$ satisfies the following:
\begin{enumerate}
\item[$(i)$]
${\mathcal J}_0 \circ \Psi = P_{\mathfrak{g}_\subS}^* \circ
{\mathcal J}^{\emph{\nh}}$. In other words,
$$
\langle {\mathcal J}_0 \circ \Psi (m) , \eta \rangle = \langle
{\mathcal J}^{\emph{\nh}} (m) , P_{\mathfrak{g}_\subS}(\eta)
\rangle, \qquad \mbox{for all } m \in \M, \eta \in \mathfrak{g}.
$$
\item[$(ii)$] If $\W$ verifies the vertical-symmetry condition \eqref{Eq:VerticalSymmetries} then
$\pi_{\mbox{\tiny ${\mathcal J}\! \mathcal{K}$}} ^\sharp (d  \langle
{\mathcal J}^{\emph{\nh}} , P_{\mathfrak{g}_\subS}(\eta) \rangle ) =
-  (P_{\mathfrak{g}_\subS}(\eta) )_\subM $,\;\; $\eta \in
\mathfrak{g}$.

\end{enumerate}
\end{theorem}

\begin{proof}
 $(i)$ Let us denote, as usual,  $\iota_0:W^\circ \to T^*Q$ and $\iota:\M \to T^*Q$ the corresponding inclusions and let $P_D: TQ \to D$ is the projection associated to decomposition $TQ= D \oplus W$. Let $v \in TQ$ and $m \in \M$, then since $\Psi(m) \in W^\circ$, we have that $\langle \iota_0( \Psi(m)) , v\rangle = \langle \iota_0(\Psi(m)), P_D(v) \rangle$.
Therefore, if $\eta \in \mathfrak{g}$, and using that $\eta_Q= T\tau_0 (\eta_{W^\circ})$ we obtain
 $$
\langle {\mathcal J}_0 \circ \Psi (m), \eta \rangle = \langle \iota_0( \Psi(m)), \eta_Q \rangle =  \langle \iota_0(\Psi(m)), P_D(\eta_Q) \rangle
$$

Now, by \eqref{Eq:MapPsi}, we have that  $\iota_0(\Psi(m)) - \iota(m) \in D^\circ \subset T^*Q$ since if $m = \kappa^\sharp(v_q)$ where $v_q \in D_q$ then
$(\iota_0(\Psi(m)) - \iota(m) )|_D = (\kappa_0^\sharp(v_q) - \kappa^\sharp(v_q))|_D = 0$ by \eqref{Eq:k0}.
Therefore,
$$\langle \iota_0(\Psi(m)), P_D(\eta_Q) \rangle = \langle \iota(m), P_D(\eta_Q) \rangle = \langle \iota(m), (P_{\mathfrak{g}_\subS} \eta )_Q \rangle =  \langle {\mathcal J}^{\nh} (m) , P_{\mathfrak{g}_\subS}(\eta) \rangle.
$$

$(ii)$  Since ${\mathcal J}_0 :W^\circ \to \mathfrak{g}^*$ is the restriction of the canonical momentum map on $T^*Q$ to $W^\circ$  then at each $\eta \in \mathfrak{g}$,  ${\bf i}_{\eta_{W^\circ}} \Omega_{\mbox{\tiny{$W^\circ$}}} = d \langle {\mathcal J}_0 , \eta\rangle$. By Lemma \ref{L:kernel} we have that $\textup{Ker}\, \Omega_{W^\circ} = \W_0$, and denoting by $P_{\subC_0}:T(W^\circ) \to \C_0$ the projection associated to decomposition $T(W^\circ) = \C_0 \oplus \W_0$ we obtain that ${\bf i}_{P_{\subC_0} (\eta_{W^\circ})}  \Omega_{W^\circ} |_{\C_0} = d \langle {\mathcal J}_0 , \eta\rangle |_{\C_0}$. Using \eqref{eq:defPiC0} we obtian that $\pi_0^\sharp( d \langle {\mathcal J}_0 , \eta\rangle) = - P_{\C_0}(\eta_{\mbox{\tiny{$W^{\!\circ}$}}})$. Prop.~\ref{P:k-k0isomorphism} asserts that, for each $\a \in T^*(W^\circ)$,  $T\Psi (\pi_{\mbox{\tiny ${\mathcal J}\! \mathcal{K}$}}^\sharp(\Psi^*(\a))) = \pi_0^\sharp(\a)$, and thus $\pi_{\mbox{\tiny ${\mathcal J}\! \mathcal{K}$}}^\sharp(d \langle {\mathcal J}_0 \circ \Psi , \eta \rangle ) = \pi_{\mbox{\tiny ${\mathcal J}\! \mathcal{K}$}}^\sharp( \Psi^*d \langle {\mathcal J}_0, \eta \rangle ) = - P_\subC (\eta_\subM)$ since $T\Psi P_\subC (\eta_\subM) = P_{\C_0}(\eta_{\mbox{\tiny{$W^{\!\circ}$}}})$. Finally, using that $P_\subC(\eta_\subM) = \left( P_{\mathfrak{g}_\subS}(\eta) \right)_\subM$ and item $(i)$, we get that  $\pi_{\mbox{\tiny ${\mathcal J}\! \mathcal{K}$}}^\sharp(d \langle {\mathcal J}^\nh , P_{\mathfrak{g}_\subS}(\eta) \rangle ) = -(P_{\mathfrak{g}_\subS} (\eta))_\subM$.

\end{proof}

\begin{corollary} \label{C:LeavesOfPiP}
The symplectic leaves of the Poisson manifold $(\M / G, \Lambda)$
are (the connected components of) ${({\mathcal
J}^{\emph{\nh}})}^{-1}(\iota_{\mathfrak{g}}^* \mu) / G_{\mu}$, for
$\mu \in \mathfrak{g}^*$, where
$\iota_{\mathfrak{g}}:\mathfrak{g}_\subS \to \mathfrak{g}_\subM$ is
the natural inclusion.
\end{corollary}

In the previous statement, we view $\mu \in \mathfrak{g}^*$ as  a
constant section of $\mathfrak{g}_\subM^*$, so that
$\iota_{\mathfrak{g}}^* \mu$ is a (not necessarily constant) section
of $\mathfrak{g}_\subS^*$.

\begin{proof} By
Theorem \ref{T:PoissonPw}, the leaves of $\Lambda$ are the connected
components of $({\mathcal J}_0 \circ \Psi)^{-1}(\mu) / G_{\mu}$, for
$\mu \in \mathfrak{g}^*$. Using Theorem
\ref{T:JnhIsMomenMapforJK}$(i)$ and the fact that
$\iota_{\mathfrak{g}}^* \circ P_{\mathfrak{g}_\subS}^* =
\textup{Id}_{\mathfrak{g}_\subS^*}$, we see that $({\mathcal J}_0
\circ \Psi)^{-1} (\mu) = {({\mathcal
J}^{\nh})}^{-1}(\iota_{\mathfrak{g}}^* \mu)$, for $\mu \in
\mathfrak{g}^*$.
\end{proof}

\begin{remark}
Following Remark~\ref{rem:novc}, we note that part $(i)$ of Theorem~\ref{T:JnhIsMomenMapforJK} still holds without the vertical-symmetry condition.
\end{remark}


\subsection{Reduced nonholonomic brackets as gauge transformations of Poisson structures} \label{Ss:DeformedReducedDyn}

As we already remarked in section~\ref{Ss:ViaJK}, the Poisson
bivector $\Lambda$ on $\M /G$ may not describe the reduced dynamics,
since $\pi_{\mbox{\tiny ${\mathcal J}\! \mathcal{K}$ }}$ does not
necessarily describe the dynamics defined by the hamiltonian
$\mathcal{H}_\subM$. Still, there are situations in which the
reduced dynamics is described by a bivector field obtained as a
gauge transformation of $\Lambda$:

\begin{theorem} \label{T:Reduced-Dyn}
Suppose that $B$ is a $G$-invariant 2-form on $\M$ satisfying
\eqref{Eq:Bsection} and defining a dynamical gauge transformation of
$\pi_{\emph\nh}$ (as in Def.~\ref{D:dynamicalGauge}). If $B + \langle
{\mathcal J}, \mathcal{K}_\subW \rangle$ is basic with respect to
the orbit projection $\rho: \M \to \M /G$, then the reduced dynamics
is described by a bivector field $\pi_{\emph{\red}}^\B$ which is a
gauge transformation of $\Lambda$ by the 2-form $\mathcal{B}$, where
$\mathcal{B}$ is such that $\rho^*\mathcal{B} = B + \langle
{\mathcal J}, \mathcal{K}_\subW \rangle$. 

Assuming that $\W$ verifies the vertical-symmetry condition \eqref{Eq:VerticalSymmetries}, we conclude that $\pi_{\emph\red}^\B$ is a twisted Poisson structure that is gauge related to the Poisson bivector $\Lambda$.

\end{theorem}

The various bivectors and gauge transformations in
Thm.~\ref{T:Reduced-Dyn} are represented by the following diagram:
\begin{equation}
\xymatrix{ (\M, \pi_{\mbox{\tiny ${\mathcal J}\! \mathcal{K}$   }} )
\ar[rr]^{B+ \langle {\mathcal J}, \mathcal{K}_\subW \rangle}
\ar[d]^{\rho} &&
 (\M, \pi_\B) \ar[r]^{-B} \ar[d]^\rho & (\M, \pi_{\nh}) \\
(\M/G, \Lambda) \ar[rr]^{\mathcal B} && (\M/G,  \pi_{\red}^\B) & \
} \end{equation}

\begin{proof}

Proposition \ref{P:diagramBbasic}
asserts that the bivector $\pi^\B_\red$ is gauge-related to
$\Lambda$ by $\mathcal{B}$. If $\W$ verifies the vertical-symmetry condition \eqref{Eq:VerticalSymmetries}, then by Corollary \ref{C:TwistedVertSymm} bivector  $\pi_\red^\B$ is
$(-d{\mathcal B})$-twisted Poisson, where ${\mathcal B}$ is the 2-form on $\M/G$
such that $\rho^*{\mathcal B}= B +  \langle {\mathcal J},
\mathcal{K}_\subW \rangle$. 
\end{proof}

The procedure in Thm.~\ref{T:Reduced-Dyn} will be illustrated in
section~\ref{S:Examples} (see section~\ref{Ex:RigidBody}).

\begin{corollary} \label{C:ConservedQuantity}
Under the assumptions of Theorem~\ref{T:Reduced-Dyn}, for each $\eta
\in \mathfrak{g}$, the function $\langle {\mathcal J}^{\emph{\nh}},
P_{\mathfrak{g}_\subS}(\eta)\rangle \in C^\infty(\M)$ is conserved
by the flow of $X_{\emph\nh}$.
\end{corollary}

\begin{proof}
Let $\xi=P_{\mathfrak{g}_\subS}(\eta)$, for $\eta \in \mathfrak{g}$.
Theorem \ref{T:JnhIsMomenMapforJK}$(ii)$ asserts that
$\pi_{\mbox{\tiny ${\mathcal J}\! \mathcal{K}$}}^\sharp(d \langle
{\mathcal J}^\nh, \xi \rangle ) = - \xi_\subM$. Since $\pi_\B$ is
the gauge transformation of $\pi_{\mbox{\tiny$\mathcal{J}\!\mathcal{K}$ }}$ by $B + \langle {\mathcal J}, \mathcal{K}_\subW
\rangle$, we have $\pi_{\mbox{\tiny $\mathcal{J}\! \mathcal{K}$
}}^\sharp= \pi^\sharp_\B \circ (\textup{Id} + (B + \langle {\mathcal
J}, \mathcal{K}_\subW \rangle)^\flat \circ \pi_{\mbox{\tiny
$\mathcal{J}\! \mathcal{K}$ }}^\sharp ).$ Therefore,
$$ - \xi_\subM = \pi_{\mbox{\tiny ${\mathcal J}\! \mathcal{K}$}}^\sharp(d \langle {\mathcal J}^\nh, \xi \rangle ) =
 \pi^\sharp_\B (d \langle {\mathcal J}^\nh, \xi \rangle -
{\bf i}_{\xi_\subM}(B + \langle {\mathcal J}, \mathcal{K}_\subW
\rangle)) = \pi_\B^\sharp(d \langle {\mathcal J}^\nh, \xi
\rangle),$$ since $B+  \langle {\mathcal J}, \mathcal{K}_\subW
\rangle$ is basic. But on the other hand if $B$ defines a dynamical
gauge then $\pi_\B^\sharp(d \Ham_\subM) = - X_\nh.$ Finally, we
obtain that ${\bf i}_{X_\nh}(d \langle {\mathcal J}^\nh, \xi
\rangle) = - d\Ham_\subM (\xi_\subM) = 0$ because of the
$G$-invariance $\Ham_\subM$.
\end{proof}

We stress that the fact that the reduced dynamics is described by a
bivector field $\pi_\red^\B$ gauge equivalent to the Poisson
structure $\Lambda$ has several consequences: for example,
$\pi_\red^\B$ is a twisted Poisson structure, so its charctaeristic
distribution is integrable -- in fact, its almost symplectic leaves
coincide with the leaves of $\Lambda$, and the Casimirs of $\Lambda$
are conserved quantities of the nonholonomic system. In particular,
the leaves of $\pi_\red^\B$ are described by Thm.~\ref{T:PoissonPw}
and Cor.~\ref{C:LeavesOfPiP}. The only difference is on the leafwise
2-forms: while for $\Lambda$ the symplectic form on a leaf
$({\mathcal J}_0 \circ \Psi)^{-1} (\mu) / G_\mu = {({\mathcal
J}^{\nh})}^{-1}(\iota_{\mathfrak{g}}^* \mu) / G_{\mu}$ is
$\Omega_\mu:= (\bar{\iota}_0 \circ \Psi_\red)^*\omega_\mu$ (see
Thm.~\ref{T:PoissonPw}$(ii)$), for $\pi_\red^\B$ the almost
symplectic form on the same leaf is
$$
\Omega_\mu - \mathcal{B}_\mu,
$$
where $\mathcal{B}_\mu$ is the pullback of the 2-form $\mathcal{B}$
on $\M / G$ (in Thm.~\ref{T:Reduced-Dyn}) to the leaf.

We derive some conclusions about the nonholonomic particle based on
the results of this section.

\begin{example}[Nonholonomic particle]  \label{Ex:NHParticleFinal}

Recall from Examples~\ref{Ex:NHParticle}, \ref{Ex:NHParticle2}, \ref{Ex:NHParticle3}, \ref{Ex:NHParticle4} that the reduced nonholonomic bivector field on $\M/G$ is given by
$$\pi_\red^\nh= \frac{\partial}{\partial y} \wedge \frac{\partial}{\partial
p_y} - \frac{yp_x}{1+y^2} \frac{\partial}{\partial p_x} \wedge
\frac{\partial}{\partial p_y}$$  while the Poisson bivector field $\Lambda$ is written in \eqref{Ex:NHPLambda}.   
As we observed in Example~4.15 the 2-form $\langle {\mathcal J},
\mathcal{K}_\subW\rangle = y p_x dx\wedge dy$ relating $\pi_\nh$ and $\pi_{\mbox{\tiny ${\mathcal J}\! \mathcal{K}$   }}$ is not basic and thus $\pi_\red^\nh$ and $\Lambda$ are not necessarly gauge related (see Prop. \ref{P:diagramBbasic}). In fact,  for
$(a,b) \in \mathfrak{g} \simeq \R^2$ and using \eqref{Ex:NHPart-LieSubspaces},   we see that the function $\langle {\mathcal J}^\nh,
P_{\mathfrak{g}_\subS}(a,b)\rangle = \langle {\mathcal J}^\nh,
a(1,y)\rangle = (1+y^2)p_x$ determines the leaves of $\Lambda$ (Corollary
\ref{C:LeavesOfPiP}) since  $f=(1+y^2)p_x$  is a basic function on $\M$.  In other words, $f$ is a Casimir of $\Lambda$. 
However, $f$ is not a Casimir of $\pi_\red^\nh$ since it is not conserved by the motion (see \cite[Sec.5.6]{BlochBook}). Therefore, $\pi_\red^\B$ is not gauge related with $\Lambda$.

On the other hand, if there was a dynamical gauge $B$
satisfying the conditions in Theorem \ref{T:Reduced-Dyn} then $B+  \langle {\mathcal J},
\mathcal{K}_\subW \rangle$ would be basic and ${\bf i}_{X_\nh} B=0$. Hence,  ${\bf i}_{X_\nh} \langle {\mathcal J},
\mathcal{K}_\subW \rangle$ would be semi-basic with respect to the
orbit projection $\rho:\M \to \M/G$. However, using equation
\eqref{Eq:JK-NHParticle} it can be checked that ${\bf i}_{X_\nh}
\langle {\mathcal J}, \mathcal{K}_\subW \rangle$  is not semi-basic.
As a final attempt to obtain a reduced bivector field describing the dynamics that is gauge related to $\Lambda$ we may
try to change the vertical complement $W$. Note, however, that in general the
vertical-symmetry condition will no longer be satisfied.


\end{example}


\section{Examples} \label{S:Examples}

So far we have illustrated our results using the nonholonomic
particle (see Examples~\ref{Ex:NHParticle}, \ref{Ex:NHParticle2}, \ref{Ex:NHParticle3}, \ref{Ex:NHParticle4} and \ref{Ex:NHParticleFinal} and Chaplygin systems (see
Example~\ref{Ex:ChapSystems} and Example~\ref{Ex:Ch}). This section presents other
mechanical examples, all of which exhibit symmetries which are not
Chaplygin, so their reduced dynamics are described by bivector
fields (not 2-forms). In all cases, we find suitable choices of
complements $\W$ and dynamical gauge transformations so that the
reduced dynamics is described by twisted Poisson structures (as e.g.
in Corollary~\ref{C:Twisted} and Theorem~\ref{T:Reduced-Dyn}).

More concretely, in Section~\ref{Ex:VerticalDisk} we verify that the
reduced dynamics of the vertical rolling disk is directly described
by the Poisson structure $\Lambda$. In Section \ref{Ex:skateboard}
we consider the Snakeboard, and show that the 2-form $\langle
{\mathcal J}, \mathcal{K}_\subW \rangle$ in this case is basic and
induces a gauge transformation relating the reduced bivector field
$\pi_\red^\nh$ to the Poisson bivector $\Lambda$. Section
\ref{Ex:RigidBody} illustrates the usage of a dynamical gauge
transformation in Theorem \ref{T:Reduced-Dyn} for the Chaplygin
ball, also leading to a description of the reduced dynamics by a
bracket given by a gauge transformation of the Poisson bracket
$\Lambda$. Finally, in Section \ref{Ex:Cylinder}, by means of a gauge transformation we obtain a Poisson bracket describing the dynamics of a homogeneous sphere rolling on a cylinder.


\subsection{The vertical rolling disk} \label{Ex:VerticalDisk}

Consider a vertical disk  of radius $R$ rolling on a plane without
sliding, see e.g. \cite[Sec.~1.4 and~5.6]{BlochBook}.  Let $(x,y) \in
\R^2$ represent the contact point between the plane and the disk,
$\varphi$ and $\psi$ represent the rotation angle of a point in the
disk with respect to the vertical axis and the orientation angle of
the disk, respectively. 
The nonsliding constraints are given by the 1-forms
$\epsilon_x= dx - R \cos\psi d \varphi$ and $\epsilon_y=dy-R\sin\psi d \varphi.$
The action of the (direct product) Lie group $G=\R^2 \times S^1$ on $Q$ 
such that $V= \textup{span}
\{\frac{\partial}{\partial x},\frac{\partial}{\partial y},
\frac{\partial}{\partial \varphi} \}$ is a symmetry of the system. 
We choose the complement
$W \subset V$ to be
$ W= \textup{span} \{\frac{\partial}{\partial x} \, , \, \frac{\partial}{\partial y} \}.$ 
Note that $W$ verifies the vertical-symmetry condition
\eqref{Eq:VerticalSymm-Q} since $\mathfrak{g}_{\mbox{\tiny{$W$}}} =
\textup{span} \{(1,0,0), (0,1,0)\}$.
As it is known $\langle {\mathcal J}, \mathcal{K}_\subW \rangle =0$. 


%
%
%

The submanifold $W^\circ$ of $T^*Q$ is locally
expressed as $W^\circ =\{\tilde p_x =\tilde p_y = 0 \}$, where $(\tilde p_x,\tilde p_y,\tilde p_\varphi,\tilde p_\psi)$ are the coordinates associated to the basis  $\{
\epsilon_x, \epsilon_y , d\varphi,  d\psi \}$, and thus
the presymplectic 2-form $\Omega_{\mbox{\tiny{$W^{\!\circ}$}}}$ is
given by $\Omega_{\mbox{\tiny{$W^{\!\circ}$}}} =  d\psi \wedge
d\tilde p_\psi + d\varphi \wedge d\tilde p_\varphi.$ Using
Prop.~\ref{P:W0-Poisson2}, on the reduced space $W^\circ\!/G$
we have the Poisson structure $\Lambda_0= \frac{\partial}{\partial
\psi} \wedge \frac{\partial}{\partial \tilde p_\psi}$ which is computed using  \eqref{Eq:PresymBivector}. To compute
the Poisson bivector $\Lambda$ on $\M/G$ we consider the isomorphism
$\Psi:\M \to W^\circ$ given by $\Psi(q; \iota^*\tilde p_x,
\iota^*\tilde p_y, \tilde p_\varphi, \tilde p_\psi) = (q; 0,0,
\tilde p_\varphi, \tilde p_\psi).$ Thus, by Proposition
\ref{P:k-k0isomorphism},
$$
\Lambda = \frac{\partial}{\partial \psi} \wedge
\frac{\partial}{\partial \tilde p_\psi}.
$$

On the other hand, at each $m \in \M$,  $\mathfrak{g}_\subS = \textup{span}
\{( R \cos \psi, R \sin \psi ,1) \}$ (see \eqref{Def:gS}). For
$\eta=(a,b,c) \in \mathfrak{g} \simeq \R^3$, we have that  $P_{\mathfrak{g}_\subS}(\eta)= c( R \cos \psi, R \sin \psi ,1)\rangle$ and thus the function $\langle {\mathcal J}^\nh, P_{\mathfrak{g}_\subS}(\eta)\rangle = c\, \tilde p_\varphi$
is conserved during the motion by Corollary
\ref{C:ConservedQuantity} (in agreement with \cite{BlochBook}). 
Observe that $\tilde p_\varphi$ is not only a conserved quantity but also it is also a Casimir of $\Lambda$ (hence defining its leaves).


\subsection{The snakeboard} \label{Ex:skateboard}

The snakeboard is a variation of the skateboard which is
modeled on the configuration manifold $Q= {\bf SE} (2) \times S^1
\times S^1$. The coordinates $(x,y,\theta) \in {\bf SE}(2)$
represent the position and orientation of the center of the board,
$\phi \in S^1$ is the angle of the momentum wheel relative to the
board, and $\psi \in S^1$ is the angle of the back and front wheels
(we are assuming that they coincide). 
We are going to work with the simplified version of the snakeboard considered in \cite[Sec.~2.5]{MarsdenKoon2} and \cite[Sec.~7.2]{Oscar} where, for $\phi \neq 0$, the non-sliding constraints are given by the annihilators of the 1-forms
$$
\epsilon_x= dx+ r\cot \phi \cos \theta d\theta \qquad \mbox{and}
\qquad \epsilon_y = dy+ r\cot \phi \sin \theta d\theta.
$$

The (direct product) Lie group $G=\R^2 \times S^1$ is a symmetry of
the system (see \cite{Oscar}) 
where the action on $Q$ is such that $\displaystyle{V =
\textup{span} \left\{
\frac{\partial}{\partial x} , \frac{\partial}{\partial y} , \frac{\partial}{\partial \psi}
\right\}}$. Consider the following basis of $D$ and the vertical
complement $W$ given by
$$
D=\textup{span} \left\{\frac{\partial}{\partial \theta} - r \cos
\theta \cot \phi \frac{\partial}{\partial x} - r \sin \theta \cot
\phi \frac{\partial}{\partial y},\frac{\partial}{\partial \phi}  ,
\frac{\partial}{\partial \psi} \right\} \qquad \mbox{and} \qquad
W=\textup{span}\left\{\frac{\partial}{\partial x},
\frac{\partial}{\partial y} \right\}.
$$
In this case the vertical-symmetry condition
\eqref{Eq:VerticalSymm-Q} is satisfied for
$\mathfrak{g}_{\mbox{\tiny{$W$}}} = \textup{span}
\{(1,0,0),(0,1,0)\}$. Hence, it is sufficient to compute the 2-form
$\langle {\mathcal J}, {\mathcal K}_\subW \rangle$, for $\W$ as in
\eqref{Eq:W}. The lagrangian function is just the kinetic energy
metric,
$$
\Lag = \frac{1}{2}\left( m(\dot x^2 +\dot y^2 + r^2 \dot \theta^2) +
2{\bf J} \dot \theta \dot \psi + 2{\bf J}_1 \dot \phi^2 + {\bf J}
\dot \psi^2 \right).$$ For $q=(x,y,\theta,\phi,\psi) \in Q$, let us
denote $( \tilde{p}_x, \tilde{p}_y, \tilde p_\theta, \tilde p_\phi,
\tilde p_\psi)$ the coordinates on $T_q^*Q$ associated to the basis
$\{\epsilon_x, \epsilon_y, d\theta, d\phi,d\psi  \}$. Therefore, the
manifold $\M = \L(D)$ is given in coordinates by
$$
\M = \left \{\tilde p_x= -\frac{m r \cos\theta \sin ^2 \phi \cot
\phi}{mr^2 - {\bf J} \sin^2\phi}  (\tilde p_\theta - \tilde p_\psi),
\ \ \ \ \tilde p_y = -\frac{mr \sin\theta  \sin ^2 \phi \cot
\phi}{mr^2 - {\bf J} \sin^2\phi}  (\tilde p_\theta - \tilde p_\psi)
\right \}. $$

Since $\mathfrak{g}_\subS = \textup{span} \{(0,0,1)\}$, the momentum
map adapted to this basis is, for $m \in \M$, ${\mathcal J}(m) =
(\iota^*\tilde p_x, \iota^*\tilde p_y,\tilde p_\psi),$ where
$\iota:\M \to T^*Q$ is the inclusion. The $\W$-curvature is given by
$\mathcal{K}_\subW = (d\epsilon_x, d\epsilon_y,0)$,  and thus the
2-form $\langle {\mathcal J}, \mathcal{K}_\subW \rangle$ is
$$
\langle {\mathcal J}, \mathcal{K}_\subW \rangle =  - \frac{mr^2 \cot
\phi}{mr^2 - {\bf J} \sin^2\phi} (\tilde p_\theta - \tilde p_\psi)
d\theta \wedge d\phi,
$$ 
which is basic with respect to the orbit
projection $\rho: \M \to \M/G$.  By Corollary
\ref{C:TwistedVertSymm}, the reduced bivector field $\pi_\red^\nh$
is twisted Poisson, and hence has an integrable characteristic
distribution.  
Observe that in \cite{Oscar} this example is studied doing, first, a Chaplygin reduction by the subgroup $G_\subW = \R^2$ and then it is shown that $\Omega_{\mbox{\tiny{$Q\!/\!G$}}} - \langle {\mathcal J}, \mathcal{K}_\subW \rangle$ admits a conformal factor after an almost symplectic reduction. In this example we emphasize the fact that it is possible to know whether the reduction of $\pi_\red^\nh$ is twisted by computing the formulas of the Jacobiator. We also note that, by Theorem \ref{T:Reduced-Dyn},
$\pi_\red^\nh$ is gauge related with the Poisson bivector $\Lambda$
on $\M/G$ by the 2-form $\langle {\mathcal J}, \mathcal{K}_\subW
\rangle_\red$ (here we are taking $B\equiv 0$). 
The submanifold $W^\circ \subset T^*Q$ is described in
coordinates by $\tilde{p}_x=\tilde{p}_y=0$. Using
\eqref{eq:OmegaCoord} we see that
$\Omega_{\mbox{\tiny{$W^{\!\circ}$}}} = d\theta \wedge d\tilde
p_\theta + d\phi \wedge d\tilde p_\phi + d\psi \wedge d\tilde
p_\psi$. Its reduction by $G$ gives the Poisson bivector $\Lambda_0$
on $W^\circ\!/G$ (Prop.~\ref{P:W0isPoisson}) described, in local
coordinates $(\theta,\phi, \tilde p_\theta, \tilde p_\phi, \tilde
p_\psi)$, by
$$
\Lambda_0 = \frac{\partial}{\partial \theta} \wedge \frac{\partial}{\partial \tilde{p}_\theta} +  \frac{\partial}{\partial \phi} \wedge \frac{\partial}{\partial \tilde{p}_\phi},
$$
for which $\tilde{p}_\psi$ is a Casimir. The Poisson bivector $\Lambda$ is
isomorphic to $\Lambda_0$ via the map $\Psi\red$ defined in \eqref{D:MetricsDiag} and thus $\Psi^*_\red\tilde
p_\psi=\tilde
p_\psi$ is also a Casimir of $\Lambda$ (Theorem \ref{T:PoissonPw}).
Since $\pi_\red^\nh$ and $\Lambda$ are gauge related, the leaves of
$\pi_\red^\nh$  are given by the level sets of  $\tilde p_\psi$ and
thus the function $\tilde p_\psi$ is a conserved quantity of the
nonholonomic system, in agreement with the fact that the system has horizontal symmetries. 
%
%


\subsection{Rigid body with nonholonomic constraints} \label{Ex:RigidBody}

We now consider the motion of an inhomogeneous sphere whose center
of mass coincides with its geometric center that rolls without
slipping on the plane.

Following \cite[Sec.~5]{PL2011}, we can see this example as a
particular case of a rigid body with generalized rolling constraints
which describes the motion of a rigid body in space that evolves
under its own inertia and it is subject to constraints that relate
the linear velocity of the center of mass with the the angular
velocity of the body. If ${\bf x} \in \R^3$ represents the center of
mass and $\vecom$ the angular velocity of the body with respect to
an inertial frame, then the constraints are written as
\begin{equation} \label{Ex:RB-constraints}
\dot {\bf x} = r A \vecom,
\end{equation}
where  $r$ denotes the radius of the sphere and $A$ is a given $3 \times 3$ matrix that, following \cite{PL2011}, satisfies one of the (simplified) conditions:
\begin{equation} \label{Ex:RB-Matrices}
A\equiv 0, \quad A=\left(\begin{array}{ccc} 0&0&0\\
0&0&0\\ 0&0&1 \end{array} \right), \quad A=\left(\begin{array}{ccc}
0&1&0\\ -1&0&0\\ 0&0&0 \end{array}\right),
\quad \mbox{or} \quad A=\left(\begin{array}{ccc} 0&1&0\\
-1&0&0\\ 0&0&1 \end{array}\right). \end{equation}

We can interpret the motion determined by \eqref{Ex:RB-constraints}
for each matrix $A$: if $A\equiv 0$, \eqref{Ex:RB-constraints}
represents the motion of a free rigid body; If $A$ has rank 1 then
the body is allowed to move along the vertical axis satisfying the
constraints; When the matrix $A$ has rank 2,
\eqref{Ex:RB-constraints} determines a ball rolling on a plane
without sliding (the Chaplygin ball), while if the rank of $A$
equals 3 then any two points on the configuration space can be
joined by a curve satisfying the constraints.

The configuration manifold is $Q= \SO(3) \times \R^3$ with
coordinates $(g, {\bf x})$, where $g$ is an orthogonal matrix that
specifies the orientation of the ball by relating two orthogonal
frames, one attached to the body and one that is fixed in space, and
${\bf x}$ represents the position of the center of mass in space. We
will assume that the body frame has its origin at the center of mass
and is aligned with the principal axes of inertia of the body. These
frames define the so-called space and body coordinates,
respectively.

Given a motion $(g(t); x(t))  \in Q$, the angular velocity vector in
space $\vecom$ and body coordinates $\vecOm$ are respectively given
by
$$\hat{\vecom}(t) = \dot g(t)g^{-1}(t)\qquad \hat{\vecOm}(t) = g^{-1}(t) \dot g(t),$$
where we use the usual identification $\hat{ \ }:\mathfrak{so}(3)
\to \R^3$, and satisfy ${\vecom} = g \vecOm$. Therefore, the
constraints can also be written as $\dot {\bf x} = rAg \vecOm.$

Let $\vecL$ and $\vecR$ be the left and right Maurer-Cartan forms,
respectively, that can be thought as $\R^3$-valued 1-forms by means
of the identification of the Lie algebra $\mathfrak{so}(3)$ with
$\R^3$ by the hat map. Thus for a tangent vector $v_g \in T_g
\SO(3)$  we have $\vecom = \vecR(v_g)$ and $\vecOm = \vecL(v_g)$.
Therefore, the constraints can be written in terms of a
$\R^3$-valued 1-form $\vecep=(\epsilon_1,\epsilon_2,\epsilon_3)$,
where $\epsilon_i$ are the 1-forms defining the constraints. In
local coordinates, $\vecep$ is given by
$$\vecep = d{\bf x} - r A \vecR \ = \ d{\bf x} - rAg\vecL,
$$
where we are using the relation $\vecR = g\vecL$.

The constraint distribution $D$ on $TQ$ is given by
\begin{equation} \label{Ex:RB-distribD}
D = \textup{span} \left\{ {\bf X}^{\mbox{\tiny right}} + r A \frac{\partial}{\partial {\bf x}} \right\} \ = \ \textup{span} \left\{  {\bf X}^{\mbox{\tiny left}} + r Ag \frac{\partial}{\partial {\bf x}} \right\},
\end{equation}
where ${\bf X}^{\mbox{\tiny right}} = (X^{\mbox{\tiny
right}}_1,X^{\mbox{\tiny right}}_2, X^{\mbox{\tiny right}}_3)$ is
the moving frame of $\SO(3)$ dual to $\vecR=(\rho_1, \rho_2,
\rho_3)$; similarly, ${\bf X}^{\mbox{\tiny left}}$ is the moving
frame of $\SO(3)$ dual to $\vecL$.

The left action of the Lie group $G= \{(h,a) \in \SO(3) \times \R^3
\ : \ h{\bf e}_3 = {\bf e}_3 \}$ on $Q$ is given by
$$(h,a): (g,{\bf x}) \mapsto (hg, h{\bf x} + a) \in \SO(3) \times \R^3.$$
This $G$-action on $Q$ is a symmetry for the nonholonomic system
since the lifted action on $TQ$ leaves the constraints and the
hamiltonian invariant (observe that $hA=Ah$ when $A$ is a matrix of
the form described in \eqref{Ex:RB-Matrices}). In this case, the
vertical space $V$ is given by
$$
V= \textup{span} \left\{ \vecgamma \cdot {\bf X}^{\mbox{\tiny left}}
- x^2  \frac{\partial}{\partial x^1} + x^1  \frac{\partial}{\partial x^2}\, , \,
\frac{\partial}{\partial x^1}\, , \, \frac{\partial}{\partial x^2}\,
, \, \frac{\partial}{\partial x^3}  \right\},
$$
where $\vecgamma = (\gamma_1, \gamma_2,\gamma_3)$ is the third row
of the matrix $g \in \SO(3)$. Note that $S = \textup{span} \left\{
\vecgamma \cdot ( {\bf X}^{\mbox{\tiny left}} + r Ag
\frac{\partial}{\partial {\bf x}} ) \right\}$.  The Lie algebra $\mathfrak{g}$
associated to $G$ is abelian Lie algebra identified with $\R \times \R^3$.

Let us define the $G$-invariant vertical complement $W$ of $D$ on $TQ$ given by
$$W = \textup{span} \left\{ \frac{\partial}{\partial
x^1},  \frac{\partial}{\partial x^2},  \frac{\partial}{\partial x^3}
\right\}. $$ The distribution $W$ satisfies the vertical-symmetry condition \eqref{Eq:VerticalSymm-Q}
for $\mathfrak{g}_{\mbox{\tiny{$W$}}} =
\textup{span} \{ \eta^i:= (0;{\bf e}_i), i=1,2,3 \}$ where ${\bf
e}_i$ are the canonical vectors in $\R^3$.

If we consider the basis $\{  {\bf X}^{\mbox{\tiny left}} + r Ag
\frac{\partial}{\partial {\bf x}}, \frac{\partial}{\partial {\bf
x}}\}$ of $TQ$ then  $\{\vecL,\vecep\}$ is its dual basis on $T^*Q$.
Let us denote by $(g,{\bf x} ; {\bf K}, \tilde {\bf p} )$ the
coordinates on $T^*Q$ adapted to this basis.

\subsubsection*{The Poisson structure $\Lambda_0$ on
$W^\circ /G$.} The submanifold $W^\circ$ of $T^*Q$ is represented,
in coordinates, by $(g,{\bf x} ; {\bf K},0)$, and the $G$-action on
$W$ is given by $(h,a): (g,{\bf x}; {\bf K}) \mapsto (hg, h{\bf x} +
a ; {\bf K}).$ Thus,  $(\vecgamma, {\bf K})$ are local coordinates
of the reduced manifold $W^\circ /G \simeq S^2 \times \R^3$. Since
$\Theta_{\mbox{\tiny{$W^{\!\circ}$}}} = {\bf K} \cdot \vecL$, then
$$
\Omega_{\mbox{\tiny{$W^{\!\circ}$}}} = \lambda_i \wedge  K_i - {\bf
K} \cdot d \vecL =  \lambda_i \wedge  K_i + K_1 \lambda_2\wedge
\lambda_3 + K_2 \lambda_3\wedge \lambda_1 + K_3 \lambda_1\wedge
\lambda_2.
$$
Since $W$ satisfies the vertical-symmetry condition, the reduction
of the presymplectic manifold $(W^\circ, \Omega_{W^\circ})$ gives
the Poisson bracket $\Lambda_0$ on $W^\circ\!/G$ described in
coordinates $(\vecgamma, {\bf K})$ as
\begin{equation} \label{Ex:pi0}
\{K_i,K_j\}_{\Lambda_0}= -\epsilon_{ijl}K_l \ \ \quad \{\gamma_i, K_j\}_{\Lambda_0} = - \epsilon_{ijl}\gamma_l \ \ \quad
\{\gamma_i,\gamma_j\}_{\Lambda_0} = 0.
\end{equation}

\subsubsection*{The Poisson structure $\Lambda$ on $\M/G$
via the isomorphism $\Psi:\M \to W^\circ$.} The lagrangian is given
by the kinetic energy $$\mathcal{L}(g,{\bf x}; \vecOm , \dot {\bf
x}) = \kappa((\vecOm , \dot {\bf x}), (\vecOm , \dot {\bf x})) =
\frac{1}{2}(\mathbb{I} \vecOm) \cdot  \vecOm  + \frac{m}{2}||\dot
{\bf x}||^2,$$ where $\mathbb{I}$ is the inertia tensor which is
represented as a diagonal $3\times 3$ matrix whose positive entries,
$\mathbb{I}_1, \mathbb{I}_2, \mathbb{I}_3$ are the principal moments
of inertia.

As it was computed in \cite{PL2011} and \cite{Naranjo2008},
$$\M =
\{(g,{\bf x} ; {\bf K}, \tilde {\bf p} ) \in T^*Q \ : \ {\bf K} =
\mathbb{I}\vecOm + mr^2 (Ag)^T (Ag) \vecOm, \qquad  \tilde {\bf p}=
m r A g \vecOm \}.$$ Therefore, since we consider coordinates
adapted to the constraints, the isomorphism $\Psi:\M \to W^\circ$ is
written in coordinates as $\Psi(g,{\bf x} ; {\bf K},\tilde {\bf p})
= (g,{\bf x} ; {\bf K}, 0)$. On the reduced manifold $\M/G$, with
local coordinates $(\vecgamma,{\bf K})$, the Poisson structure
$\Lambda$ coincides with the formulas given in \eqref{Ex:pi0}, for
$\W$ as in \eqref{Eq:W}.

\begin{remark}
 At this point,
we can check whether the reduced nonholonomic vector field belongs
to the characteristic distribution of $\Lambda$. In fact, since
(\cite{PL2011,Naranjo2008})
$$\displaystyle{ X^\nh_\red = \vecgamma \times \vecOm \cdot
\frac{\partial}{\partial \vecgamma} + {\bf K} \times \vecOm \cdot
\frac{\partial}{\partial {\bf K}}},
$$  we
see that $\Lambda^\sharp(\vecOm \cdot d{\bf K} ) = X^\nh_\red$. Even
though $\vecOm \cdot  d{\bf K}$ is not necessarily a closed 1-form
(observe that it depends on the matrix $A$ by the relation between
${\bf K}$ and $\vecOm$), the fact that $X_\red^\nh$ belongs to
$\Lambda^\sharp(T^*(\M/G))$ says that it makes sense to think that
the reduced dynamics may be described by a gauge transformation of
$\Lambda$.
\end{remark}

\subsubsection*{The 2-form $\langle {\mathcal J}, \mathcal{K}_\subW \rangle$.}
The projection $P_{\mbox{\tiny{$W$}}}: T\M \to \W$ is given by $P_{\mbox{\tiny{$W$}}} = \epsilon_i \otimes \frac{\partial}{\partial x^i}$. Thus, recalling that $\mathfrak{g} = \R \times \R^3$ we have that the 
map $\mathcal{A}_\subW: T\M \to \mathfrak{g}$, defined in \eqref{Eq:AkinPw},  is  $\mathcal{A}_\subW = (0, \vecep)$.
Using Definition \ref{Def:KinCurvature}, the $\W${\it-curvature} $\mathcal{K}_\subW$ is given by
$$\mathcal{K}_\subW |_\C= d\mathcal{A}_\subW |_\C = (0, - r Ag \cdot  d\vecL ) |_\C ,$$ where $d\vecL = - (\lambda_2 \wedge \lambda_3, \lambda_3 \wedge \lambda_1,  \lambda_1 \wedge \lambda_2)$.

In order to compute the momentum map, let us call $J_0 \in C^\infty(\M)$ such that $J_0= {\bf i}_{(1;{\bf 0})_\M}\Theta_\subM$. 
Since an element $\a \in \M \subset T^*Q$ can be written in local coordinates as $\a_q = {\bf K}^T \vecL + mr (Ag\vecOm )^T \vecep$, then  $\langle {\mathcal J}(\alpha_q), (0; {\bf e}_i) \rangle = \alpha_q (\frac{\partial}{\partial x^i})= m r  [(g\vecOm)^TA^T]_i$. Hence the momentum map ${\mathcal J} \in \Gamma(\M \times \mathfrak{g}^*)$ written with respect to the basis $\{(1;{\bf 0}), (0;{\bf e}_i) \}$ of $\mathfrak{g}$ is
$$
{\mathcal J} = (J_0; mr (g\vecOm)^TA^T).
$$
Therefore, the 2-form $\langle {\mathcal J}, \mathcal{K}_\subW \rangle$ on $\M$  is given by
\begin{equation} \label{Ex:RB-JK} \langle {\mathcal J}, \mathcal{K}_\subW \rangle |_\C =   r^2 m (g \vecOm )^T A^TAg \cdot \vecL \times \vecL |_\C \\
\end{equation}

Now, it is possible to compute the Jacobiator of $\pi_\nh$ and
$\pi^\nh_\red$ by Theorem \ref{T:JacobVerticalSymm}. In order to see
if the reduced bivector $\pi^\nh_\red$ has an integrable
characteristic distribution  (or if there is a dynamical gauge
transformation such that $\pi_\red^\B$ has an integrable
characteristic distribution) it is necessary to analyze whether the
2-from  $\langle {\mathcal J}, \mathcal{K}_\subW \rangle$ is basic
with respect to the orbit projection $\rho:\M \to \M/G$ and as a
result we will be able to explain the dynamical gauge
transformations chosen in \cite{Naranjo2008,PL2011}.

\noindent {\it If $A \equiv 0$}. In this case, $\langle {\mathcal
J}, \mathcal{K}_\subW \rangle \equiv 0$ and thus $[\pi_\nh,
\pi_\nh]=0$ since $\psi_{\pi_\nh} =0$ (observe that
$\mathcal{K}_\subW \equiv 0$). This is coherent with the fact that
$A \equiv 0$ describes the free rigid body.

\noindent {\it If $A$ has rank 1}.  In this case, $A^T A=A_3$ with
$A_3 = {\bf e}_3 \, {\bf e}_3^T $,  where ${\bf e}_3$ is the third
canonical vector in $\R^3$. Then, using also that $d\vecgamma =
\vecgamma \times \vecL$ \cite{Naranjo2008,PL2011} the 2-form in
\eqref{Ex:RB-JK} becomes
$$
\langle {\mathcal J}, \mathcal{K}_\subW \rangle |_\C =  mr^2 (g \vecOm )^T
A_3 g \cdot \vecL \times \vecL |_\C = r^2 m \langle \vecgamma , \vecOm
\rangle \vecgamma  \cdot \vecL \times \vecL |_\C =  r^2 m \langle
\vecgamma , \vecOm \rangle \vecgamma  \cdot d\vecgamma \times
d\vecgamma |_\C,
$$
which is basic. From Corollary \ref{C:TwistedVertSymm},  the
reduction of $\pi_\nh$ induces a bivector $\pi_\red^\nh$ on $\M/G$
that is $(-d\langle {\mathcal J}, \mathcal{K}_\subW
\rangle_\red)$-twisted where $\langle {\mathcal J},
\mathcal{K}_\subW \rangle_\red = r^2 m \langle \vecgamma , \vecOm
\rangle \vecgamma  \cdot d\vecgamma \times d\vecgamma \, \in \,
\Omega^2(\M/G)$  (which explains the result in \cite[Theorem
6]{PL2011}). In particular, it has an integrable characteristic
distribution. Moreover, following Theorem \ref{T:Reduced-Dyn}, the
bivector $\pi_\red^\nh$ is a gauge transformation of  $\Lambda$
(given in \eqref{Ex:pi0}) by the 2-form $ \langle {\mathcal J},
\mathcal{K}_\subW \rangle_\red$.

\medskip

\noindent {\it If $A$ has rank 2}. Since $A^TA= \textup{Id}- A_3$
where $\textup{Id}$ is the identity matrix, then
$$\langle {\mathcal J}, \mathcal{K}_\subW \rangle |_\C = \left(r^2 m \langle \vecOm , \vecL \times \vecL \rangle -  r^2 m \langle \vecgamma , \vecOm \rangle \vecgamma  \cdot d\vecgamma \times d\vecgamma \right) |_\C.$$
In this case, $d\langle {\mathcal J}, \mathcal{K}_\subW \rangle $ is
not basic. Note that one can conclude more: there is no closed 3-form $\phi$ on $\M/G$ making the reduced bivector field $\pi_\red^\nh$ twisted Poisson. In fact, if such $\phi$ existed then, by Corollary \ref{C:TwistedVertSymm}, we would have $d\langle {\mathcal J}, \mathcal{K}_\subW \rangle |_\C = -\rho^*\phi |_\C$, which would imply that  ${\bf i}_X d\langle {\mathcal J}, \mathcal{K}_\subW \rangle = 0$ for any $X \in \Gamma(\S)$. However one can verify that this is not the case. As a consequence, since $\pi_\red^\nh$ is regular (\cite{PL2011}), it follows that its characteristic distribution is not integrable, in agreement with \cite{Naranjo2008}.


Let us consider the 2-form $B$ on $\M$ given by $B=-r^2 m \langle
\vecOm , \vecL \times \vecL \rangle $, which is the non-basic
part of $\langle {\mathcal J}, \mathcal{K}_\subW \rangle$.  Observe
that  $B$ satisfies the conditions of Definition
\ref{D:dynamicalGauge} and thus $B$ defines a dynamical gauge for
the bivector $\pi_\nh$ and the hamiltonian $\Ham_\subM = \frac{1}{2}
{\bf K} \cdot \vecOm$ \ \cite[eq.(48)]{PL2011}. Moreover, since
$B$ is $G$-invariant, it defines an invariant bivector field
$\pi_\B$ that describes the dynamics. Following Corollary
\ref{C:TwistedVertSymm}, the reduction of $\pi_\B$ gives a
$(-d\mathcal{B}_\red)$-twisted Poisson bivector field $\pi_\red^\B$, where $\rho^*\mathcal{B}_\red=   \langle {\mathcal J}, \mathcal{K}_\subW \rangle + B$, i.e.,
$$\mathcal{B}_\red=
- r^2 m \langle \vecgamma , \vecOm \rangle \vecgamma  \cdot
d\vecgamma \times d\vecgamma .
$$
So the reduced bivector $\pi_\red^\B$ has an integrable
characteristic distribution and, moreover, $\pi_\red^\B$ is gauge
related to the Poisson bivector field $\Lambda$, given in
\eqref{Ex:pi0}, by the 2-from $\mathcal{B}$. Again, here we are
recovering the results in \cite[Proposition 2 and Theorem
5]{PL2011}.

\medskip

\noindent {\it If $A$ has rank 3}. Finally $A^TA= \textup{Id}$, then
$\langle {\mathcal J}, \mathcal{K}_\subW \rangle |_\C =   r^2 m \langle
\vecOm , \vecL \times \vecL \rangle|_\C.$
In this case, $\langle {\mathcal J}, \mathcal{K}_\subW \rangle$ is not basic but it defines a dynamical gauge,
that is ${\bf i}_{X_\nh}\langle {\mathcal J}, \mathcal{K}_\subW \rangle=0$.
Thus, the gauge transformation of $\pi_\nh$ by $B=-\langle {\mathcal J}, \mathcal{K}_\subW \rangle$
gives a bivector $\pi_\B$ describing the dynamics.
Moreover, $\pi_\B$ coincides with $\pi_{\mbox{\tiny ${\mathcal J}\! \mathcal{K}$}}$.
So, the Poisson bracket $\Lambda$ on $\M/G$ describes the reduced dynamics
(this result recovers \cite[Prop.~1]{PL2011}).

\subsubsection*{Conserved quantities.} For
$\eta=(\mu;(a,b,c)) \in \mathfrak{g}$ and $\alpha \in \M$,  we have that $
\langle {\mathcal J}^\nh(\alpha), P_{\mathfrak{g}_\subS}(\eta) \rangle  =  \langle {\mathcal J}^\nh (\alpha)  , \mu(1;-y,x,0)\rangle =  \mu \langle {\bf K}, \vecgamma \rangle$ since $(1;-y,-x,0)_\subM =  \vecgamma \cdot ( {\bf X}^{\mbox{\tiny left}} + r Ag \frac{\partial}{\partial {\bf x}} )$. 
Since $\langle {\bf K}, \vecgamma \rangle$ is a
basic function, then it is a Casimir of $\Lambda$ according to
Corollary \ref{C:LeavesOfPiP}.  Since, for any rank of $A$, the
reduced dynamics is described by a bivector field $\pi_\red^\B$ that
is gauge related with $\Lambda$, $\langle {\bf K},  \vecgamma
\rangle $ is also a Casimir for $\pi_\red^\B$. Thus, as Corollary
\ref{C:ConservedQuantity} asserts,  $\langle {\bf K}, \vecgamma
\rangle$ is a conserved quantity of the nonholonomic system (in
agreement with \cite{Naranjo2008}).

\subsubsection*{Integrability.} As it was analyzed in \cite{PL2011} the rigid body with constraints is integrable for the three cases studied here. In fact, for each case, the bivector field describing the dynamics that is gauge related with $\Lambda$ (i.e., $\pi_\red^\nh$ in the case of $rank(A)=1$ and $\pi_\red^\B$ in the cases of $rank(A)=2,3$ for the corresponding 2-forms $B$) not only has an integrable characteristic distribution but is conformally Poisson  \cite[Sec.5.2]{PL2011}. Therefore, there is a Poisson bracket with 4-dimensional symplectic leaves and two first integrals in involution on each leaf: the hamiltonian $\Ham_\red$ and the function $F=  {\bf K} \cdot {\bf K}$ (whose joint level sets are compact on the leaf). Arnold-Liouville's Theorem guarantees the integrability of the reduced equations on each leaf (after the time reparametrization using the conformal factor).

\subsubsection*{The nonholonomic bracket.} Observe that it
is not necessary to compute the nonholonomic bracket $\pi_\nh$ in
order to obtain its Jacobiator and draw conclusions about the
integrability of the characteristic distribution associated to the
reduced nonholonomic bracket.  Note that the 2-section $\Omega_\subM
|_\C$ and the distribution $\C$ defining $\pi_\nh$ are computed to
be $\C = \textup{span}  \left\{  {\bf X}^{\mbox{\tiny left}} + r Ag
\frac{\partial}{\partial {\bf x}} , \frac{\partial}{\partial {\bf
K}} \right\}$ and
$$\Omega_\subM |_\C = \Psi^*\Omega_{\mbox{\tiny{$W^{\!\circ}$}}} |_\C - \langle {\mathcal J}, {\mathcal K}_\subW \rangle |_\C = \left( \lambda_i \wedge d K_i - {\bf K} \cdot d \vecL \right) |_\C - \langle {\mathcal J}, {\mathcal K}_\subW \rangle|_\C. $$
From here we can recover the formula for the nonholonomic bracket given in \cite[Prop.~15]{PL2011}.

\subsection{The homogeneous ball rolling on a cylinder} \label{Ex:Cylinder}

Let us consider a homogeneous sphere of radius $r$ rolling without sliding on the inner surface of a circular cylinder of radius $R$. We denote by $m$ the mass of the ball and by $\mathbb{I}$ its moment of inertia (which is a multiple of the identity, i.e., $\mathbb{I}=I.\textup{Id}$ where $I \in \R$ and $\textup{Id}$ is the identity matrix).  In what follows we assume that the symmetry axis of the cylinder is the vertical axis. We follow the notation of \cite{Marle2003,BGMConservation}.

The configuration space of the system is the manifold $\R \times S^1 \times SO(3)$ with (local) coordinates $(z, \theta, g)$, where $(z,\theta)$ represent the cylindrical coordinates of the center of mass of the ball and $g$ is the orientation matrix of the ball with respect to fixed axis. 
The lagrangian $\mathcal{L}:TQ \to \R$ is given by
$$
L(z,\theta,g; \dot z, \dot \theta, \vecom) = \frac{m}{2} \left( (R-r)^2 \dot \theta + \dot z^2\right) + \frac{I}{2} \omega_ i^2 + m {\bf g} z, 
$$
where $\vecom = (\omega_1, \omega_2, \omega_3)$ is the spatial angular velocity of the ball and ${\bf g}$ is the gravity acceleration. The non sliding constraints are 
$$
\dot z = r(\omega_2\cos \theta - \omega_1 \sin \theta) \qquad \mbox{and} \qquad \dot \theta = -\frac{r}{(R-r)} \omega_3.
$$

Following \cite{Marle2003}, we perform a change of coordinates by rotating the frame ${\bf X}^{\mbox{\tiny{right}}} \! = (X^{\mbox{\tiny{right}}}_1 \!, X^{\mbox{\tiny{right}}}_2\!, X^{\mbox{\tiny{right}}}_3)$ by the angle $\theta$ around the vertical axis. Hence we obtain a set of vector fields $\{X_n, X_\theta, X_z\}$ attached to the cylinder so that $R_{\theta} {\bf X}^{\mbox{\tiny{right}}} = (X_n, X_\theta, X_z)^T$, where $R_{\theta}$ the matrix representing the rotation by the angle $\theta$ around the vertical axis.  

Let us consider the basis of $TQ$ given by $\{\frac{\partial}{\partial z}, \frac{\partial}{\partial \theta}, X_n, X_\theta, X_z\}$ and its dual basis $\{dz, d\theta,\beta_n, \beta_\theta, \beta_z\}$ of $T^*Q$ with associated coordinates $(p_z, p_\theta, M_n, M_\theta, M_z)$. The submanifold $\M \subset T^*Q$ is given by 
$$
\M=\{(z,\theta,g, p_z, p_\theta, M_n, M_\theta, M_z ) \ : \ M_\theta = \frac{I}{rm}p_z,    \qquad  M_z = - \frac{I}{rm(R-r)} p_\theta  \}.
$$
The constraint 1-forms are written as
$$
\epsilon_z = \beta_z + \frac{(R-r)}{r} d\theta \qquad \mbox{and} \qquad \epsilon_\theta = \beta_\theta - \frac{dz}{r},
$$ 
and the distribution $D$ is generated by the basis $\{Y_z := \frac{\partial}{\partial z} + \frac{ X_\theta}{r}, Y_\theta := \frac{\partial}{\partial \theta} - \frac{(R-r)}{r} X_z, X_n\}$.

\subsubsection*{The nonholonomic bracket and the dynamics.}
Using \eqref{Eq:bivector-section} we compute the nonholonomic bracket $\pi_\nh$ on $\M$. If we denote $E= I+r^2m$, we obtain
\begin{equation*}
 \begin{split}
\pi_\nh = & \frac{r^2m}{E} Y_\theta \wedge \frac{\partial}{\partial p_\theta} + \frac{r^2m}{E} Y_z \wedge \frac{\partial}{\partial p_z} + X_n \wedge \frac{\partial}{\partial M_n}  + \frac{r^2m^2R}{E^2} M_n \frac{\partial}{\partial p_\theta} \wedge \frac{\partial}{\partial p_z}  + \frac{RI}{E} p_z \frac{\partial}{\partial M_n} \wedge \frac{\partial}{\partial p_\theta} \\
& + \frac{I}{E(R-r)}p_\theta \frac{\partial}{\partial p_z} \wedge\frac{\partial}{\partial M_n}.
\end{split}
\end{equation*}

Since the hamiltonian $\Ham_\subM:\M \to \R$ is
$$
\Ham_\subM(z,\theta,g; p_z, p_\theta, M_n) = \frac{E}{2r^2m^2} \left(\frac{1}{(R-r)^2}  p_\theta^2  +  p_z^2 \right) + \frac{M_n^2 }{2I} + m{\bf g} z,
$$
the nonholonomic vector field is given by 
$$
X_\nh = \frac{p_\theta}{m(R-r)^2} Y_\theta + \frac{p_z}{m}Y_z - \frac{M_n}{I} X_n - \frac{rp_\theta M_n }{E(R-r)^2} \frac{\partial}{\partial p_z} + \frac{Ip_\theta p_z}{rm^2(R-r)^2}  \frac{\partial}{\partial M_n} -  \frac{{\bf g}m^2 r^2}{E} \frac{\partial}{\partial p_z}.
$$

At this point, we observe that $p_\theta$ is a conserved quantity of the dynamics as it was already observed in \cite{BGMConservation,Marle2003}. 
Now we use Corollary \ref{C:RedJacobiator} in order to see whether the reduced bracket $\pi_\nh^\red$ admits leaves.

\subsubsection*{Reduction and the $d{\mathcal J} \wedge {\mathcal K}_\subW$ 3-form.}
The system admits the symmetry group $G = S^1 \times SO(3)$ (see \cite{BGMConservation,Marle2003}) with the action on $Q$ given by 
$$\left( (\varphi, h) , (z,\theta,g) \right) \mapsto (z, \theta+\varphi, L(g_\varphi) g h),
$$
where  $g_\varphi$ is the rotation of angle $\varphi$ around the vertical axis, and $L(g_\varphi):  SO(3) \to SO(3)$ maps $h \in SO(3)$ in $g_\varphi h$. The canonical lift of the action to $T^*Q$ is 
$$
\left( (\varphi, h) , (z,\theta,g; p_z,p_\theta, {\bf M}) \right) \mapsto (z, \theta+\varphi, L(g_\varphi) g h ; p_z, p_\theta, \hat{L}(g_\varphi){\bf M}) ),
$$
where ${\bf M} = (M_n,M_\theta,M_z)$ and  $\hat{L}(g_\varphi)$ is the cotangent lift of the left translation $L(g_\varphi)$.  
Observe that the coordinates ${\bf M}$ do not see the right action by $h$ since they are right invariant.

The projection $T^*Q \to T^*Q/G$ is given by $(z,\theta, g ; p_z, p_\theta, {\bf M}) \mapsto (z, p_z, p_\theta , {\bf M}_0)$, where ${\bf M}_0 = Ad^*_{{g_\theta}^{\!-1}} {\bf M}$. Then, for ${\bf M}_0 = (M_{0x}, M_{0y},M_{0z})$  we obtain 
$$
\M /G = \{ (z, p_z, p_\theta,{\bf M}_0) \ : \ M_{0y} = \frac{I}{rm} p_z, \ \ M_{0z} = -\frac{I}{rm(R-r)} p_\theta \}.
$$
Therefore, the reduced manifold $\M/G$ is (locally) represented by the coordinates $(z, p_z, p_\theta,M_{0x})$, with $\rho^*M_{0x} = M_n$ for $\rho: \M \to \M/G$ the orbit projection (see \cite{Marle2003}). 

In what follows we compute the 3-form $d{\mathcal J} \wedge {\mathcal K}_\subW$ on $\M$ in order to apply Corollary \ref{C:RedJacobiator} and obtain the Jacobiator of $\pi_\red$. 

\noindent {\bf The $\W$-curvature $\mathcal{K}_\subW$.} \ The vertical space $V$ is given by $V = \textup{span} \{Y_\theta, X_n,X_\theta, X_z\}$, since for each element of the Lie algebra $\mathfrak{g} = \R \times \R^3$ of $G$, the infinitesimal generators are given by
$$
(1;{\bf 0})_{T^*Q}  = \frac{\partial}{\partial \theta} + X^{\mbox{\tiny{right}}}_3, \qquad (0; {\bf e}_i)_{T^*Q} = g^T {\bf X}^{\mbox{\tiny{right}}}, 
$$
where ${\bf e}_i$ for $i=1,2,3$ are the canonical vectors on $\R^3$. 
We consider the following decomposition of $TQ$ as in \eqref{Eq:DecompTQ}:
\begin{equation} \label{Ex:Cylind:decomp}
D = \textup{span}\{Y_z, Y_\theta, X_n\} \qquad \mbox{and} \qquad W = \textup{span} \{ X_\theta, X_z\} .
\end{equation}
Observe that $W$ is a vertical complement of the constraints but it does {\it not} satisfy the vertical-symmetry condition \eqref{Eq:VerticalSymm-Q}. 

The projection  $P_{\mbox{\tiny{$W$}}} :TQ \to W$ associated to the decomposition \eqref{Ex:Cylind:decomp} is given by $P_{\mbox{\tiny{$W$}}} = \epsilon_\theta \otimes X_\theta + \epsilon_z \otimes X_z$. Writing the vector fields $X_\theta$ and $X_z$ in terms of the infinitesimal generators, we compute $\mathcal{A}:TQ \to \mathfrak{g}$ as in \eqref{Eq:AkinPw}, and the $\W$-curvature is given by 
$$\mathcal{K}_\subW = d \mathcal{A} |_\C =\left( 0,  \frac{R}{r}(- \vecalpha \sin \theta  + \vecbeta \cos \theta) \beta_n \wedge d\theta + \frac{\vecgamma}{r} \beta_n \wedge dz\right),
$$
where $\vecalpha$, $\vecbeta$ and $\vecgamma$ are the rows of the $3\times 3$ matrix $g$. 

\noindent {\bf The moment map $\mathcal{J}:T\M \to \mathfrak{g}^*$.} \ 
If $J_0 = {\bf i}_{(1,{\bf 0})_{T^*Q}} \Theta_Q$ then the canonical momentum map $J:T^*Q \to \mathfrak{g}^*$ is given by 
$ J = (J_0, M_n A + M_\theta B + M_z \vecgamma)$, where $A= \vecalpha \cos\theta + \vecbeta \sin\theta $ and $B= -\vecalpha \sin\theta + \vecbeta \cos\theta$.
Since ${\mathcal J} = J \circ \iota :T\M \to \mathfrak{g}^*$ (for $\iota:\M \to T^*Q$ the inclusion) then  
\begin{equation*}
 \begin{split}
  d{\mathcal J} |_\C = & \left( dJ_0, A\, dM_n +\frac{I}{rm}B \, dp_z - \frac{I \vecgamma}{rm(R-r)} dp_\theta 
  + (\frac{M_n}{r} \vecgamma + \frac{Ip_\theta}{r^2m(R-r)}  A ) dz  
   \right. \\ 
  & \left.   - \frac{I}{rm}(p_z \vecgamma + \frac{ p_\theta}{(R-r)} B)\beta_n 
 + \frac{R}{r} (M_nB - \frac{I}{rm}p_zA) d\theta \right).
 \end{split}
\end{equation*}

The 3-form $d{\mathcal J} \wedge {\mathcal K}_\subW$ is computed as in \eqref{Eq:dJ^K}:
\begin{equation} \label{Ex:Cylind:dJ^K}
d{\mathcal J} \wedge {\mathcal K}_\subW = \frac{RI}{r^2m} dp_z \wedge \beta_n \wedge d\theta - \frac{I}{r^2m(R-r)} dp_\theta \wedge \beta_n\wedge dz.
\end{equation}
We observe that this 3-form is not basic and, moreover, 
$
d{\mathcal J} \wedge {\mathcal K}_\subW (\pi_\nh^\sharp(d M_n), \pi_\nh^\sharp(\rho^*dz), \pi_\nh^\sharp(\rho^*dp_\theta) ) \neq 0.
$
By Corollary \ref{C:RedJacobiator} we conclude that $\pi^\nh_\red$ is not Poisson nor admits leaves. 

In order to obtain a (Poisson) bracket describing the reduced dynamics, we will study a gauge transformation of the nonholonomic bivector $\pi_\nh$. 

\subsubsection*{The gauge transformation of $\pi_\nh$.}
Since the conserved quantity $p_\theta$ is a basic function, it is conserved also by the reduced dynamics $X^\nh_\red = T\!\rho X_\nh$.  However, $p_\theta$ is not conserved by every hamiltonian vector field associated to $\pi^\nh_\red$ (i.e., $(\pi^\nh_\red)^\sharp(dp_\theta) \neq 0$).
In what follows, we perform a gauge transformation of $\pi_\nh$ so that  $p_\theta$  is a Casimir of $\pi_\red^\B$ (and thus,  $p_\theta$ is conserved by every hamiltonian vector field associated to $\pi_\red^\B$). 

Let $B$ be the $G$-invariant 2-form on $\M$ given by 
$$
B = \frac{R}{r^2} (M_n dz\wedge d\theta + \frac{I}{m}p_z d\theta \wedge \beta_n + \frac{I}{m(R-r)^2} p_\theta \beta_n \wedge dz).
$$

First, observe that $B$ satisfies the condition in Remark \ref{R:DynGaugeSemi-basic} and also that ${\bf i}_{X_\nh} B = 0$, hence the 2-form $B$ defines a dynamical gauge transformation.  
If we compute the bivector $\pi_\B$ associated to the 2-form $\Omega_\subM +B$ and the distribution $\C$ we obtain 
\begin{equation*}
\pi_\B =  \frac{r^2m}{E} Y_\theta \wedge \frac{\partial}{\partial p_\theta} + \frac{r^2m}{E} Y_z \wedge \frac{\partial}{\partial p_z} + X_n \wedge \frac{\partial}{\partial M_n}  - \frac{Ir}{E(R-r)^2} p_\theta \frac{\partial}{\partial p_z} \wedge\frac{\partial}{\partial M_n}.
\end{equation*}

Note that $\pi_\B^\sharp(dp_\theta) = -\frac{r^2m}{E} Y_\theta$, hence $p_\theta$ is a Casimir of $\pi_\red^\B$.  Thus, since $\textup{dim}\, \M/G = 4$, we conclude that, if $\pi_\red^\B$ admits leaves, it will be automatically Poisson. 

\noindent {\bf The Jacobiator of $\pi_\red^\B$}. \ Using Corollary \ref{C:RedJacobiator} and formula \eqref{Ex:Cylind:dJ^K} we get that
\begin{equation} \label{Ex:Cylind:dJ^K-dB}
\left( d{\mathcal J} \wedge {\mathcal K}_\subW + dB \right) |_\C = \frac{I}{rm(R-r)^2} dp_\theta \wedge \beta_n \wedge dz + \frac{R}{r^2} dM_n\wedge dz \wedge d\theta,
\end{equation}
where we used the fact that $d\beta_n \wedge d\theta |_\C = d\beta_n\wedge dz|_\C = 0$. Again, note that the 3-form \eqref{Ex:Cylind:dJ^K-dB} is non-basic.   However, we observe that ${\bf i}_{\pi_\B^\sharp (\rho^*dz)} \left( d{\mathcal J} \wedge {\mathcal K}_\subW + dB \right) |_\C = 0$ and also 
$$
\left( d{\mathcal J} \wedge {\mathcal K}_\subW + dB \right) \, (\pi_\B^\sharp (\rho^*dp_\theta), \pi_\B^\sharp (\rho^*d p_z) , \pi_\B^\sharp (dM_n)) = 0.
$$
Therefore, $\pi_\B^\sharp (d{\mathcal J} \wedge {\mathcal K}_\subW - dB ) \circ \rho^* = 0$ and hence by Corollary \ref{C:RedJacobiator} the bivector field $\pi_\red^\B$ is Poisson.  

Finally we conclude that $\pi_\red^\B$ is a Poisson bivector field on $\M/G$ describing the reduced dynamics of the homogeneous ball rolling on a cylinder.  
In order to complete the example, we write explicitly the reduced bivector $\pi_\red^\B$:
\begin{equation*}
\pi_\red^\B =  \frac{r^2m}{E} \frac{\partial}{\partial z} \wedge \frac{\partial}{\partial p_z}  - \frac{Ir}{E(R-r)^2} p_\theta \frac{\partial}{\partial p_z} \wedge\frac{\partial}{\partial M_{0x}}.
\end{equation*}

This bivector field agrees with the one in \cite{BoMaKi2002,Ramos2004}, where it appears with no geometric interpretation.

\appendix

\section{Appendix}

Let $(R, \Omega)$ be a symplectic manifold equipped with a free and
proper action of a Lie group $G$ preserving $\Omega$. There is an
induced Poisson structure on $R/G$ by reduction, that we denote by
$\pi_{\mbox{\tiny{$R/G$}}}$. We use the notation
$\mathcal{O}_{\mbox{\tiny{$R/G$}}}$ for its symplectic leaves and $\rho_{\scriptscriptstyle{R}}: R\to R/G$ for the quotient
map.

Let $\iota: P\hookrightarrow Q$ be a $G$-invariant submanifold of
$R$ and $\Omega_P:=\iota^*\Omega$. We denote by $\V_P\subseteq TP$
the distribution tangent to the orbits of the restricted action on
$P$, and by $\rho_{\scriptscriptstyle{P}}: P\to P/G$ the quotient
map.

\begin{lemma} \label{L:A-PresymplecticRed}
Suppose that $\textup{Ker} \, \Omega_P$ has constant rank and
$\textup{Ker} \, \Omega_P \subset \V_P$. Then:
\begin{enumerate}
\item[$(i)$] There is an induced Poisson structure $\pi_{\mbox{\tiny{$P/G$}}}$ on
$P/G$.
\item[$(ii)$] Each symplectic leaf $\mathcal{O}_{\mbox{\tiny{$P/G$}}}$   of
$\pi_{\mbox{\tiny{$P/G$}}}$ is a  symplectic submanifold of a leaf
$\mathcal{O}_{\mbox{\tiny{$R/G$}}}$ of $\pi_{\mbox{\tiny{$R/G$}}}$,
given by a connected component of the intersection
$\mathcal{O}_{\mbox{\tiny{$R/G$}}} \cap (P/G)$.
\end{enumerate}
\end{lemma}

Note that we have an induced embedding $\bar{\iota}:P/G
\hookrightarrow R/G$. Condition $(ii)$ says that $P/G$ is a {\it
Poisson-Dirac} submanifold of $R/G$, in the sense of
\cite[Sec.~8]{CrFe} (see also \cite[Sec.~6]{NotasDirac}). We will
refer to the Poisson structure $\pi_{\mbox{\tiny{$P/G$}}}$ in $(i)$
as the {\it reduction} of the presymplectic form $\Omega_P$.

\begin{proof}
$(i)$: Consider $C^\infty(P)_{adm}:=\{ f \in C^\infty(P)\,|\,
df(\textup{Ker} \, \Omega_P)=0\}$. The condition that $\textup{Ker}
\, \Omega_P$ has constant rank implies that $C^\infty(P)_{adm}$ is a
Poisson algebra with Poisson bracket given by
$$
\{f,g\} = - dg(X),
$$
where $X$ is any vector field on $P$ such that ${\bf i}_X\Omega_P =
df$ (such $X$ always exists).  If
$C^\infty(P)^G$ denotes the space of $G$-invariant functions on $P$,
then condition $\textup{Ker} \, \Omega_P \subset \V_P$ implies that
$C^\infty(P)^G \subseteq C^\infty(P)_{adm}$.

We now check that if
$f,g \in C^\infty(P)^G$, then $\{f,g\} \in C^\infty(P)^G$.
For any $\gamma\in G$, using that $\gamma^*\Omega_P=\Omega_P$ and
$\gamma^*f=f$, we see that if ${\bf i}_X\Omega_P=df$, then $\gamma_*
X - X \in \textup{Ker} \, \Omega_P$. Hence $dg(X) = dg(\gamma_*X)$
when $g\in C^\infty(P)^G$. From that we directly verify that
$\gamma^* \{f,g\} = -\gamma^* (dg(X))=-dg(X)=\{f,g\}$. So the
Poisson bracket on $C^\infty(P)_{adm}$ restricts to $C^\infty(P)^G$,
and there is an induced Poisson structure
$\pi_{\mbox{\tiny{$P/G$}}}$ on $P/G$ since
$\rho_{\scriptscriptstyle{P}}^*: C^\infty(P/G)\to C^\infty(P)^G$ is
an isomorphism. More explicitly,
\begin{equation} \label{Eq:PresymBivector}
\pi_{\mbox{\tiny{$P/G$}}}^\sharp(\alpha) = T\! \rho_{\!
\mbox{\tiny{$P$}}} (Z) \quad \mbox{if and only if} \quad  Z \in \mathfrak{X}(P) \mbox{ verifies that } {\bf i}_{Z} \Omega_P = -\rho_{\!
\mbox{\tiny{$P$}}}^*(\alpha).
\end{equation}

$(ii)$: By \cite[Sec.~8]{CrFe}, showing that the inclusion
$\bar{\iota}: (P/G, \pi_{\mbox{\tiny{$P/G$}}}) \hookrightarrow (R/G,
\pi_{\mbox{\tiny{$R/G$}}}) $ is a Poisson-Dirac submanifold (which
is equivalent to the statement in $(ii)$) is the same as showing
that the subbundle $\mathrm{graph}(\pi_{\mbox{\tiny{$P/G$}}}^\sharp)
= \{(\pi_{\mbox{\tiny{$P/G$}}}^\sharp(\alpha),\alpha)\,|\, \alpha\in
T^*(P/G)\}$ of $T(P/G)\oplus T^*(P/G)$ equals the subbundle
\begin{equation}\label{eq:bk}
\{ (Y, \bar{\iota}^*\a) \in T(P/G)\oplus T^*(P/G) \, | \, \a \in
T^*(R/G) \mbox{ s.t. } \pi_{\mbox{\tiny{$R/G$}}}^\sharp(\a) =
T\bar{\iota}(Y) \}.
\end{equation}
Since both subbundles have the same rank (equal to the dimension of
$P/G$, as both are Dirac structures, see \cite{NotasDirac}), it
suffices to verify that the subbundle \eqref{eq:bk} is contained in
$\mathrm{graph}(\pi_{\mbox{\tiny{$P/G$}}}^\sharp)$.

Let $\pi$ be the bivector field corresponding to the 2-form $\Omega$
on $R$. Given $(Y, \bar{\iota}^*\a)$ in the subbundle \eqref{eq:bk}
at
$\rho_{\scriptscriptstyle{P}}(x)=\rho_{\scriptscriptstyle{R}}(x)$,
for $x\in P\subseteq R$, we know that $T\! \rho_{\!
\mbox{\tiny{$R$}}} (\pi^\sharp (\rho_{\! \mbox{\tiny{$R$}}}^*\a) ) =
\pi^\sharp_{\mbox{\tiny{$R/G$}}}(\a) = T\bar{\iota}(Y)$, where
$\rho_{\! \mbox{\tiny{$R$}}}^*\a$ is taken at $x$. For $Z \in T_xP$
such that $T\! \rho_{\! \mbox{\tiny{$P$}}} Z = Y$, we see that
$\pi^\sharp (\rho_{\! \mbox{\tiny{$R$}}}^*\a) - T\iota(Z)$ is
vertical (i.e., tangent to a $G$-orbit), and hence must be tangent
to $P$ since $P$ is $G$-invariant. It follows that $\pi^\sharp
(\rho_{\! \mbox{\tiny{$R$}}}^*\a)=T\iota(X)$, for $X\in T_xP$, and
$T\bar{\iota}(Y)= T\! \rho_{\! \mbox{\tiny{$R$}}} (T\iota(X))=
T\bar{\iota}(T\! \rho_{\! \mbox{\tiny{$P$}}}(X))$, i.e., $Y = T\!
\rho_{\! \mbox{\tiny{$P$}}} (X)$. Finally note that
$$
{\bf i}_{X} \Omega_P = \iota^* {\bf i}_{\pi^\sharp (\rho_{\!
\mbox{\tiny{$R$}}}^*\a)} \Omega = -\iota^* \rho_{\!
\mbox{\tiny{$R$}}}^*\a = -\rho_{\!
\mbox{\tiny{$P$}}}^*(\bar{\iota}^*\a),
$$
which says that $Y = T\! \rho_{\! \mbox{\tiny{$P$}}} (X)=
\pi_{\mbox{\tiny{$P/G$}}}^\sharp(\bar{\iota}^* \alpha)$, that is,
$(Y,\bar{\iota}^*\a) \in
\mathrm{graph}(\pi^\sharp_{\mbox{\tiny{$P/G$}}})$.

\end{proof}

\begin{remark}\label{rm:ch}
In the special case when $\mathrm{Ker} \Omega_P = \V_P$, then the
reduced Poisson $\pi_{\mbox{\tiny{$P/G$}}}$ is defined by a
nondegenerate 2-form $\Omega_{\mbox{\tiny{$P/G$}}}$. Indeed, we can
check that $\pi_{\mbox{\tiny{$P/G$}}}$ is nondegenerate by noticing
that, if $\pi_{\mbox{\tiny{$P/G$}}}^\sharp(\alpha)=0$, then there is
a vector field $Z \in \V_P = \mathrm{Ker}  \Omega_P$ such that ${\bf
i}_Z \Omega_P = -\rho_{\! \mbox{\tiny{$P$}}}^*(\alpha)=0$, hence
$\alpha=0$; that is, $\pi_{\mbox{\tiny{$P/G$}}}^\sharp$ is an
isomorphism. Note also that  $\Omega_{\mbox{\tiny{$P/G$}}}$ is
uniquely determined by the condition that
\begin{equation}\label{Eq:pb}
\rho_{\!
\mbox{\tiny{$P$}}}^*\Omega_{\mbox{\tiny{$P/G$}}} = \Omega_P.
\end{equation}
To verify this last claim, recall that $\pi_{\mbox{\tiny{$P/G$}}}^\sharp(\alpha) =  T\! \rho_{\!
\mbox{\tiny{$P$}}} (Z)$, where $Z$ is any vector field on $P$ such
that ${\bf i}_{Z} \Omega_P = -\rho_{\!
\mbox{\tiny{$P$}}}^*(\alpha)$, and that ${\bf i}_{T\! \rho_{\!
\mbox{\tiny{$P$}}} (Z)} \Omega_{\mbox{\tiny{$P/G$}}} = -\alpha$. It follows that
$\rho_{\!
\mbox{\tiny{$P$}}}^*({\bf i}_{T\! \rho_{\!
\mbox{\tiny{$P$}}} (Z)} \Omega_{\mbox{\tiny{$P/G$}}}) = - \rho_{\!
\mbox{\tiny{$P$}}}^*\alpha = {\bf i}_{Z} \Omega_P$, which shows that \eqref{Eq:pb} holds.
\end{remark}

Suppose now that the $G$-action on $R$ is hamiltonian, with momentum
map $\mathcal{J}_R: R\to \mathfrak{g}^*$, and let ${\mathcal J}_P =
\iota^* {\mathcal J}_R$. For $\mu \in \mathfrak{g}^*$, let $G_\mu$
be the stabilizer of $\mu$ with respect to the coadjoint action, and
consider the Marsden-Weinstein quotients $({\mathcal J}_R^{-1} (\mu)
/ G_\mu, \omega_\mu)$, whose connected components are the symplectic
leaves of $\pi_{\mbox{\tiny{$R/G$}}}$, see e.g. \cite{MMORR}. Consider
the natural embeddings
\begin{equation}\label{Eq:incl}
\bar{\iota}: {\mathcal J}_P^{-1} (\mu) / G_\mu \hookrightarrow
{\mathcal J}_R^{-1} (\mu) / G_\mu.
\end{equation}

\begin{corollary}\label{C:Asl}
The symplectic leaves of $\pi_{\mbox{\tiny{$P/G$}}}$ are given by
the connected components of ${\mathcal J}_P^{-1} (\mu) / G_\mu$, for
$\mu\in \mathfrak{g}^*$, with symplectic form
$\bar{\iota}^*\omega_\mu$.
\end{corollary}

\begin{proof}
The leaves of $\pi_{\mbox{\tiny{$R/G$}}}$ are given by the connected
components of ${\mathcal J}^{-1}_{\mbox{\tiny{$R$}}}(\mu) /G_\mu$,
for $\mu\in \mathfrak{g}^*$, and by
Lemma~\ref{L:A-PresymplecticRed}, $(ii)$, we have that the leaves of
$\pi_{\mbox{\tiny{$P/G$}}}$ are the connected components of
$({\mathcal J}^{-1}_{\mbox{\tiny{$R$}}}(\mu) /G_\mu) \cap (P/G)$ and
that they sit in $({\mathcal J}^{-1}_{\mbox{\tiny{$R$}}}(\mu)
/G_\mu)$ as symplectic submanifolds. Finally, note that
$$
({\mathcal J}^{-1}_{\mbox{\tiny{$R$}}}(\mu) /G_\mu) \cap (P/G)
\simeq ( {\mathcal J}^{-1}_{\mbox{\tiny{$R$}}}(\mathcal{O}_\mu) \cap
P ) /G = {\mathcal J}^{-1}_{\mbox{\tiny{$P$}}}(\mathcal{O}_\mu) /G =
{\mathcal J}_{\mbox{\tiny{$P$}}}^{-1} (\mu) / G_\mu,
$$
where $\mathcal{O}_\mu$ is the coadjoint orbit through $\mu$.
\end{proof}

Suppose that $F$ is a regular distribution on $P$ such that $TP= F
\oplus \textup{Ker}\, \Omega_P$. The pair $(F,\Omega_P)$ defines a
bivector field $\pi_{\! \mbox{\tiny{$P$}}}$ on $P$, as in
\eqref{Eq:bivector-section}. If $F$ is $G$-invariant, then $\pi_{\!
\mbox{\tiny{$P$}}}$ is also $G$-invariant, and defines a bivector
field on $P/G$ via reduction. The relation between this bivector and
the one in Lemma~\ref{L:A-PresymplecticRed}$(i)$ is as follows.

\begin{lemma} \label{L:B-PresymAlmPoisson}
The reduction of the almost Poisson structure $ \pi_{\!
\mbox{\tiny{$P$}}}$ to $P/G$ coincides with the Poisson structure
$\pi_{\mbox{\tiny{$P/G$}}}$ given by the reduction of the
presymplectic structure $ \Omega_P$ as in
Lemma~\ref{L:A-PresymplecticRed}$(i)$.
\end{lemma}

\begin{proof} We have to
 show that
 $\pi_{\mbox{\tiny{$P/G$}}}^\sharp(\a) =
 T \! \rho_{\! \mbox{\tiny{$P$}}} ( \pi_{\! \mbox{\tiny{$P$}}}^\sharp(\rho_{\! \mbox{\tiny{$P$}}}^*\a) )$
for each 1-form $\a$ on $P /G$.
Note that the vector field $\pi_{\!
\mbox{\tiny{$P$}}}^\sharp(\rho_{\! \mbox{\tiny{$P$}}}^*\a)$
satisfies ${\bf i}_{\pi_{\! \mbox{\tiny{$P$}}}^\sharp(\rho_{\!
\mbox{\tiny{$P$}}}^*\a)} \Omega_P |_F  = - \rho_{\!
\mbox{\tiny{$P$}}} ^*\a |_F$. So there is a 1-form $\Gamma$ on $P$
such that $\Gamma |_F \equiv 0$ and ${\bf i}_{\pi_{\!
\mbox{\tiny{$P$}}}^\sharp(\rho_{\! \mbox{\tiny{$P$}}}^*\a)} \Omega_P
= - \rho_{\! \mbox{\tiny{$P$}}}^*\a + \Gamma$.  Using that
$\textup{Ker}\, \Omega_P \subset \V_{\! \mbox{\tiny{$P$}}}$
we see that ${\bf i}_{\pi_{\!\mbox{\tiny{$P$}}}^\sharp(\rho_{\! \mbox{\tiny{$P$}}}^*\a)} \Omega_P
(X) =  \Gamma(X)$ for all $X \in \textup{Ker}\, \Omega_P $ and thus
$\Gamma \equiv 0$. Hence ${\bf i}_{\pi_{\!
\mbox{\tiny{$P$}}}^\sharp(\rho_{\! \mbox{\tiny{$P$}}}^*\a)} \Omega_P
= - \rho_{\! \mbox{\tiny{$P$}}}^*\a$, which implies that
$\pi_{\mbox{\tiny{$P/G$}}}^\sharp(\a) = T \! \rho_{\!
\mbox{\tiny{$P$}}} (\pi_{\! \mbox{\tiny{$P$}}}^\sharp(\rho_{\!
\mbox{\tiny{$P$}}}^*\a) )$.
\end{proof}

\end{document}